\documentclass[submission]{eptcs-tampered}
\providecommand{\event}{TERMGRAPH 2020} 

\title{Structure-Constrained Process Graphs\\ for the Process Semantics of Regular Expressions}
\author{Clemens Grabmayer
\institute{Gran Sasso Science Institute\\ Viale F.\ Crispi, 7, 67100 L'Aquila AQ, Italy}
\email{clemens.grabmayer@gssi.it}} 

\def\titlerunning{Structure-Constrained Process Graphs for the Process Semantics of Regular Expressions}
\def\authorrunning{C. Grabmayer}

\usepackage{afterpage}
\usepackage{amsfonts,amsmath,amsthm,amscd,amssymb}
\usepackage[american]{babel}
\usepackage{breakurl}         
\usepackage{bussproofs}
\usepackage{calc}
\usepackage{centernot}
\usepackage{chngcntr}
\usepackage{color}
\usepackage[dmyyyy]{datetime}
\usepackage[inline]{enumitem}
\usepackage{esvect}
\usepackage{fancybox}
\usepackage{float}
\usepackage[left=2.5cm,top=3cm,right=2.5cm,bottom=3cm,bindingoffset=0.5cm]{geometry}
\usepackage{graphicx}
\usepackage{multirow}
\usepackage{mathtools}
\usepackage{placeins}
\usepackage{soul}
\usepackage{stmaryrd}           
\usepackage{underscore} 
\usepackage{url}
\usepackage{ushort}
\usepackage{xcolor}

\usepackage{tikz,tikz-cd,fp}
\usetikzlibrary{arrows.meta,arrows,calc,automata,angles,quotes,matrix,
                positioning,shapes.geometric,shapes.symbols,shapes.multipart,through,trees,
                decorations.markings,decorations.text
                }          

\pgfdeclaredecoration{funbisim}{final}{
  \state{final}[width=\pgfdecoratedpathlength]{
    \draw[->]
      (0,\pgfdecorationsegmentamplitude+0.1mm) -- +(\pgfdecoratedpathlength,0);
    \draw[-]
      (0,-\pgfdecorationsegmentamplitude-0.1mm) -- +(\pgfdecoratedpathlength,0);}}
\tikzset{
  funbisim/.style={
    decoration={funbisim, amplitude=0.25ex},
    decorate,
    funbisim options/.style={#1}    
  }}

\pgfdeclaredecoration{bisim}{final}{
  \state{final}[width=\pgfdecoratedpathlength]{
    \draw[<->]
      (0,\pgfdecorationsegmentamplitude+0.1mm) -- +(\pgfdecoratedpathlength,0);
    \draw[-]
      (0,-\pgfdecorationsegmentamplitude-0.1mm) -- +(\pgfdecoratedpathlength,0);}}
\tikzset{
  bisim/.style={
    decoration={bisim, amplitude=0.25ex},
    decorate,
    bisim options/.style={#1}    
  }}

\makeatletter
\def\calcLength(#1,#2)#3{%
\pgfpointdiff{\pgfpointanchor{#1}{center}}%
             {\pgfpointanchor{#2}{center}}%
\pgf@xa=\pgf@x%
\pgf@ya=\pgf@y%
\FPeval\@temp@a{\pgfmath@tonumber{\pgf@xa}}%
\FPeval\@temp@b{\pgfmath@tonumber{\pgf@ya}}%
\FPeval\@temp@sum{(\@temp@a*\@temp@a+\@temp@b*\@temp@b)}%
\FProot{\FPMathLen}{\@temp@sum}{2}%
\FPround\FPMathLen\FPMathLen5\relax
\global\expandafter\edef\csname #3\endcsname{\FPMathLen}
}
\makeatother    

\theoremstyle{plain}
\newtheorem{thm}{Theorem}[section]

\newtheorem{lem}[thm]{Lemma}
\newtheorem*{repeatedthm}{Theorem}
\newtheorem*{repeatedlem}{Lemma}

\newtheorem{prop}[thm]{Proposition}

\theoremstyle{definition}
\newtheorem{defi}[thm]{Definition}

\newtheorem{rem}[thm]{Remark}
\newtheorem{exa}[thm]{Example}

\setul{1pt}{.4pt}  

\makeatletter
 \newcommand{\mynobreakpar}{\par\nobreak\@afterheading}
\makeatother   
\definecolor{azure}{rgb}{0.94,1.00,1.00}
\definecolor{brown}{rgb}{.75,.25,.25}
\definecolor{cyan}{rgb}{0.25,0.88,0.82}
\definecolor{chocolate}{rgb}{0.82,0.41,0.12}
\definecolor{darkcyan}{rgb}{0.5,0,1}
\definecolor{darkgreen}{rgb}{0,0.39,0}
\definecolor{darkmagenta}{rgb}{0.5,0,0.5}
\definecolor{darkgoldenrod}{RGB}{184,134,11}
\definecolor{firebrick}{RGB}{175,25,25}
\definecolor{forestgreen}{rgb}{0.13,0.55,0.13}
\definecolor{goldenrod}{RGB}{218,165,32}
\definecolor{lightcyan}{rgb}{0.88,1.00,1.00}
\definecolor{lightpink}{rgb}{1.00,0.71,0.76}
\definecolor{myyellow}{RGB}{235,235,0}
\definecolor{lightyellow}{rgb}{1.00,1.00,0.88}
\definecolor{lightgoldenrod}{rgb}{0.83,0.97,0.51}
\definecolor{lightgoldenrodyellow}{rgb}{0.98,0.98,0.82}
\definecolor{lightskyblue}{rgb}{0.53,0.81,0.98}
\definecolor{moccasin}{rgb}{1.00,0.89,0.71}
\definecolor{magenta}{rgb}{1,0,1}
\definecolor{navyblue}{rgb}{0,0,0.5}
\definecolor{orange}{rgb}{1.0,0.65,0.0}
\definecolor{orangered}{rgb}{1.0,0.27,0.0}
\definecolor{palegreen}{rgb}{0.60,0.98,0.60}
\definecolor{powderblue}{rgb}{0.69,0.88,0.90}
\definecolor{purple}{rgb}{1,0.5,1}
\definecolor{royalblue}{RGB}{65,105,225}
\definecolor{mediumblue}{RGB}{0,0,205}
\definecolor{cornflowerblue}{RGB}{100,149,237}
\definecolor{springgreen}{rgb}{0.0,1.0,0.5}
\definecolor{turquoise}{rgb}{0.25,0.88,0.82}
\definecolor{snow}{rgb}{1.00,0.98,0.98}
\definecolor{tan}{rgb}{0.82,0.71,0.55}
\definecolor{red}{rgb}{1,0,0}
\definecolor{violetred}{RGB}{208,32,144}

\newcommand{\colorin}[1]{\textcolor{#1}}

\newcommand{\black}{\colorin{black}}

\newcommand{\colorred}{\colorin{red}}
\newcommand{\alert}{\colorred}

\newcommand{\darkcyan}{\colorin{darkcyan}}

\newcommand{\mediumblue}{\colorin{mediumblue}}

\newcommand{\nb}{\nobreakdash}
\newcommand{\punc}[1]{\ensuremath{\hspace*{1.5pt}{#1}}}

\newcommand{\nf}{\normalfont}

\newcommand{\bs}[1]{\boldsymbol{#1}}

\newenvironment{new}{\color{chocolate}}{\color{black}}
\newenvironment{newer}{\color{firebrick}}{\color{black}}
\newenvironment{newest}{\color{red}}{\color{black}}
\newenvironment{revised}{\color{violetred}}{\color{black}}
\newenvironment{change}{\color{violetred}}{\color{black}}

\newcommand{\isnewest}[1]{\begin{newest}{#1}\end{newest}}

\newcommand{\funin}{\mathrel{:}}
\newcommand{\fap}[2]{{#1}(\hspace*{-0.5pt}{#2}\hspace*{-0.5pt})}

\newcommand{\iap}[2]{#1_{#2}}
\newcommand{\bap}[2]{#1_{#2}}
\newcommand{\pap}[2]{#1^{#2}}
\newcommand{\bpap}[3]{#1_{#2}^{#3}}
\newcommand{\pbap}[3]{#1_{#3}^{#2}}
\newcommand{\sproj}{\pi}

\newcommand{\proj}{\fap{\sproj}}

\newcommand{\sdefdby}{{:=}}
\newcommand{\defdby}{\mathrel{\sdefdby}}

\newcommand\tuple[1]{\langle #1 \rangle}

\newcommand\tuplespace{\hspace*{0.5pt}}
\newcommand\pair[2]{\tuple{#1, \tuplespace #2}}
\newcommand\triple[3]{\tuple{#1, \tuplespace #2, \tuplespace #3}}

\newcommand{\nat}{\mathbb{N}}
\newcommand{\natplus}{\pap{\nat}{+}} 

\newcommand{\BNFor}{\:\mid\:}
\newcommand{\BNFdefdby}{\:{::=}\:}

\newcommand{\sred}{{\to}}

\newcommand{\sredi}[1]{{\iap{\sred}{#1}}}
\newcommand{\redi}[1]{\mathrel{\sredi{#1}}}

\newcommand{\sconvredi}[1]{{\iap{\leftarrow}{#1}}}

\newcommand{\sredrtc}{\sred^{*}}
\newcommand{\redrtc}{\mathrel{\sredrtc}}
\newcommand{\sredrtci}[1]{{\iap{\sredrtc}{#1}}}
\newcommand{\redrtci}[1]{\mathrel{\sredrtci{#1}}}

\newcommand{\scomprewrels}[2]{{#1}\cdot{#2}}
\newcommand{\comprewrels}[2]{\mathrel{\scomprewrels{#1}{#2}}}

\newcommand{\stavoidsv}{\text{\nf\bf\st{$\hspace*{1.5pt}$t$\hspace*{1.5pt}$}}}
\newcommand{\tavoidsv}{\fap{\stavoidsv}}
\newcommand{\sredtavoidsv}[1]{{\xrightarrow[\raisebox{0pt}{\scriptsize {$\tavoidsv{#1}$}}]{}}}
\newcommand{\sredtavoidsvi}[2]{{\sredtavoidsv{#1}{}_{#2}}}

\newcommand{\sredtavoidsvrtci}[2]{{\xrightarrow[\raisebox{0pt}{\scriptsize {$\tavoidsv{#1}$}}]{}}{^{*}_{#2}}}

  \newcommand{\specfontsize}{\fontsize{5}{6}\selectfont} 
  \newcommand{\subotr}{\hspace*{-1pt}\mbox{\specfontsize $(\hspace*{-0.6pt}\sone\hspace*{-0.85pt})$}}

\newcommand{\descsetexpmid}{\mathrel{\vert}}

\newcommand{\descsetexp}[2]{\left\{{#1}\descsetexpmid{#2}\right\}}

\newcommand{\safun}{f}

\newcommand{\setexp}[1]{\left\{{#1}\right\}}

\renewcommand{\emptyset}{\varnothing}

\newcommand{\slogand}{\wedge}

\newcommand{\logand}{\mathrel{\slogand}}

\newcommand{\existsstzero}[1]{\exists{\hspace*{1pt}#1}}

\newcommand{\length}[1]{\left|{#1}\right|}

\newcommand{\depth}{\length}

\newcommand{\noninitial}{non-ini\-tial}

\newcommand{\entrytransition}{entry-tran\-si\-tion}
\newcommand{\entrytransitions}{\entrytransition{s}}

\newcommand{\generatedby}[1]{${#1}$\nb-ge\-ne\-ra\-ted}

\newcommand{\loopentry}{loop-en\-try}
\newcommand{\loopbody}{loop-body}

\newcommand{\entrybodylabeling}{en\-try\discretionary{/}{}{/}body-la\-be\-ling}
\newcommand{\entrybodylabelings}{\entrybodylabeling{s}}

\newcommand{\LEEwitness}{$\LEE$\nb-wit\-ness}

\newcommand{\LEEwitnesses}{$\LEE$\hspace*{1.25pt}\nb-wit\-nes\-ses}

\newcommand{\LLEEwitness}{{\nf LLEE}\nb-wit\-ness}
\newcommand{\LLEEwitnesses}{\LLEEwitness{es}}

\newcommand{\nonempty}{non-emp\-ty}

\newcommand{\nontrivial}{non-triv\-i\-al}

\newcommand{\onetransition}{$1$\nb-tran\-si\-tion}
\newcommand{\onetransitions}{\onetransition{s}}

\newcommand{\onefree}{$1$\nb-free}

\newcommand{\onechart}{$1$\nb-chart}

\newcommand{\subonechart}{sub-$1$\nb-chart}

\newcommand{\transitionact}[1]{{${#1}$}\nb-tran\-si\-tion}

\newcommand{\LTS}{LTS}
\newcommand{\LTSs}{LTSs}
\newcommand{\oneLTS}{1\nb-\LTS}
\newcommand{\TSS}{TSS}
\newcommand{\TSSs}{TSSs}

\newcommand{\astexp}{e}
\newcommand{\bstexp}{f}
\newcommand{\cstexp}{g}
\newcommand{\dstexp}{h}
\newcommand{\astexpi}{\iap{\astexp}}
\newcommand{\bstexpi}{\iap{\bstexp}}
\newcommand{\cstexpi}{\iap{\cstexp}}

\newcommand{\astexpacc}{\astexp'}

\newcommand{\cstexpacc}{\cstexp'}
\newcommand{\dstexpacc}{\dstexp'}
\newcommand{\astexpacci}{\iap{\astexpacc}}

\newcommand{\astexptilde}{\tilde{\astexp}}

\newcommand{\cstexptilde}{\tilde{\cstexp}}

\newcommand{\asstexp}{E}
\newcommand{\bsstexp}{F}
\newcommand{\csstexp}{G}

\newcommand{\asstexpi}{\iap{\asstexp}}
\newcommand{\bsstexpi}{\iap{\bsstexp}}
\newcommand{\csstexpi}{\iap{\csstexp}}

\newcommand{\asstexpacc}{\asstexp'}
\newcommand{\bsstexpacc}{\bsstexp'}

\newcommand{\asstexpacci}{\iap{\asstexpacc}}

\newcommand{\asstexptilde}{\widetilde{\asstexp}}
\newcommand{\asstexptildei}{\iap{\asstexptilde}}

\newcommand{\StExp}{\mathit{StExp}}
\newcommand{\StExpover}{\fap{\StExp}}

\newcommand{\stackStExp}{\mathit{StExp}^{{\scriptscriptstyle(}\sstexpstackprod{\scriptscriptstyle )}}}
\newcommand{\stackStExpover}{\fap{\stackStExp\hspace*{-1.5pt}}}

\newcommand{\AppCxt}{\textit{AppCxt}}
\newcommand{\AppCxtover}{\fap{\AppCxt}}

\newcommand{\acxt}{C}

\newcommand{\acxtwh}{\cxtap{\acxt}{\cdot}}

\newcommand{\cxtap}[2]{{#1}[{#2}]}
\newcommand{\acxtap}{\cxtap{\acxt}}

\newcommand{\stexpzero}{0}
\newcommand{\stexpone}{1}

\newcommand{\sstexpit}{\sstar}
\newcommand{\stexpit}[1]{{#1^{\sstexpit}}}

\newcommand{\sstexpprod}{{\cdot}}
\newcommand{\stexpprod}[2]{{#1}\mathrel{\sstexpprod}{#2}}
\newcommand{\sstexpstackprod}{{\sstar}}
\newcommand{\stexpstackprod}[2]{{#1}\mathrel{\sstexpstackprod}{#2}}
\newcommand{\sstexpsum}{+}
\newcommand{\stexpsum}[2]{{#1}\sstexpsum{#2}}

\newcommand{\sstexpbit}{\circledast} 
\newcommand{\stexpbit}[2]{{#1}\hspace*{0.35pt}\pap{}{\sstexpbit}\hspace*{-0.6pt}{#2}}

\newcommand{\sth}[1]{|{#1}|_{\scalebox{0.8}{$\sstar$}}}

\newcommand{\sdescrelstexpit}{\sredi{\scriptscriptstyle(\sstar)}}

\newcommand{\descrelstexpit}[1]{\mathrel{\sdescrelstexpit}}

\newcommand{\sconvdescrelstexpit}{\sconvredi{\scriptscriptstyle(\sstar)}}

\newcommand{\convdescrelstexpit}[1]{\mathrel{\sconvdescrelstexpit}}

\newcommand{\slt}[1]{{\xrightarrow{#1}}}
\newcommand{\slti}[2]{{\xrightarrow{#1}}{_{#2}}}

\newcommand{\slttc}[1]{{\xrightarrow{#1}}{^{+}}}

\newcommand{\lt}[1]{\mathrel{\slt{#1}}}
\newcommand{\lti}[2]{\mathrel{\slti{#1}{#2}}}

\newcommand{\sone}{1}
\newcommand{\sstar}{*}

\newcommand{\bodylab}{\text{\nf bo}}
\newcommand{\bodylabcol}{\text{\nf\darkcyan{bo}}}
\renewcommand{\bodylab}{\bodylabcol}

\newcommand{\looplab}[1]{{\darkcyan{[#1]}}}
\newcommand{\loopnsteplab}[1]{[{#1}]}
\newcommand{\sloopnstepto}[1]{{\iap{\rightarrow}{\loopnsteplab{#1}}}}
\newcommand{\loopnstepto}[1]{\mathrel{\sloopnstepto{#1}}}

\newcommand{\loopentrytransition}{loop-entry tran\-si\-tion}
\newcommand{\loopentrytransitions}{\loopentrytransition{s}}

\newcommand{\aLname}{n}

\newcommand{\bLname}{m}

\newcommand{\txtnormedplus}{normed$^+$}

\newcommand{\sterminates}{{\downarrow}}

\newcommand{\terminates}[1]{{#1}{\sterminates}}
\newcommand{\terminatesi}[2]{{#2}{\iap{\sterminates}{#1}}}
\newcommand{\snotterminates}{\ndownarrow}  
\newcommand{\notterminates}[1]{{#1}{\snotterminates}}

\newcommand{\onebrackscript}{\scalebox{0.75}{$\scriptstyle (1)$}}

\newcommand{\soneterminates}{{\pap{\downarrow}{\hspace*{-1.5pt}\onebrackscript}}}
\newcommand{\oneterminates}[1]{{#1}{\soneterminates}}
\newcommand{\soneterminatesalert}{{\pap{\downarrow}{\hspace*{-1.5pt}\isnewest{\onebrackscript}}}}

\newcommand{\aLTS}{\mathcal{L}}

\newcommand{\aLTShat}{\widehat{\smash{\aLTS}\rule{0pt}{7pt}}} 
\newcommand{\aLTShatsubscript}{\widehat{\smash{\aLTS}\rule{0pt}{5.5pt}}} 
\newcommand{\aLTSi}{\iap{\aLTS}}
\newcommand{\LTSof}{\fap{\aLTS}}

\newcommand{\LTSdefdby}{\iap{\aLTS}}
\newcommand{\aoneLTS}{\underline{\aLTS}}
\newcommand{\aoneLTShat}{\underline{\widehat{\aLTS}}}
\newcommand{\oneLTSof}{\fap{\aoneLTS}}
\newcommand{\oneLTShatof}{\fap{\aoneLTShat}}
\newcommand{\oneLTSdefdby}{\iap{\aoneLTS}}
  \newcommand{\subscriptindtrans}{\mediumblue{\scriptscriptstyle\pmb{\otind{\cdot}}}} 
\newcommand{\indLTS}[1]{\iap{{#1}}{\subscriptindtrans}}
\newcommand{\indLTSof}{\indLTS}

\newcommand{\achart}{\mathcal{C}}
\newcommand{\acharti}{\iap{\achart}}

\newcommand{\acharthat}{\hspace*{0.75pt}\widehat{\hspace*{-0.75pt}\achart}\hspace*{-0pt}} 

\newcommand{\chartnei}{\pbap{\achart}{\text{\nf (ne)}}}

\newcommand{\aonechart}{\underline{\mathcal{C}}}

\newcommand{\chartof}{\fap{\achart}}
\newcommand{\onechartof}{\fap{\aonechart}}

\newcommand{\indscchartof}[1]{\iap{#1}{\subscriptindtrans}} 

\newcommand{\onecharthatof}[1]{\widehat{\rule{0pt}{7pt}\smash{\fap{\aonechart}{{#1}}}}}

\newcommand{\aloop}{\mathcal{L\hspace*{-4pt}C}}

\newcommand{\indsubchartinat}[1]{\fap{\acharti{#1}}}

\newcommand{\otind}[1]{({#1}]}

\newcommand{\iact}[1]{{\scriptscriptstyle\pmb (}\hspace*{-0pt}{#1}\hspace*{0.4pt}{\scriptscriptstyle\pmb ]}}
\newcommand{\iactalert}[1]{\mediumblue{{\scriptscriptstyle\pmb (}\hspace*{-0pt}{\black{#1}}\hspace*{0.4pt}{\scriptscriptstyle\pmb ]}}}

\newcommand{\silt}[1]{\slt{\iact{#1}}}

\newcommand{\silttc}[1]{\slttc{\iact{#1}}} 

\newcommand{\ilt}[1]{\mathrel{\silt{#1}}}

\newcommand{\ilttc}[1]{\mathrel{\silttc{#1}}}
\newcommand{\entrystep}{en\-try-step}
\newcommand{\entrysteps}{\entrystep{s}}

\newcommand{\actions}{\mathit{A}} 
\newcommand{\oneactions}{\underline{\actions}}

\newcommand{\aact}{a}
\newcommand{\bact}{b}
\newcommand{\cact}{c}
\newcommand{\dact}{d}

\newcommand{\aacti}{\iap{\aact}}
\newcommand{\bacti}{\iap{\bact}}

\newcommand{\aoneact}{\underline{a}}

\newcommand{\verts}{V}

\newcommand{\start}{\averti{\hspace*{-0.5pt}\text{\nf s}}}
\newcommand{\transs}{{\sred}} 

\newcommand{\termexts}{\sterminates} 
\newcommand{\states}{S}
\newcommand{\alab}{l}
\renewcommand{\alab}{\darkcyan{l}}

\newcommand{\vertsof}{\fap{\verts\hspace*{-1pt}}}

\newcommand{\genap}{\bap}
\newcommand{\genvertsof}{\genap{\verts}}
\newcommand{\gentranssof}{\genap{\transs}}
\newcommand{\gentermextsof}{\genap{\termexts}}

\newcommand{\vertsi}[1]{\iap{\verts}{\hspace*{-0.25pt}{#1}}}

\newcommand{\starti}[1]{\averti{\text{\nf s},#1}}

\newcommand{\termextsi}[1]{\iap{\termexts}{{#1}}}
\newcommand{\statesi}{\iap{\states}}

\newcommand{\transshat}{\hspace*{-1.5pt}\bs{\Hat{\normalfont \hspace*{1.5pt}\transs}}} 

\newcommand{\astate}{s}
\newcommand{\bstate}{t}

\newcommand{\astatei}{\iap{\astate}}

\newcommand{\astateacc}{\astate'}
\newcommand{\bstateacc}{\bstate'}

\newcommand{\astateacci}{\iap{\astateacc}}

\newcommand{\asettranss}{U}

\newcommand{\sentries}{\mathit{E}}
\newcommand{\entriesof}{\fap{\sentries}}

\newcommand{\avert}{v}
\newcommand{\bvert}{w}

\newcommand{\averti}{\iap{\avert}}

\newcommand{\atrans}{\tau}

\newcommand{\atranshat}{\widehat{\atrans}}

\newcommand{\sfunbisim}{%
    \setbox0=\hbox{\kern-.1ex{$\rightarrow$}\kern-.1ex}
    \setbox1=\vbox{\hbox{\raise .1ex \box0}\hrule}%
    {\hbox{\kern.05ex\box1\kern.1ex}}
  }
\newcommand{\funbisim}{\mathrel{\sfunbisim}}

\newcommand{\sconvfunbisim}{%
    \setbox0=\hbox{\kern-.1ex{$\leftarrow$}\kern-.1ex}
    \setbox1=\vbox{\hbox{\raise .1ex \box0}\hrule}%
    {\hbox{\kern.05ex\box1\kern.1ex}}
  }

\newcommand{\sbisim}{%
    \setbox0=\hbox{\kern-.1ex{$\leftrightarrow$}\kern-.1ex}
    \setbox1=\vbox{\hbox{\raise .1ex \box0}\hrule}%
    \hbox{\kern.1ex\box1\kern.1ex}
  }
\newcommand{\bisim}{\mathrel{\sbisim\hspace*{1pt}}}

\newcommand{\sfunbisimos}{%
    \setbox0=\hbox{\kern-.1ex{$\rightarrow$}\kern-.1ex}
    \setbox1=\vbox{\hbox{\raise .1ex \box0}\hrule}%
    {\pap{\hbox{\kern.05ex\box1\kern.1ex}}{\hspace*{0.5pt}\subotr}}
  }

\newcommand{\sconvfunbisimos}{%
    \setbox0=\hbox{\kern-.1ex{$\leftarrow$}\kern-.1ex}
    \setbox1=\vbox{\hbox{\raise .1ex \box0}\hrule}%
    {\pap{\hbox{\kern.05ex\box1\kern.1ex}}{\hspace*{0.5pt}\subotr}}
  }

\newcommand{\sbisimos}{%
    \setbox0=\hbox{\kern-.1ex{$\leftrightarrow$}\kern-.1ex}
    \setbox1=\vbox{\hbox{\raise .1ex \box0}\hrule}%
    \ensuremath{\pap{\mathrel{\hbox{\kern.1ex\box1\kern.1ex}}}{\hspace*{0.5pt}\subotr}}
  }

\newcommand{\sfunbisimosvia}[1]{%
    \setbox0=\hbox{\kern-.1ex{$\rightarrow$}\kern-.1ex}
    \setbox1=\vbox{\hbox{\raise .1ex \box0}\hrule}%
    {\bpap{\hbox{\kern.05ex\box1\kern.1ex}}{#1}{\hspace*{0.5pt}\subotr}}
  }

\newcommand{\sconvfunbisimosvia}[1]{%
    \setbox0=\hbox{\kern-.1ex{$\leftarrow$}\kern-.1ex}
    \setbox1=\vbox{\hbox{\raise .1ex \box0}\hrule}%
    {\bpap{\hbox{\kern.05ex\box1\kern.1ex}}{#1}{\hspace*{0.5pt}\subotr}}
  }

\newcommand{\abisim}{B}

\newcommand{\sbehinc}{{\sqsubseteq}}

\newcommand{\sbehinca}[1]{{\prescript{#1}{}{\sbehinc}}}
\newcommand{\behinca}[1]{\mathrel{\sbehinca}}

\newcommand{\sonebehinc}{{\pap{\sbehinc}{\subotr}}}

\newcommand{\sonebehinca}[1]{{{}_{#1}\sonebehinc}}
\newcommand{\onebehinca}[1]{\mathrel{\sonebehinca}}
\newcommand{\arule}{R}
\newcommand{\arulei}{\iap{\arule}}
\newcommand{\aruleacc}{\arule'}
\newcommand{\aruleacci}{\iap{\aruleacc}}
\newcommand{\sderivablein}[1]{\vdash_{#1}}
\newcommand{\derivablein}[1]{\sderivablein{#1}}

\newcommand{\saTSS}{{\cal{T}}}
\newcommand{\StExpTSS}{\text{$\saTSS$}}
\newcommand{\StExpTSSover}[1]{\text{$\fap{\StExpTSS}{#1}$}}
\newcommand{\stackStExpTSS}{\text{$\underline{\saTSS}^{{\scriptscriptstyle (}\sstexpstackprod{\scriptscriptstyle )}}$}}
\newcommand{\stackStExpTSSover}[1]{\text{$\fap{\stackStExpTSS\hspace*{-1.5pt}}{#1}$}}
\newcommand{\stackStExpindTSS}{\text{$\underline{\saTSS}_{\hspace*{1pt}\subscriptindtrans}^{{\scriptscriptstyle (}\sstexpstackprod{\scriptscriptstyle )}}$}}
\newcommand{\stackStExpindTSSover}[1]{\text{$\fap{\stackStExpindTSS\hspace*{-1.5pt}}{#1}$}}
\newcommand{\stackStExpTSShat}{\text{$\underline{\widehat{\saTSS}}^{{\scriptscriptstyle (}\sstexpstackprod{\scriptscriptstyle )}}$}}
\newcommand{\stackStExpTSShatover}[1]{\text{$\fap{\stackStExpTSShat\hspace*{-3pt}}{#1}$}}

\newcommand{\aDeriv}{\mathcal{D}}
\newcommand{\aDerivi}{\iap{\mathcal{D}}}
\newcommand{\aDerivacc}{\aDeriv'}
\newcommand{\aDerivacci}{\iap{\aDerivacc}}

\newcommand{\aDerivtildeacc}{\widetilde{\aDeriv}'}
\newcommand{\aDerivtildeacci}{\iap{\aDerivtildeacc}}

\newcommand{\sLEE}{\text{\nf LEE}}

\newcommand{\LEE}{\sLEE}

\newcommand{\thplus}[2]{{#1}{+}{#2}}

\begin{document}
  
\maketitle

\begin{abstract}
  Milner (1984) introduced a process semantics for regular expressions as process graphs. 
  Unlike for the language semantics, where every regular (that is, DFA-accepted) language is the interpretation of some regular expression,
  there are finite process graphs that are not bisimilar to the process interpretation of any regular expression.
  For reasoning about graphs that are expressible by regular expressions 
  it is desirable to have structural representations of process graphs in the image of the interpretation.
  
  For `1-free' regular expressions,
  their process interpretations satisfy the structural property LEE (loop existence and elimination).
  But this is not in general the case for all regular expressions, as we show by examples.
  Yet as a remedy, we describe the possibility to recover the property \LEE\ for a close variant of the process interpretation.  
  For this purpose we refine the process semantics of re\-gu\-lar ex\-pres\-sions to yield process graphs with \onetransitions,
  similar to silent moves \mbox{for finite-state automata}.
  
  This report accompanies 
  the paper with the same title in the post-proceedings of the workshop TERMGRAPH~2020.
  Here we give the proofs of not only one but of both of the two central theorems.
\end{abstract}

%


\section{Introduction}%
  \label{intro}

Milner \cite{miln:1984} (1984) defined a process semantics for regular expressions as process graphs:
the interpretation of $\stexpzero$ is deadlock, of $\stexpone$ is successful termination, letters $a$ are atomic actions,
the operators $\sstexpsum$ and $\sstexpprod$ stand for choice and concatenation of processes,
and (unary) Kleene star $\stexpit{(\cdot)}$ represents iteration with the option to terminate successfully after each pass-through.
In order to disambiguate the use of regular expressions for denoting processes, Milner called them `star expressions' in this context. 
Unlike for the standard language semantics, where every regular language is the interpretation of some regular expression,
there are finite process graphs that are not bisimilar to the process interpretation of any star expression.%
  \footnote{E.g., the process graphs $\chartnei{1}$ and $\chartnei{2}$ in Ex.~\ref{ex:chart:interpretation} on page~\pageref{ex:chart:interpretation} are not expressible
            by a star expression modulo bisimilarity.}
This phenomenon led Milner to the formulation of two natural questions: (R)~the problem of recognizing whether a given 
process graph is bisimilar to one in the image of the process interpretation of a star expression,
and (A)~whether a natural adaptation of Salomaa's complete proof system for language equivalence of regular expressions
is complete for bisimilarity of the process interpretation of star expressions.   
While (R) has been shown to be decidable in principle, 
so far only partial solutions have been obtained~for~(A). 
    
For tackling these problems it is expedient to obtain structural representations of process graphs in the image of the interpretation.
The result of Baeten, Corradini, and myself \cite{baet:corr:grab:2007} that the problem~(R) is decidable in principle 
was based on the concept of `well-behaved (recursive) specifications' that links process graphs with star expressions. 
Recently in \cite{grab:fokk:2020:LICS,grab:fokk:2020:arxiv}, Wan Fokkink and I obtained a partial solution for~(A) in the form of a complete proof system
for `\onefree' star expressions, which do not contain~$\stexpone$, but are formed with binary Kleene star iteration $\stexpbit{(\cdot)}{(\cdot)}$ instead of unary iteration.
For this, we defined the efficiently decidable `loop existence and elimination property (\LEE)' of process graphs
that holds for all process graph interpretations of \onefree\ star expressions, and for their bisimulation collapses.

Unfortunately, the property \LEE\ does not hold for process graph interpretations $\chartof{\astexp}$ of all star expressions~$\astexp$.
However, it is the aim of this article is to describe how \LEE\ can nevertheless be made applicable, 
by stepping over to a variant $\onechartof{\cdot}$ of the process interpretation $\chartof{\cdot}$. 
In Section~\ref{LEE} we explain the loop existence and elimination property \LEE\ for process graphs, 
and we define the concept of a `layered \LEEwitness', for short a `\LLEEwitness' for process graphs.
Hereby Section~\ref{LEE} 
            is an adaptation for star expressions that may contain~$1$
of the motivation of \LLEEwitnesses\ in Section~3 in \cite{grab:fokk:2020:LICS},
which was concerned with the process semantics of `\onefree\ star expressions'.
\LLEEwitnesses\ arise by adding natural-number labels to transitions that are subject to suitable constraints. 
A process graph for which a \LLEEwitness\ exists is `structure constrained', since it satisfies \LEE\ (as guaranteed by the \LLEEwitness)
in contrast with a process graph for which no \LLEEwitness\ exists (which then does not satisy~\LEE). 

In Section~\ref{LLEE:fail} we explain examples that show that \LEE\ does not hold in general
for process interpretations of star expressions from the full class. 
As a remedy, we introduce process graphs with \onetransitions\ (similar to silent moves for finite-state automata).
In Section~\ref{recover:LEE} we define the variant~$\onechartof{\cdot}$ of the process graph semantics $\chartof{\cdot}$ such that $\onechartof{\cdot}$ yields process graphs with \onetransitions.
Furthermore, we formulate and illustrate by examples the following two properties of the variant process graph semantics~$\onechartof{\cdot}$ 
concerning its relation to $\chartof{\cdot}$, and the structure of process graphs that $\onechartof{\cdot}$ defines: 
\begin{enumerate}[label={(P\arabic{*})}]
  \item{}\label{property:1} \emph{$\chartof{\cdot}$ and $\onechartof{\cdot}$ coincide up to bisimilarity} (Theorem~\ref{thm:onechart-int:funbisim:chart-int}):   
    For every star expression~$\astexp$, there is a functional bisimulation from the variant process semantics $\onechartof{\astexp}$ of $\astexp$
    to the process semantics $\chartof{\astexp}$ of~$\astexp$. 
    Hence
    the process interpretation $\chartof{\astexp}$ of a star expression~$\astexp$ and its variant $\onechartof{\astexp}$ 
    are bisimilar.
  
  \item{}\label{property:2} \emph{$\onechartof{\cdot}$ guarantees \LEE{}-structure} (Theorem~\ref{thm:onechart-int:LLEEw}):    
    The variant process semantics $\onechartof{\astexp}$ of a star expression $\astexp$ 
    satisfies the loop existence and elimination property \LEE.
\end{enumerate}
Section~\ref{proofs} is devoted to the proof of these properties of $\onechartof{\astexp}$. 
There we elaborate the proofs of \ref{property:1} and \ref{property:2}.
While the proof of \ref{property:1} is only sketched in the proceedings version \cite{grab:2021:TERMGRAPH-postproceedings},
it is presented here in detail.

We expect that these results
  can be valuable steps towards solving the axiomatization problem~(A), based on the partial solution in \cite{grab:fokk:2020:LICS},
  and perhaps also for finding an efficient decision procedure for the recognition problem~(R). 
Please see the concluding section of the proceedings version \cite{grab:2021:TERMGRAPH-postproceedings}
for some explanation and motivation about why we think that these hopes may be justified.

The idea to define structure-constrained process graphs via edge-labelings with constraints, on which \LLEEwitnesses\ are based,
originated from `higher-order term graphs' that can be used for representing functional programs in a maximally compact, shared form 
(see \cite{grab:roch:2014,grab:2019}). There, additional concepts (scope sets of vertices, or abstraction-prefix labelings) 
are used to constrain the form of term graphs. 
The common underlying idea with \LLEEwitnesses\ is an enrichment of graphs that:
(i)~guarantees that graphs can be directly expressed by terms of some language, 
(ii)~does not significantly hamper sharing of represented subterms,
(iii)~is simple enough so as to keep reasoning about graph transformations feasible.

\section{Preliminaries on
         the process semantics of star expressions}%
  \label{prelims}         

In this section we define the process semantics of regular expressions 
as charts: finite labeled transition systems with initial states. 
We proceed by a sequence of definitions, and conclude by providing examples.

\begin{defi}\nf\label{def:StExp}
  We assume, for subsequent definitions implicitly, a set $\actions$ whose members we call \emph{actions}.
  The set $\StExpover{\actions}$ of \emph{star expressions over (actions in) $\actions$} is defined by the following grammar:
  \begin{center}
    $
    \astexp, \astexpi{1}, \astexpi{2}
      \;\;\BNFdefdby\;\;
        \stexpzero
          \BNFor
        \stexpone
          \BNFor
        \aact
          \BNFor
        \stexpsum{\astexpi{1}}{\astexpi{2}}
          \BNFor
        \stexpprod{\astexpi{1}}{\astexpi{2}}
          \BNFor
        \stexpit{\astexp} 
          \qquad\text{(where $\aact\in\actions$)} \punc{.}
    $
  \end{center}
  The \emph{(syntactic) star height} $\sth{\astexp}$ of a star expression $\astexp\in\StExpover{\actions}$
  denotes the maximal nesting depth of stars in $\astexp$
  via: $\sth{\stexpzero} \defdby \sth{\stexpone}  \defdby \sth{\aact} \defdby 0$, 
       $\sth{\stexpsum{\astexpi{1}}{\astexpi{2}}} \defdby \sth{\stexpprod{\astexpi{1}}{\astexpi{2}}}
                                                  \defdby \max\setexp{\sth{\astexpi{1}}, \sth{\astexpi{2}}}$, 
       and $\sth{\stexpit{\astexp}} \defdby 1 + \sth{\astexp}$. 
\end{defi}

\begin{defi}
  A \emph{labeled transition system (LTS) with termination and actions in $\actions$}
  is a 4\nb-tuple
  $\tuple{\states,\actions,\sred,\sterminates}$
  where $\states$ is a \nonempty\ set of \emph{states},
  $\actions$ is a set of $\emph{action labels}$,
  $\sred \subseteq \states\times\actions\times\states$ is the \emph{labeled transition relation},
  and $\sterminates \subseteq \verts$ is a set of \emph{states with immediate termination},
  for~short, the \emph{terminating~states}.
  In such an \LTS, we write $\astatei{1} \lt{\aact} \astatei{2}$ for a transition $\triple{\astatei{1}}{\aact}{\astatei{2}}\in\sred$,
  and we write $\terminates{\astate}$ for a terminating state $\astate\in\termexts$. 
  %
\end{defi}

\begin{defi}\nf\label{def:process:semantics}\enlargethispage{1ex}
  The transition system specification (TSS)~$\StExpTSSover{\actions}$ is defined by the axioms and rules:   
  \begin{center}\label{StExpTSS}
    $
    \begin{gathered}
      \begin{aligned}
         &
         \AxiomC{\phantom{$\terminates{\stexpone}$}}
         \UnaryInfC{$\terminates{\stexpone}$}
         \DisplayProof
         & \hspace*{-1.5ex} & & &
         \AxiomC{$ \terminates{\astexpi{1}} $}
         \UnaryInfC{$ \terminates{(\stexpsum{\astexpi{1}}{\astexpi{2}})} $}
         & \hspace*{2ex} & & & 
         \AxiomC{$ \terminates{\astexpi{i}} $}
         \RightLabel{\scriptsize $(i\in\setexp{1,2})$}
         \UnaryInfC{$ \terminates{(\stexpsum{\astexpi{1}}{\astexpi{2}})} $}
         \DisplayProof
         & \hspace*{2ex} & & & 
         \AxiomC{$\terminates{\astexpi{1}}$}
         \AxiomC{$\terminates{\astexpi{2}}$}
         \BinaryInfC{$\terminates{(\stexpprod{\astexpi{1}}{\astexpi{2}})}$}
         \DisplayProof
         & \hspace*{2ex} & & & 
         \AxiomC{$\phantom{\terminates{\stexpit{\astexp}}}$}
         \UnaryInfC{$\terminates{(\stexpit{\astexp})}$}
         \DisplayProof
      \end{aligned} 
      \\[1ex]
      \begin{aligned}
        & 
        \AxiomC{$\phantom{a_i \:\lt{a_i}\: \stexpone}$}
        \UnaryInfC{$a \:\lt{a}\: \stexpone$}
        \DisplayProof
        & & & &
        \AxiomC{$ \astexpi{i} \:\lt{a}\: \astexpacci{i} $}
        \RightLabel{\scriptsize $(i\in\setexp{1,2})$}
        \UnaryInfC{$ \stexpsum{\astexpi{1}}{\astexpi{2}} \:\lt{a}\: \astexpacci{i} $}
        \DisplayProof 
        & & & & 
        \AxiomC{$ \astexpi{1} \:\lt{a}\: \astexpacci{1} $}
        \UnaryInfC{$ \stexpprod{\astexpi{1}}{\astexpi{2}} \:\lt{a}\: \stexpprod{\astexpacci{1}}{\astexpi{2}} $}
        \DisplayProof
        & &
        \AxiomC{$\terminates{\astexpi{1}}$}
        \AxiomC{$ \astexpi{2} \:\lt{a}\: \astexpacci{2} $}
        \BinaryInfC{$ \stexpprod{\astexpi{1}}{\astexpi{2}} \:\lt{a}\: \astexpacci{2} $}
        \DisplayProof
        & & & & 
        \AxiomC{$\astexp \:\lt{a}\: \astexpacc$}
        \UnaryInfC{$\stexpit{\astexp} \:\lt{a}\: \stexpprod{\astexpacc}{\stexpit{\astexp}}$}
        \DisplayProof
      \end{aligned}
    \end{gathered}
    $
  \end{center}
  If $\astexp \lt{\aact} \astexpacc$ is derivable in $\StExpTSSover{\actions}$, for $\astexp,\astexpacc\in\StExpover{\actions}$,
  and $\aact\in\actions$, then we say that $\astexpacc$ is a \emph{derivative} of $\astexp$.
  If $\terminates{\astexp}$ is derivable in $\StExpTSSover{\actions}$, for $\astexp\in\StExpover{\actions}$,
  then we say that $\astexp$ \emph{permits immediate termination}.
  If $\terminates{\astexp}$ is not derivable in $\StExpTSSover{\actions}$, then we write $\notterminates{\astexp}$.
  
  The \TSS\ $\StExpTSSover{\actions}$ defines the process semantics for star expressions in $\StExpover{\actions}$
  in the form of the \emph{star expressions \LTS}~$\LTSof{\StExpover{\actions}} \defdby \LTSdefdby{\StExpTSSover{\actions}}$,
  where $\LTSdefdby{\StExpTSSover{\actions}} = \tuple{\StExpover{\actions},\actions,\transs,\termexts}$ 
  is the \emph{\LTS\ generated by} $\StExpTSSover{\actions}$, that is,
  its set $\transs \subseteq  \StExpover{\actions}\times\actions\times\StExpover{\actions}$ of transitions,
  and its set $\termexts\subseteq\StExpover{\actions}$ of vertices with the im\-me\-di\-ate-ter\-mi\-na\-tion property
  are defined via derivations in~$\StExpTSSover{\actions}$ in the natural way.
  
  For every set $S \subseteq \StExpover{\actions}$ we denote by $\LTSof{S}$ 
  the \emph{$S$\nb-generated sub-LTS} $\tuple{\genvertsof{S},\actions,\gentranssof{S},\gentermextsof{S}}$ of the star expressions \LTS\ $\LTSof{\StExpover{\actions}}$, 
  that is, the sub-LTS whose states are those in $S$ together with all star expressions that are reachable
  from ones in $S$ via paths that follow transitions of $\LTSof{\StExpover{\actions}}$.
\end{defi}

\begin{defi}
  A \emph{chart} is a 5-tuple $\achart = \tuple{\verts,\actions,\start,\transs,\termexts}$
  such that $\tuple{\verts,\actions,\transs,\termexts}$ is an \LTS, which we call the \emph{LTS underlying $\achart$},
  with a \underline{\smash{finite}} set $\verts$ of states, which we call \emph{vertices}, and 
  that is rooted by a specified \emph{start vertex} $\start\in\verts$;
  we call $\transs \subseteq \verts\times\actions\times\verts$ the set of \emph{labeled transitions} of $\achart$,
  and $\termexts \subseteq \verts$ the set the vertices \emph{with immediate termination} of $\achart$.
  %
\end{defi}

\begin{defi}\nf\label{def:chart:interpretation}
  The \emph{chart interpretation $\chartof{\astexp} = \tuple{\vertsof{\astexp},\actions,\astexp,\transs,\termexts}$}  
  of a star expression $\astexp\in\StExpover{\actions}$ is 
  the $\setexp{\astexp}$\nb-gen\-er\-ated sub-LTS $\LTSof{\setexp{\astexp}} = \tuple{\genvertsof{\setexp{\astexp}},\actions,\gentranssof{\setexp{\astexp}},\gentermextsof{\setexp{\astexp}}}$
  of $\LTSof{\StExpover{\actions}}$.%
    \footnote{That this generated sub-LTS is finite, and so that $\chartof{\astexp}$ is indeed a chart,
      follows from Antimirov's result in \cite{anti:1996} that every regular expression has only finitely many iterated `partial derivatives',
      which coincide with repeated derivatives from Def.~\ref{def:process:semantics}.}
\end{defi}

\newcommand{\picarrowstart}{\raisebox{2pt}{\begin{tikzpicture}%
                                             \draw[<-,very thick,>=latex,chocolate,shorten <=2pt](0,0) -- ++ (180:{12pt});%
                                           \end{tikzpicture}}}
  
\newcommand{\pictermvert}{\begin{tikzpicture}%
                           \node[draw,chocolate,very thick,circle,minimum width=2.5pt,fill,inner sep=0pt,outer sep=2pt](v){};%
                           \draw[thick,chocolate] (v) circle (0.12cm);%
                         \end{tikzpicture}}

\begin{exa}\label{ex:chart:interpretation}\nf
  In the chart illustrations below and later, we indicate 
  the start vertex by a brown arrow~\picarrowstart,
  and the property of a vertex $\avert$ to permit immediate termination
  by emphasizing $\avert$ in brown as \pictermvert\ including a boldface ring. 
  Each of the vertices of the charts $\chartnei{1}$ and $\chartof{\astexp}$ below permits immediate termination, 
  but none of the vertices of the other charts does.
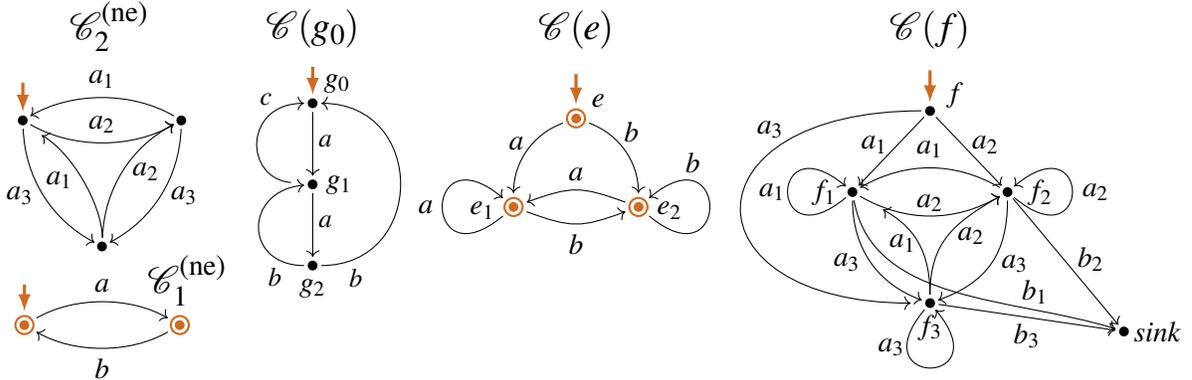
\vspace{0ex}
  \begin{center}
    \begin{tikzpicture}

\matrix[anchor=north,row sep=0.9cm,every node/.style={draw,very thick,circle,minimum width=2.5pt,fill,inner sep=0pt,outer sep=2pt}] at (2.8,0.2) {
  \node(v-0){};
  \\
  \node(v-1){};
  \\
  \node(v-2){};
  \\
};
\calcLength(v-0,v-1){mylen}
\draw[<-,very thick,>=latex,chocolate](v-0) -- ++ (90:{0.45*\mylen pt});
\path(v-0) ++ (0cm,1cm) node{\Large $\chartof{\cstexpi{0}}$};
\path(v-0) ++ ({0.3*\mylen pt},{0.25*\mylen pt}) node{$\cstexpi{0}$};
\draw[->](v-0) to node[right,xshift={-0.05*\mylen pt},pos=0.45]{\small $\aact$} (v-1); 
\path(v-1) ++ ({0.325*\mylen pt},0cm) node{$\cstexpi{1}$};
\draw[->](v-1) to node[right,xshift={-0.05*\mylen pt},pos=0.45]{\small $\aact$} (v-2);
\draw[->,shorten <= 5pt](v-1) to[out=175,in=180,distance={0.75*\mylen pt}]
         node[above,yshift={0.05*\mylen pt},pos=0.7]{\small $\cact$} (v-0);
\path(v-2) ++ (-0cm,{-0.275*\mylen pt}) node{$\cstexpi{2}$};
\draw[->](v-2) to[out=180,in=185,distance={0.75*\mylen pt}] 
               node[below,yshift={0.0*\mylen pt},pos=0.2]{\small $\bact$} (v-1);
\draw[->](v-2) to[out=0,in=0,distance={1.3*\mylen pt}]  
               node[below,yshift={0.00*\mylen pt},pos=0.125]{\small $\bact$} (v-0);

%
\matrix[anchor=north,row sep=0.75cm,column sep=0.924cm,
        every node/.style={draw,very thick,circle,minimum width=2.5pt,fill,inner sep=0pt,outer sep=2pt}] at (0,-0.025) {
  \node(C-2-1){};  &                  &     \node(C-2-2){};
  \\
                   &                  &                  
  \\
                   & \node(C-2-3){};  &
  \\
};
\draw[<-,very thick,>=latex,color=chocolate](C-2-1) -- ++ (90:0.5cm);  

\draw[->,bend right,distance=0.65cm] (C-2-1) to node[above]{$\aacti{2}$} (C-2-2); 
\draw[->,bend right,distance=0.65cm] (C-2-1) to node[left]{$\aacti{3}$}  (C-2-3);

\path(C-2-1) ++ (1.15cm,1.25cm) node{\Large $\chartnei{2}$};
\draw[->,bend right,distance=0.65cm]  (C-2-2) to node[above]{$\aacti{1}$} (C-2-1); 
\draw[->,bend left,distance=0.65cm]  (C-2-2) to node[right]{$\aacti{3}$} (C-2-3);

\draw[->,bend right,distance=0.45cm] (C-2-3) to node[left]{$\aacti{1}$}  ($(C-2-1)+(0.25cm,-0.2cm)$);
\draw[->,bend left,distance=0.65cm]  (C-2-3) to node[right]{$\aacti{2}$} (C-2-2);

%
\matrix[anchor=north,row sep=0.8cm,column sep=0.924cm,
        every node/.style={draw,very thick,circle,minimum width=2.5pt,fill,inner sep=0pt,outer sep=2pt}] at (0,-2.75) {
  \node[color=chocolate](C-1-0){};  &                  &     \node[color=chocolate](C-1-1){};
  \\
                   &                  &                  
  \\
                   & \node[draw=none,fill=none](C-1-2){};  &
  \\
};
\draw[<-,very thick,>=latex,chocolate,shorten <=2pt](C-1-0) -- ++ (90:{0.5*\mylen pt});
%
\path(C-1-1) ++ (0.1cm,0.6cm) node{\Large $\chartnei{1}$};


\draw[thick,chocolate] (C-1-1) circle (0.12cm);
\draw[thick,chocolate] (C-1-0) circle (0.12cm);
\draw[->,bend left,distance=0.65cm,shorten <=2pt,shorten >=2pt] (C-1-0) to node[above]{$a$} (C-1-1); 
\draw[->,bend left,distance=0.65cm,shorten <=2pt,shorten >=2pt] (C-1-1) to node[below]{$b$} (C-1-0); 


\matrix[anchor=north,row sep=1cm,column sep=0.7cm,every node/.style={draw,very thick,circle,minimum width=2.5pt,fill,inner sep=0pt,outer sep=2pt}] at (6.3,0) {
                 &  \node[color=chocolate](e){};
  \\
  \node[color=chocolate](e-1){};  &  \node[draw=none,fill=none](dummy){};  
                                   & \node[color=chocolate](e-2){}; 
  \\
};
\calcLength(e,dummy){mylen}
\draw[<-,very thick,>=latex,chocolate,shorten <=2pt](e) -- ++ (90:{0.5*\mylen pt});
\path(e) ++ ({0.25*\mylen pt},{0.2*\mylen pt}) node{$\astexp$};  
\path(e) ++ ({0*\mylen pt},{1*\mylen pt}) node{\Large $\chartof{\astexp}$};

\draw[chocolate,thick] (e) circle (0.12cm);
\draw[->,shorten <=2pt,shorten >=2pt,out=200,in=90] (e) to node[left,pos=0.4]{$\aact$} (e-1);
\draw[->,shorten <=2pt,shorten >=2pt,out=-20,in=90] (e) to node[right,pos=0.225,xshift={0.1*\mylen pt}]{$\bact$} (e-2);

\draw[chocolate,thick] (e-1) circle (0.12cm);
\path(e-1) ++ ({-0.35*\mylen pt},{-0.035*\mylen pt}) node{$\astexpi{1}$}; 
\draw[->,shorten <=2pt,shorten >=2pt,out=220,in=140,distance={1.25*\mylen pt}] (e-1) to node[left]{$\aact$} (e-1);
\draw[->,shorten <=2pt,shorten >=2pt,out=-20,in=200,distance={0.6*\mylen pt}] (e-1) to node[below]{$\bact$} (e-2);

\draw[thick,chocolate] (e-2) circle (0.12cm);
\path(e-2) ++ ({0.35*\mylen pt},{-0.035*\mylen pt}) node{$\astexpi{2}$};  
\draw[->,shorten <=2pt,shorten >=2pt,out=-40,in=40,distance={1.25*\mylen pt}] (e-2) to node[above,yshift={0*\mylen pt},pos=0.75]{$\bact$} (e-2);
\draw[->,shorten <=2pt,shorten >=2pt,out=160,in=20,distance={0.6*\mylen pt}] (e-2) to node[above]{$\aact$} (e-1);

\matrix[anchor=north,row sep=1cm,column sep=0.9cm,every node/.style={draw,very thick,circle,minimum width=2.5pt,fill,inner sep=0pt,outer sep=2pt}] at (11.5,0.1) {
                 &  \node(f){};
  \\[-0.1cm]
  \node(f-1){};  &  \node[draw=none,fill=none](dummy){};  
                                   & \node(f-2){}; 
  \\[0.3cm]
                 &  \node(f-3){};  &                      & \node[draw=none,fill=none](helper){};
  \\
};
\calcLength(f,dummy){mylen}
\path (helper) ++ ({0.5*\mylen pt},{-0.35*\mylen pt}) node[draw,very thick,circle,minimum width=2.5pt,fill,inner sep=0pt,outer sep=2pt](sink){}; 
\draw[<-,very thick,>=latex,chocolate](f) -- ++ (90:{0.5*\mylen pt});
\path(f) ++ ({0.3*\mylen pt},{0.2*\mylen pt}) node{$\bstexp$};  
\path(f) ++ ({0*\mylen pt},{1*\mylen pt}) node{\Large $\chartof{\bstexp}$};

\draw[->
         ] (f) to node[left,pos=0.4]{$\aacti{1}$} (f-1);
\draw[->
        ] (f) to node[right,pos=0.4]{$\aacti{2}$} (f-2);
\draw[->,out=180,in=180,distance={2.95*\mylen pt},shorten >=4pt] (f) to node[above,pos=0.3,yshift={0.05*\mylen pt}]{$\aacti{3}$} (f-3);

\path(f-1) ++ ({-0.35*\mylen pt},{-0.01*\mylen pt}) node{$\bstexpi{1}$};  
\draw[->,out=270,in=160] (f-1) to node[left]{$\aacti{3}$} (f-3);
\draw[->,out=220,in=140,distance={1.25*\mylen pt}]  (f-1) to node[left,xshift={0.1*\mylen pt}]{$\aacti{1}$} (f-1);
\draw[->,out=-30,in=210,distance={0.65*\mylen pt}]  (f-1) to node[above,yshift={-0.065*\mylen pt}]{$\aacti{2}$} (f-2);
\draw[->,out=-80,in=162.5,distance={1.3*\mylen pt}] (f-1) to node[above,pos=0.75,yshift={-0.05*\mylen pt},xshift={0*\mylen pt}]{$\bacti{1}$} (sink);

\path(f-2) ++ ({0.4*\mylen pt},{-0.01*\mylen pt}) node{$\bstexpi{2}$};  
\draw[->,out=270,in=20,distance={0.65*\mylen pt}] (f-2) to node[right]{$\aacti{3}$} (f-3);
\draw[->,out=-40,in=40,distance={1.25*\mylen pt}] (f-2) to node[right]{$\aacti{2}$} (f-2);
\draw[->,out=150,in=30,distance={0.65*\mylen pt}] (f-2) to node[above]{$\aacti{1}$} (f-1);
\draw[->] (f-2) to node[above,pos=0.7,yshift={0.075*\mylen pt},xshift={0.05*\mylen pt}]{$\bacti{2}$} (sink);

\path(f-3) ++ ({0*\mylen pt},{-0.3*\mylen pt}) node{$\bstexpi{3}$};  
\draw[->,out=230,in=310,distance={1.25*\mylen pt}] (f-3) to node[left,pos=0.2,xshift={0.1*\mylen pt}]{$\aacti{3}$} (f-3);
\draw[->,out=90,in=-30,distance={0.65*\mylen pt},shorten >=10pt] (f-3) to node[pos=0.42,left,xshift={0.11*\mylen pt}]{$\aacti{1}$} (f-1);
\draw[->,out=90,in=210,distance={0.65*\mylen pt}] (f-3) to node[right,pos=0.5]{$\aacti{2}$} (f-2);
\draw[->] (f-3) to node[below,pos=0.5,yshift={0.05*\mylen pt}]{$\bacti{3}$} (sink);

\path(sink) ++ ({0*\mylen pt},{0*\mylen pt}) node[right]{$\textit{sink}$};

\end{tikzpicture}
  
  \end{center}\vspace{-7ex}
  The charts $\chartnei{1}$ and $\chartnei{2}$ are not expressible by the process interpretation modulo bisimilarity,
  as shown by Bosscher \cite{boss:1997} 
           and Milner \cite{miln:1984}. \pagebreak 
  That $\chartnei{2}$ is not expressible, Milner proved by observing the absence of a `loop behaviour'.
  That concept has inspired the stronger concept of `loop chart' in Def.~\ref{def:loop:chart} below.
  For the weaker result that $\chartnei{1}$ and $\chartnei{2}$ are not expressible by \onefree\ star expressions
  please see Remark~\ref{rem:expressible}.
  
  The chart $\chartof{\cstexpi{0}}$ on the left above is the interpretation of the 
  star expression $\cstexpi{0} = \stexpprod{(\stexpprod{(\stexpprod{\stexpone}{\aact})}{\cstexp})}{\stexpzero}$
  where $\cstexp = \stexpit{(\stexpsum{\stexpprod{\cact}{\aact}}{\stexpprod{\aact}{(\stexpsum{\bact}{\stexpprod{\bact}{\aact}})}})}$,
  and with $\cstexpi{1} = \stexpprod{(\stexpprod{\stexpone}{\cstexp})}{\stexpzero}$,
  and $\cstexpi{2} = \stexpprod{(\stexpprod{(\stexpprod{\stexpone}{(\stexpsum{\bact}{\stexpprod{\bact}{\aact}})})}{\cstexp})}{\stexpzero}$
  as remaining vertices.
  The chart $\chartof{\astexp}$ is the interpretation
  of 
  $\astexp =  
     \stexpit{\stexpprod{(\stexpit{\aact}}{\stexpit{\bact}})}$
  with     
  $\astexpi{1} = \stexpprod{(\stexpprod{(\stexpprod{\stexpone}{\stexpit{\aact}})}{\stexpit{\bact}})}{\astexp}$,
  and $\astexpi{2} = \stexpprod{(\stexpprod{\stexpone}{\stexpit{\bact}})}{\astexp}$.
  With 
  $\bstexpi{0} = 
        \stexpsum{\stexpprod{\aacti{1}}{(\stexpsum{\stexpone}{\stexpprod{\bacti{1}}{\stexpzero}})}}
                            {\stexpsum{{\stexpprod{\aacti{2}}{(\stexpsum{\stexpone}{\stexpprod{\bacti{2}}{\stexpzero}})}}}
                                      {{\stexpprod{\aacti{3}}{(\stexpsum{\stexpone}{\stexpprod{\bacti{3}}{\stexpzero}})}}}}$,
  the chart $\chartof{\bstexp}$ is the interpretation of
  $\bstexp = \stexpprod{\stexpit{\bstexpi{0}}}{\stexpzero}$ 
  with
  $\bstexpi{i} = \stexpprod{(\stexpprod{\stexpprod{\stexpone}{(\stexpsum{\stexpone}{\stexpprod{\bacti{i}}{\stexpzero}})}}
                                       {\stexpit{\bstexpi{0}}})}
                          {\stexpzero}$ (for $i\in\setexp{1,2,3}$),
  and $\textit{sink} = \stexpprod{(\stexpprod{(\stexpprod{\stexpone}{\stexpzero})}{\stexpit{\bstexpi{0}}})}
                                 {\stexpzero}$.
  The chart interpretations $\chartof{\astexp}$ and $\chartof{\bstexp}$, which will be used later,
  have been constructed as expressible variants of the not expressible charts $\chartnei{1}$ and $\chartnei{2}\,$.
  In particular, $\chartof{\astexp}$ contains $\chartnei{1}$ as a subchart\vspace*{-0.1ex},
  and $\chartof{\bstexp}$ contains $\chartnei{2}$ as a subchart
  (a `subchart' arises by taking a part of a chart, and picking~a~start~vertex).
  %
\end{exa}


\begin{defi}
            \label{def:bisims:LTSs}
  For $i\in\setexp{1,2}$ we consider the \LTSs\
  $\aLTSi{i} = \tuple{\statesi{i},\actions,\sredi{i},\termextsi{i}}$.
  By a \emph{bisimulation between $\aLTSi{1}$ and $\aLTSi{2}$}
  we mean a binary relation $\abisim \subseteq \statesi{1}\times\statesi{2}$ 
  with the properties that it is \nonempty, that is, $\abisim\neq\emptyset$, 
  and that for every $\pair{\astatei{1}}{\astatei{2}}\in\abisim$ the following three conditions hold:
  \begin{itemize}[labelindent=5.5em,leftmargin=*,itemsep=0.25ex,labelsep=1em]
    \item[(forth)]
      $ \forall \astateacci{1}\in\statesi{1}
          \forall \aact\in\actions
              \bigl(\,
                \astatei{1} \lti{\aact}{1} \astateacci{1}
                  \;\;\Longrightarrow\;\;
                    \exists \astateacci{2}\in\statesi{2}
                      \bigl(\, \astatei{2} \lti{\aact}{2} \astateacci{2} 
                                 \logand
                               \pair{\astateacci{1}}{\astateacci{2}}\in\abisim \,)
            \,\bigr) \punc{,} $
      
    \item[(back)]
      $ \forall \astateacci{2}\in\statesi{2}
          \forall \aact\in\actions
            \bigr(\,
              \bigr(\, 
                \exists \astateacci{1}\in\statesi{1}
                  \bigl(\, \astatei{1} \lti{\aact}{1} \astateacci{1} 
                             \logand
                           \pair{\astateacci{1}}{\astateacci{2}}\in\abisim \,)
              \,\bigr)           
                  \;\;\Longleftarrow\;\;
                \astatei{2} \lti{\aact}{2} \astateacci{2}
              \,\bigr) \punc{,} $
      
    \item[(termination)]
      $ \terminatesi{1}{\astatei{1}}
          \;\;\Longleftrightarrow\;\;
            \terminatesi{2}{\astatei{2}} \punc{.}$
  \end{itemize}
  
  For a partial function $\safun \funin \statesi{1} \rightharpoonup \statesi{2}$
  we say that $\safun$ \emph{defines} a bisimulation between $\aLTSi{1}$ and $\aLTSi{2}$
  if its graph $\descsetexp{ \pair{\avert}{\fap{f}{\avert}} }{\avert\in\statesi{1}}$ is a bisimulation between $\aLTSi{1}$ and $\aLTSi{2}$.
  We call this graph a \emph{functional}~bisimulation. 
\end{defi}

\begin{defi}[bisimulation between charts]\label{def:bisims:charts}
  For $i\in\setexp{1,2}$ we consider the charts
  $\acharti{i} = \tuple{\vertsi{i},\actions,\starti{i},\sredi{i},\termextsi{i}}$.
  
  By a \emph{bisimulation between $\acharti{1}$ and $\acharti{2}$}
  we mean a binary relation $\abisim \subseteq \vertsi{1}\times\vertsi{2}$ 
  such that $\pair{\starti{1}}{\starti{2}}\in\abisim$ ($\abisim$ relates the start vertices of $\acharti{1}$ and $\acharti{2}$),
  and $\abisim$ is a bisimulation between the underlying \LTSs. 
  
  We denote by $\acharti{1} \bisim \acharti{2}$, and say that $\acharti{1}$ and $\acharti{2}$ are \emph{bisimilar},
  if there is a bisimulation between $\acharti{1}$ and $\acharti{2}$.
  By $\acharti{1} \funbisim \acharti{2}$ we denote the stronger statement 
  that there is a partial function $f \funin \vertsi{1} \rightharpoonup \vertsi{2}$ whose graph $\descsetexp{ \pair{\avert}{\fap{f}{\avert}} }{\avert\in\vertsi{1}}$  
  is a bisimulation between~$\acharti{1}$~and~$\acharti{2}$. 
\end{defi}

Each of four charts $\chartnei{1}$, $\chartnei{2}$, $\chartof{\cstexpi{0}}$, and $\chartof{\bstexp}$  
  in Ex.~\ref{ex:chart:interpretation} is a `bisimulation collapse':
by that we mean a chart for which no two subcharts that are induced by different vertices are bisimilar. 
But that does not hold for $\chartof{\astexp}$ in which any two subcharts that are induced by different vertices are bisimilar.

\section{Loop existence and elimination}%
  \label{LEE}

The chart translation $\chartof{\cstexpi{0}}$ of $\cstexpi{0}$ as in Ex.~\ref{ex:chart:interpretation}
satisfies the `loop existence and elimination' property \LEE\ that we will explain in this section. 
For this purpose we summarize Section~3 in \cite{grab:fokk:2020:LICS}, 
  and in doing so we adapt the concepts defined there from \onefree\ star expressions
  to the full class of star expressions as defined in Section~\ref{prelims}.
The property \LEE\ is defined by a dynamic elimination procedure that analyses the structure of a chart
by peeling off `loop sub\-charts'. Such subcharts capture,
within the chart interpretation of a star expression~$\astexp$,
the behavior of the iteration of $\bstexp$ within innermost subterms $\stexpit{\bstexp}$ in~$\astexp$.

\begin{defi}\label{def:loop:chart}\nf
  A chart~$\aloop = \tuple{\verts,\actions,\start,\transs,\termexts}$ is called a \emph{loop chart} if:\vspace*{-0.25ex}
  \begin{enumerate}[label={{\rm (L\arabic*)}},leftmargin=*,align=left,itemsep=0ex]
    \item{}\label{loop:1}
      There is an infinite path from the start vertex $\start$.
    \item{}\label{loop:2}  
      Every infinite path from $\start$ returns to $\start$ after a positive number of transitions
      (and so visits $\start$ infinitely often).
    \item{}\label{loop:3}
      Immediate termination is only permitted at the start vertex, that is, $\termexts\subseteq\setexp{\start}$.
  \end{enumerate}\vspace*{-0.25ex}
  We call the transitions from $\start$ \emph{\loopentry\ transitions},
  and all other transitions \emph{\loopbody\ transitions}.
  A loop chart~$\aloop$ is a \emph{loop subchart of} a chart $\achart$
  if it is the subchart of $\achart$ rooted at some vertex $\avert\in\verts$ 
           that is generated, for a nonempty set $\asettranss$ of transitions of $\achart$ from $\avert$,
           by all paths that start with a transition in $\asettranss$ and continue onward until $\avert$ is reached again
  (so the transitions in $\asettranss$ are the \loopentrytransitions~of~$\aloop$).
\end{defi}

Both of the not expressible charts $\chartnei{1}$ and $\chartnei{2}$ in Ex.~\ref{ex:chart:interpretation} are not loop charts:
$\chartnei{1}$ violates \ref{loop:3}, and $\chartnei{2}$ violates \ref{loop:2}.
Moreover, none of these charts contains a loop subchart.
The chart $\chartof{\cstexpi{0}}$ in Ex.~\ref{ex:chart:interpretation} is not
a loop chart either, as it violates \ref{loop:2}. But we will see that $\chartof{\cstexpi{0}}$ has loop subcharts. 

Let $\aloop$ be a loop subchart of a chart~$\achart$.
The result of \emph{eliminating $\aloop$ from $\achart$}
arises by removing all \loopentrytransitions\ of $\aloop$ from $\achart$, 
and then removing all vertices and transitions that become unreachable. 
We say that a chart $\achart$ has the \emph{loop existence and elimination property (LEE)}
if the procedure, started on~$\achart$, of repeated eliminations of loop subcharts
results in a chart that does~not~have~an~infinite~path.

For the not expressible charts $\chartnei{1}$ and $\chartnei{2}$ in Ex.~\ref{ex:chart:interpretation} the procedure stops immediately,
as they do not contain loop subcharts. Since both of them have infinite paths,
it follows that they do not satisfy LEE. 

Now we consider (see below) three runs of the elimination procedure for the
chart~$\chartof{\cstexpi{0}}$ in Ex.~\ref{ex:chart:interpretation}. 
The \loopentrytransitions\ of loop subcharts that are removed 
in each step are marked in bold.
Each run witnesses that $\chartof{\cstexpi{0}}$ satisfies \LEE. 
Note that loop elimination does not yield a unique result.%
  \footnote{
    Confluence 
               can be shown 
    if a pruning operation is added that permits to drop transitions to deadlocking vertices.}
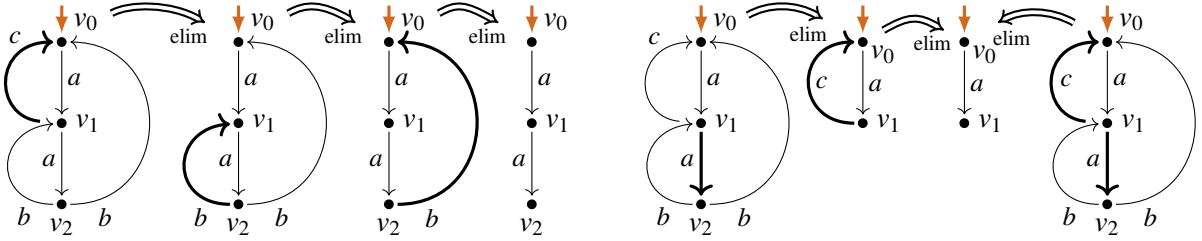
\begin{center}
  \input{figs/ex-12-LEE.tex}
\end{center}\vspace*{-0.75ex}
Runs can be recorded, in the original chart, by attaching a marking label
to transitions that get removed in the elimination procedure. 
That label is the sequence number of the corresponding elimination step.
For the three runs of loop elimination above we get the following 
marking labeled versions of $\achart$, respectively:
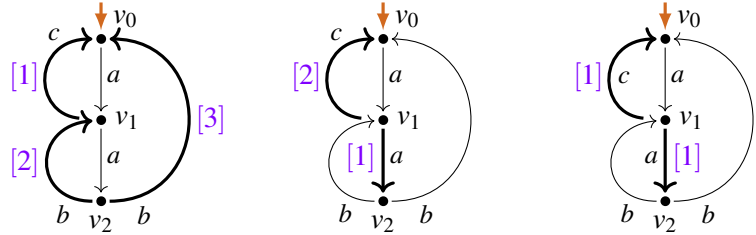
\begin{center} 
  \begin{tikzpicture}
  %
  %
  
\matrix[anchor=north,row sep=0.9cm,column sep=1cm,every node/.style={draw,very thick,circle,minimum width=2.5pt,fill,inner sep=0pt,outer sep=2pt}] at (-3.75,0) {
  \node(v-0-hat-1){};
  \\
  \node(v-1-hat-1){};
  \\
  \node(v-2-hat-1){};
  \\
};
\calcLength(v-0-hat-1,v-1-hat-1){mylen}
\draw[<-,very thick,>=latex,chocolate](v-0-hat-1) -- ++ (90:{0.45*\mylen pt});
\path(v-0-hat-1) ++ ({0.3*\mylen pt},{0.25*\mylen pt}) node{$\averti{0}$};
\draw[->](v-0-hat-1) to node[right,xshift={-0.05*\mylen pt},pos=0.45]{\small $\aact$} (v-1-hat-1); 
\path(v-1-hat-1) ++ ({0.325*\mylen pt},0cm) node{$\averti{1}$};
\draw[->](v-1-hat-1) to node[right,xshift={-0.05*\mylen pt},pos=0.45]{\small $\aact$} (v-2-hat-1);
\draw[->,very thick,shorten <= 5pt](v-1-hat-1) to[out=175,in=180,distance={0.75*\mylen pt}] 
         node[left,pos=0.5,xshift={0.05*\mylen pt}]{$\darkcyan{\loopnsteplab{1}}$} 
         node[above,yshift={0.05*\mylen pt},pos=0.7]{\small $\cact$} (v-0-hat-1);
\path(v-2-hat-1) ++ (-0cm,{-0.275*\mylen pt}) node{$\averti{2}$};
\draw[->,very thick](v-2-hat-1) to[out=180,in=185,distance={0.75*\mylen pt}] 
               node[left,pos=0.5,xshift={0.05*\mylen pt}]{$\darkcyan{\loopnsteplab{2}}$} 
               node[below,yshift={0.0*\mylen pt},pos=0.2]{\small $\bact$} (v-1-hat-1);
\draw[->,very thick](v-2-hat-1) to[out=0,in=0,distance={1.3*\mylen pt}] 
               node[right,pos=0.5,xshift={-0.05*\mylen pt}]{$\darkcyan{\loopnsteplab{3}}$} 
               node[below,yshift={0.00*\mylen pt},pos=0.125]{\small $\bact$} (v-0-hat-1);

\matrix[anchor=north,row sep=0.9cm,every node/.style={draw,very thick,circle,minimum width=2.5pt,fill,inner sep=0pt,outer sep=2pt}] at (0,0) {
  \node(v-0-hat-2){};
  \\
  \node(v-1-hat-2){};
  \\
  \node(v-2-hat-2){};
  \\
};
\calcLength(v-0-hat-2,v-1-hat-2){mylen}
\draw[<-,very thick,>=latex,chocolate](v-0-hat-2) -- ++ (90:{0.45*\mylen pt});
\path(v-0-hat-2) ++ ({0.3*\mylen pt},{0.25*\mylen pt}) node{$\averti{0}$};
\draw[->](v-0-hat-2) to node[right,xshift={-0.05*\mylen pt},pos=0.45]{\small $\aact$} (v-1-hat-2); 
\path(v-1-hat-2) ++ ({0.325*\mylen pt},0cm) node{$\averti{1}$};
\draw[->,very thick](v-1-hat-2) to node[right,xshift={-0.05*\mylen pt},pos=0.45]{\small $\aact$} 
                                   node[left,pos=0.45,xshift={0.05*\mylen pt}]{$\darkcyan{\loopnsteplab{1}}$} (v-2-hat-2);
\draw[->,very thick,shorten <= 5pt](v-1-hat-2) to[out=175,in=180,distance={0.75*\mylen pt}] 
         node[left,pos=0.5,xshift={0.05*\mylen pt}]{$\darkcyan{\loopnsteplab{2}}$} 
         node[above,yshift={0.05*\mylen pt},pos=0.7]{\small $\cact$} (v-0-hat-2);
\path(v-2-hat-2) ++ (-0cm,{-0.275*\mylen pt}) node{$\averti{2}$};
\draw[->](v-2-hat-2) to[out=180,in=185,distance={0.75*\mylen pt}] 
               node[below,yshift={0.0*\mylen pt},pos=0.2]{\small $\bact$} (v-1-hat-2);
\draw[->](v-2-hat-2) to[out=0,in=0,distance={1.3*\mylen pt}]  
               node[below,yshift={0.00*\mylen pt},pos=0.125]{\small $\bact$} (v-0-hat-2);

\matrix[anchor=north,row sep=0.9cm,every node/.style={draw,very thick,circle,minimum width=2.5pt,fill,inner sep=0pt,outer sep=2pt}] at (3.75,0) {
  \node(v-0-hat-3){};
  \\
  \node(v-1-hat-3){};
  \\
  \node(v-2-hat-3){};
  \\
};
\calcLength(v-0-hat-3,v-1-hat-3){mylen}
\draw[<-,very thick,>=latex,chocolate](v-0-hat-3) -- ++ (90:{0.45*\mylen pt});
\path(v-0-hat-3) ++ ({0.3*\mylen pt},{0.25*\mylen pt}) node{$\averti{0}$};
\draw[->](v-0-hat-3) to node[right,xshift={-0.05*\mylen pt},pos=0.45]{\small $\aact$} (v-1-hat-3); 
\path(v-1-hat-3) ++ ({0.325*\mylen pt},0cm) node{$\averti{1}$};
\draw[->,very thick](v-1-hat-3) to node[left,xshift={0.05*\mylen pt},pos=0.45]{\small $\aact$} 
                                   node[right,pos=0.45,xshift={-0.05*\mylen pt}]{$\darkcyan{\loopnsteplab{1}}$} (v-2-hat-3);
\draw[->,very thick,shorten <= 5pt](v-1-hat-3) to[out=175,in=180,distance={0.75*\mylen pt}] 
                                                 node[left,pos=0.5,xshift={0.05*\mylen pt}]{$\darkcyan{\loopnsteplab{1}}$} 
                                                 node[right,xshift={-0.05*\mylen pt},pos=0.5]{\small $\cact$} (v-0-hat-3);
\path(v-2-hat-3) ++ (-0cm,{-0.275*\mylen pt}) node{$\averti{2}$};
\draw[->](v-2-hat-3) to[out=180,in=185,distance={0.75*\mylen pt}] 
                     node[below,yshift={0.0*\mylen pt},pos=0.2]{\small $\bact$} (v-1-hat-3);
\draw[->](v-2-hat-3) to[out=0,in=0,distance={1.3*\mylen pt}] 
                     node[below,yshift={0.00*\mylen pt},pos=0.125]{\small $\bact$} (v-0-hat-3);

\end{tikzpicture}  
\end{center}\vspace*{-0.75ex}
Since all three runs were successful (as they yield charts without infinite paths), 
these recordings (marking-labeled charts) can be viewed as `\LEEwitnesses'.
We now will define the concept of a `layered \LEEwitness' (\LLEEwitness), i.e., a \LEEwitness\
with the added constraint that in the recorded run of the loop elimination procedure
it never happens that a \loopentrytransition\ is removed from within the body of 
a previously removed loop subchart. This refined concept has simpler properties, but is equally powerful.


\begin{defi}
  An \emph{\entrybodylabeling\ of a chart $\achart = \tuple{\verts,\actions,\start,\transs,\termexts}$}
    is a chart~$\acharthat = \tuple{\verts,\actions\times\nat,\start,\transshat,\termexts}$ 
  where $\transshat\subseteq\actions\times\nat$ arises from $\transs$ by adding, for each transition~$\atrans = \triple{\averti{1}}{\aact}{\averti{2}}\in\transs$, 
  to the action label~$\aact$ of $\atrans$ a 
  \emph{marking label} $\aLname\in\nat$,
  yielding \mbox{$\atranshat = \triple{\averti{1}}{\pair{\aact}{\aLname}}{\averti{2}} \in \transshat$}. 
  In an \entrybodylabeling\ we call transitions with marking label $0$ \emph{body transitions},
  and transitions with marking labels in $\natplus$ \emph{entry transitions}.
    
  By an \emph{\entrybodylabeling\ of an \LTS~$\aLTS = \tuple{\states,\actions,\transs,\termexts}$}
  we mean an \LTS~$\aLTShat = \tuple{\states,\actions\times\nat,\transshat,\termexts}$ 
  in which $\transshat \subseteq \states\times(\actions\times\nat)\times\states$ 
  arises from $\transs$ by adding marking labels in $\nat$ to the action labels of transitions.
  
  We define the following designations and concepts for \LTSs, but will use them also for charts.
  Let $\aLTShat$ be an \entrybodylabeling\ of $\aLTS$, and let $\avert$ and $\bvert$ be vertices of $\aLTS$ and $\aLTShat$.
  We denote by $\avert \redi{\bodylabcol} \bvert$\vspace*{-3pt} that there is a body transition $\avert \lt{\pair{\aact}{0}} \bvert$ in $\aLTShat$ for some $\aact\in\actions$,
  and by $\avert \redi{\darkcyan{\loopnsteplab{\aLname}}} \bvert$, for $\aLname\in\natplus$
  that there\vspace*{-3.5pt} is an \entrytransition\ $\avert \lt{\pair{\aact}{\aLname}} \bvert$ in $\aLTShat$ for some $\aact\in\actions$.
  By the set $\entriesof{\aLTShat}$ of \emph{\entrytransition\ identifiers} we denote the set of pairs $\pair{\avert}{\aLname}\in\verts\times\natplus$ 
  such that an \entrytransition\ $\sloopnstepto{\darkcyan{\aLname}}$ departs from $\avert$ in $\aLTShat$. 
  For $\pair{\avert}{\aLname}\in\entriesof{\aLTShat}$,
  we define by $\indsubchartinat{\aLTShatsubscript}{\avert,\aLname}$ the subchart of $\aLTS$ 
  with start vertex $\start$
  that consists of the vertices and transitions which occur on paths in $\aLTS$ 
  as follows: any path that starts with a $\sloopnstepto{\darkcyan{\aLname}}$ \entrytransition\ from $\avert$,
  continues with body transitions only (thus does not cross another \entrytransition), and halts immediately if $\avert$~is~revisited. 
\end{defi}

The three recordings obtained above of the loop elimination procedure for the chart~$\chartof{\cstexpi{0}}$ in Ex.~\ref{ex:chart:interpretation}
indicate \entrybodylabelings\ by signaling the \entrytransitions\ but neglecting body-step labels $0$.

\begin{defi}\label{def:LLEEwitness}\nf
  Let $\aLTS = \tuple{\verts,\actions,\transs,\termexts}$ be an \LTS.
  A \emph{\LLEEwitness} (a \emph{layered \LEEwitness}) \emph{of $\aLTS$}
  is an \entrybodylabeling~$\aLTShat$ of $\aLTS$ that satisfies the following three properties:
  \begin{enumerate}[label={\mbox{\rm (W\arabic*)}},leftmargin=*,align=left,itemsep=0ex]
    \item{}\label{LLEEw:1}%
      \emph{Body-step termination:} There is no infinite path of $\sredi{\bodylabcol}$ transitions in $\aLTS$.
    \item{}\label{LLEEw:2}%
      \emph{Loop condition:} 
      For all $\pair{\astate}{\aLname}\in\entriesof{\aLTShat}$, \mbox{}
        $\indsubchartinat{\aLTShatsubscript}{\astate,\aLname}$ is a loop chart.
    \item{}\label{LLEEw:3}%
      \emph{Layeredness:} 
      For all $\pair{\astate}{\aLname}\in\entriesof{\aLTShat}$, 
        if an \entrytransition\ $\bstate \loopnstepto{\darkcyan{\bLname}} \bstateacc$ departs from a state $\bstate\neq \astate$ of $\indsubchartinat{\aLTShatsubscript}{\astate,\aLname}$,
        then its marking label $\bLname$ satisfies $\bLname < \aLname$.
  \end{enumerate}
  The condition \ref{LLEEw:2} justifies to call an \entrytransition\ in a \LLEEwitness\
  a \emph{\loopentrytransition}.
  For a \loopentrytransition\ $\sredi{\darkcyan{\loopnsteplab{\bLname}}}$ with $\bLname\in\natplus$, we call $\bLname$ its \emph{loop level}. 
  
  For a chart $\achart = \tuple{\verts,\actions,\start,\transs,\termexts}$, we define a \emph{\LLEEwitness~$\acharthat$}
  analogously as an \entrybodylabeling~$\acharthat$ of $\achart$ with the properties \ref{LLEEw:1}, \ref{LLEEw:2}, and \ref{LLEEw:3} with $\achart$ and $\acharthat$
  for $\aLTS$ and $\aLTShat$, respectively.
\end{defi}

\begin{exa}
  The three \entrybodylabelings\ of the chart $\chartof{\cstexpi{0}}$ in Ex.~\ref{ex:chart:interpretation}
  that we have obtained as recordings of runs of the loop elimination procedure
  are \LLEEwitnesses~of~$\chartof{\cstexpi{0}}$, as is easy to verify. 
\end{exa}

\begin{prop}
  If a chart $\achart$ has a \LLEEwitness, then it satisfies \LEE. 
\end{prop}\vspace{-2ex}

\begin{proof}
  Let $\acharthat$ be a \LLEEwitness\ of a chart $\achart$. 
  Repeatedly pick an \entrytransition\ identifier $\pair{\avert}{\aLname}\in\entriesof{\acharthat}$ 
  with $\aLname\in\natplus$ minimal, remove the loop subchart that is generated by \loopentry\ transitions of level $\aLname$ from $\avert$
  (it is indeed a loop by condition~\ref{LLEEw:2} on $\acharthat$, 
   noting that minimality of $\aLname$ and condition~\ref{LLEEw:3} on $\acharthat$ ensure
   the absence of departing \loopentrytransitions\ of lower level), and perform garbage collection. 
  Eventually the part of $\achart$ that is reachable by body transitions from the start vertex
  is obtained. This subchart does not have an infinite path due to condition~\ref{LLEEw:1} on $\acharthat$.
  Therefore $\achart$ satisfies \LEE. 
\end{proof}

The condition \ref{LLEEw:2} on a \LLEEwitness~$\acharthat$ of a chart~$\achart$
requires the loop structure defined by $\acharthat$ to be hierarchical. 
This permits to extract a star expression $\astexptilde$ from $\acharthat$ (defined in \cite{grab:fokk:2020:LICS,grab:fokk:2020:arxiv}) 
that expresses $\achart$ in the sense that $\chartof{\astexptilde}\bisim\achart$ holds,
intuitively by unfolding the underlying chart~$\achart$ to the syntax tree of~$\astexptilde$.

\begin{rem}\label{rem:expressible}\nf
  In \cite{grab:fokk:2020:LICS,grab:fokk:2020:arxiv} we established a connection between charts that have a \LLEEwitness\ (and hence satisfy \LEE)
  and charts that are expressible by \onefree\ star expressions (that is, star expressions without~$\stexpone$, and with binary star iteration instead of unary star iteration).
  By saying that a chart~$\achart$ `is expressible' we mean here that $\achart$ is bisimilar to the 
  chart interpretation $\chartof{\astexptilde}$ of some star expression $\astexptilde$. 
  Now Corollary~6.10 in \cite{grab:fokk:2020:LICS,grab:fokk:2020:arxiv} states 
  that if a chart is expressible by a \onefree\ star expression 
  then its bisimulation collapse has a \LLEEwitness, and thus satisfies \LEE. 
  This statement entails that neither of the charts $\chartnei{1}$ and $\chartnei{2}$ in Ex.~\ref{ex:chart:interpretation}
  is expressible by a \onefree\ star expression,
  because both are bisimulation collapses, and neither of them satisfies \LEE, 
  as we have already observed above.
\end{rem}

\section{LEE may fail for process interpretations of star expressions}%
  \label{LLEE:fail}

The chart interpretations $\chartof{\astexp}$ of $\astexp$, and $\chartof{\bstexp}$ of $\bstexp$
in Ex.~\ref{ex:chart:interpretation} do not satisfy \LEE, contrasting with $\chartof{\cstexpi{0}}$. 
%
For $\chartof{\astexp}$ we find the following run of the loop elimination procedure
that successively eliminates the two loop subcharts that are induced by the cycling transitions at $\astexpi{1}$ and $\astexpi{2}$:
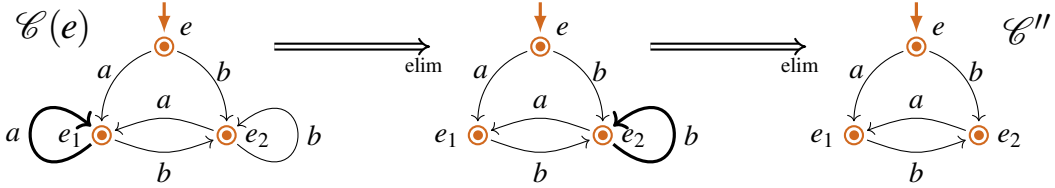
\begin{center}
  \begin{tikzpicture}
  %
  
\matrix[anchor=north,row sep=1cm,column sep=0.7cm,every node/.style={draw,very thick,circle,minimum width=2.5pt,fill,inner sep=0pt,outer sep=2pt}] at (0,0) {
                 &  \node[color=chocolate](e){};
  \\
  \node[color=chocolate](e-1){};  &  \node[draw=none,fill=none](dummy){};  
                                   & \node[color=chocolate](e-2){}; 
  \\
};
\calcLength(e,dummy){mylen}
\draw[<-,very thick,>=latex,chocolate,shorten <=2pt](e) -- ++ (90:{0.5*\mylen pt});
\path(e) ++ ({0.25*\mylen pt},{0.2*\mylen pt}) node{$\astexp$};  
\path(e) ++ ({-1.25*\mylen pt},{0.2*\mylen pt}) node{\Large $\chartof{\astexp}$};

\draw[thick,chocolate] (e) circle (0.12cm);
\draw[->,shorten <=2pt,shorten >=2pt,out=200,in=90] (e) to node[left,pos=0.4]{$\aact$} (e-1);
\draw[->,shorten <=2pt,shorten >=2pt,out=-20,in=90] (e) to node[right,pos=0.4]{$\bact$} (e-2);

\draw[thick,chocolate] (e-1) circle (0.12cm);
\path(e-1) ++ ({-0.35*\mylen pt},{-0.035*\mylen pt}) node{$\astexpi{1}$}; 
\draw[->,very thick,shorten <=2pt,shorten >=2pt,out=220,in=140,distance={1.25*\mylen pt}] (e-1) to node[left]{$\aact$} (e-1);
\draw[->,shorten <=2pt,shorten >=2pt,out=-20,in=200,distance={0.6*\mylen pt}] (e-1) to node[below]{$\bact$} (e-2);

\draw[thick,chocolate] (e-2) circle (0.12cm);
\path(e-2) ++ ({0.35*\mylen pt},{-0.035*\mylen pt}) node{$\astexpi{2}$};  
\draw[->,shorten <=2pt,shorten >=2pt,out=-40,in=40,distance={1.25*\mylen pt}] (e-2) to node[right]{$\bact$} (e-2);
\draw[->,shorten <=2pt,shorten >=2pt,out=160,in=20,distance={0.6*\mylen pt}] (e-2) to node[above]{$\aact$} (e-1);

\matrix[anchor=north,row sep=1cm,column sep=0.7cm,every node/.style={draw,very thick,circle,minimum width=2.5pt,fill,inner sep=0pt,outer sep=2pt}] at (5,0) {
                 &  \node[color=chocolate](e-1){};
  \\
  \node[color=chocolate](e-1-1){};  &  \node[draw=none,fill=none](dummy-1){};  
                                   & \node[color=chocolate](e-2-1){}; 
  \\
};
\calcLength(e-1,dummy-1){mylen}
\draw[<-,very thick,>=latex,chocolate,shorten <=2pt](e-1) -- ++ (90:{0.5*\mylen pt});
\path(e-1) ++ ({0.25*\mylen pt},{0.2*\mylen pt}) node{$\astexp$};  

\draw[thick,chocolate] (e-1) circle (0.12cm);
\draw[->,shorten <=2pt,shorten >=2pt,out=200,in=90] (e-1) to node[left,pos=0.4]{$\aact$} (e-1-1);
\draw[->,shorten <=2pt,shorten >=2pt,out=-20,in=90] (e-1) to node[right,pos=0.4]{$\bact$} (e-2-1);

\draw[thick,chocolate] (e-1-1) circle (0.12cm);
\path(e-1-1) ++ ({-0.35*\mylen pt},{-0.035*\mylen pt}) node{$\astexpi{1}$}; 
\draw[->,shorten <=2pt,shorten >=2pt,out=-20,in=200,distance={0.6*\mylen pt}] (e-1-1) to node[below]{$\bact$} (e-2-1);

\draw[thick,chocolate] (e-2-1) circle (0.12cm);
\path(e-2-1) ++ ({0.35*\mylen pt},{-0.035*\mylen pt}) node{$\astexpi{2}$};  
\draw[->,very thick,shorten <=2pt,shorten >=2pt,out=-40,in=40,distance={1.25*\mylen pt}] (e-2-1) to node[right]{$\bact$} (e-2-1);
\draw[->,shorten <=2pt,shorten >=2pt,out=160,in=20,distance={0.6*\mylen pt}] (e-2-1) to node[above]{$\aact$} (e-1-1);

\draw[-implies,thick,double equal sign distance,
               shorten <= 1.35cm,shorten >= 1.35cm
               ] (e) to node[below,pos=0.7]{\scriptsize elim} (e-1);

\matrix[anchor=north,row sep=1cm,column sep=0.7cm,every node/.style={draw,very thick,circle,minimum width=2.5pt,fill,inner sep=0pt,outer sep=2pt}] at (10,0) {
                 &  \node[color=chocolate](e-2){};
  \\
  \node[color=chocolate](e-1-2){};  &  \node[draw=none,fill=none](dummy-2){};  
                                   & \node[color=chocolate](e-2-2){}; 
  \\
};
\calcLength(e-2,dummy-2){mylen}
\draw[<-,very thick,>=latex,chocolate,shorten <=2pt](e-2) -- ++ (90:{0.5*\mylen pt});
\path(e-2) ++ ({0.25*\mylen pt},{0.2*\mylen pt}) node{$\astexp$}; 
\path(e-2) ++ ({1.25*\mylen pt},{0.2*\mylen pt}) node{\Large $\achart''$};

\draw[thick,chocolate] (e-2) circle (0.12cm);
\draw[->,shorten <=2pt,shorten >=2pt,out=200,in=90] (e-2) to node[left,pos=0.4]{$\aact$} (e-1-2);
\draw[->,shorten <=2pt,shorten >=2pt,out=-20,in=90] (e-2) to node[right,pos=0.4]{$\bact$} (e-2-2);

\draw[thick,chocolate] (e-1-2) circle (0.12cm);
\path(e-1-2) ++ ({-0.35*\mylen pt},{-0.035*\mylen pt}) node{$\astexpi{1}$}; 
\draw[->,shorten <=2pt,shorten >=2pt,out=-20,in=200,distance={0.6*\mylen pt}] (e-1-2) to node[below]{$\bact$} (e-2-2);

\draw[thick,chocolate] (e-2-2) circle (0.12cm);
\path(e-2-2) ++ ({0.35*\mylen pt},{-0.035*\mylen pt}) node{$\astexpi{2}$};  
\draw[->,shorten <=2pt,shorten >=2pt,out=160,in=20,distance={0.6*\mylen pt}] (e-2-2) to node[above]{$\aact$} (e-1-2);

\draw[-implies,thick,double equal sign distance,
               shorten <= 1.35cm,shorten >= 1.35cm
               ] (e-1) to node[below,pos=0.7]{\scriptsize elim} (e-2);

\end{tikzpicture}
\end{center}\vspace{-2.5ex}
The resulting chart $\achart''$ does not contain loop subcharts any more,
because taking, for example, a transition from $\astexpi{1}$ to $\astexpi{2}$ as an \entrytransition\ 
does not yield a loop subchart, because in the induced subchart immediate termination is not only possible at the start vertex~$\astexpi{1}$ but also in the body vertex $\astexpi{2}$,
in contradiction to \ref{loop:3}.  
But while $\achart''$ does not contain a loop subchart any more, it still has an infinite trace.  
Therefore it follows that $\chartof{\astexp}$ does not satisfy \LEE.

In order to see that $\chartof{\bstexp}$ does not satisfy \LEE, we can consider a run of the loop elimination procedure 
that successively removes the cyclic transitions at $\bstexpi{1}$, $\bstexpi{2}$, and $\bstexpi{3}$. 
After these removals a variant of the not expressible chart $\chartnei{2}$ is obtained that still describes an infinite behavior, 
but that does not contain any\vspace*{-0.25ex} loop subchart. The latter can be argued analogously as for $\chartnei{2}$, namely that 
for all choices of \entrytransitions\ between $\bstexpi{1}$, $\bstexpi{2}$, and $\bstexpi{3}$ the loop condition \ref{loop:2} fails.
We conclude that $\chartof{\bstexp}$ does not~satisfy~\LEE.

The reason for this failure of \LEE\ is that, while the syntax trees of star expressions can provide a nested loop-chart structure,
this is not guaranteed by the specific form of the TSS~$\StExpTSS$. 
Execution of an iteration $\stexpit{\cstexp}$ in an expression $\stexpprod{\stexpit{\cstexp}}{\dstexp}$  
leads eventually, in case that termination is reachable in $\cstexp$, to an iterated derivative $\stexpprod{(\stexpprod{\stexpone}{\stexpit{\cstexp}})}{\dstexp}$. 
Also, as in the examples above, 
an iterated derivative $\stexpprod{(\stexpprod{\cstexptilde}{\stexpit{\cstexp}})}{\dstexp}$ with $\terminates{\cstexptilde}$ may be reached. 
In these cases,
continued execution will bypass the initial term $\stexpprod{\stexpit{\cstexp}}{\dstexp}$, and either proceed 
with another execution of the iteration to $\stexpprod{(\stexpprod{\cstexpacc}{\stexpit{\cstexp}})}{\dstexp}$, where $\cstexpacc$ is a derivative of $\cstexp$,
or take a step into the exit to $\dstexpacc$, where $\dstexpacc$ is a derivative of~$\dstexp$.  
In both cases the execution does not return to the initial term $\stexpit{\cstexp}$ of the execution,
as would be required for a loop subchart at $\stexpit{\cstexp}$ to arise in accordance with loop condition~\ref{loop:2}.%

\section{Recovering LEE for a variant definition 
         of the process semantics}%
  \label{recover:LEE}

A remedy for the frequent failure of \LEE\ for the chart translation of star expressions can consist in the use
of `\onetransitions'. Such transitions may be used to create a back-link to an expression $\stexpprod{\stexpit{\cstexp}}{\dstexp}$
from an iterated derivative $\stexpprod{(\stexpprod{\cstexptilde}{\stexpit{\cstexp}})}{\dstexp}$ with $\terminates{\cstexptilde}$
(where $\cstexptilde$ is an iterated derivative of $\cstexp$) that is reached by a descent of the execution
into the body of $\cstexp$. 
This requires an adapted refinement of the TSS~$\StExpTSS$ from page~\pageref{StExpTSS}. 
While different such refinements are conceivable, 
our choice is to introduce \onetransitions\ most sparingly, making sure that 
every \onetransition\ can be construed as a backlink in an accompanying \LLEEwitness.


In particular
we want to create transition rules that facilitate a back-link to an expression $\stexpit{\cstexp}$ after the execution has descended into $\cstexp$
reaching $\stexpprod{\cstexptilde}{\stexpit{\cstexp}}$ with $\terminates{\cstexptilde}$.
In order to distinguish a concatenation expression $\stexpprod{\cstexptilde}{\stexpit{\cstexp}}$ 
that arises from the descent of the execution into an iteration
$\stexpit{\cstexp}$ from other concatenation expressions 
we introduce a variant operation $\sstexpstackprod$. 
The rules of the refined TSS should guarantee that in the example
the reached iterated derivative of $\stexpit{\cstexp}$ is a `stacked star expression' $\stexpstackprod{\csstexp}{\stexpit{\cstexp}}$
where $\csstexp$ is itself a stacked star expression that denotes an iterated derivative of $\cstexp$.
If now $\csstexp$ is also a star expression~$\cstexptilde$~with~$\terminates{\cstexptilde}$,
then the expression $\stexpstackprod{\csstexp}{\stexpit{\cstexp}}$ of the form $\stexpstackprod{\cstexptilde}{\stexpit{\cstexp}}$
should permit a \onetransition\ that returns to $\stexpit{\cstexp}$. 

This intuition guided the definition of the rules of the TSS~$\stackStExpTSS$ in Def.~\ref{def:stackStExpTSS} below,
starting from the adaptation of the 
              rule for steps from iterations $\stexpit{\astexp}$,
and the rule that creates \onetransition\ backlinks to iterations $\stexpit{\astexpi{2}}$ 
from stacked expressions $\stexpstackprod{\asstexpi{1}}{\stexpit{\astexpi{2}}}$ with $\terminates{\asstexpi{1}}$. 
The `stacked product' $\sstexpstackprod$ has the following features:
$\stexpstackprod{\asstexp}{\stexpit{\astexp}}$ never permits immediate termination;
for defining transitions it behaves similarly as concatenation $\sstexpprod$ except that a transition 
from $\stexpstackprod{\asstexp}{\stexpit{\astexp}}$ into $\stexpit{\astexp}$ when $\asstexp$ permits immediate termination
now requires a \protect\onetransition\ to $\stexpit{\astexp}$ first.
The formulation of these rules of $\stackStExpTSS$ led to the tailor-made set of stacked star expressions as defined below.

\begin{defi}
            \nf\label{def:stackStExp}\enlargethispage{2ex}
  Let $\actions$ be a set whose members we call \emph{actions}.
  The set $\stackStExpover{\actions}$ of \emph{stacked star expressions over (actions in) $\actions$} is defined by the following grammar:
  \begin{center}
    $
    \asstexp
      \;\;\BNFdefdby\;\;
        \astexp
          \BNFor
        \stexpprod{\asstexp}{\astexp}
          \BNFor
        \stexpstackprod{\asstexp}{\stexpit{\astexp}}
          \qquad\text{(where $\astexp\in\StExpover{\actions}$)} \punc{.}  
          $
  \end{center}
  Note that the set $\StExpover{\actions}$ of star expressions would arise again if the clause $\stexpstackprod{\asstexp}{\stexpit{\astexp}}$ were dropped.
  
  The \emph{star height} $\sth{\asstexp}$ of stacked star expressions $\asstexp$ is defined by adding 
  $\sth{ \stexpprod{\asstexp}{\astexp} } \defdby \max \setexp{ \sth{\asstexp} + \sth{\astexp} }$,
  and
  $\sth{ \stexpstackprod{\asstexp}{\stexpit{\astexp}} } \defdby \max \setexp{ \sth{\asstexp} + \sth{\stexpit{\astexp}} }$
  to the defining clauses for star height of star expressions.
 
  The \emph{projection function} $\sproj \funin \stackStExpover{\actions} \to \StExpover{\actions}$  
  is defined by interpreting $\sstexpstackprod$ as $\sstexpprod$ by the clauses:
  $\proj{\stexpprod{\asstexp}{\astexp}} \defdby \stexpprod{\proj{\asstexp}}{\astexp}$, \mbox{}
  $\proj{\stexpstackprod{\asstexp}{\stexpit{\astexp}}} \defdby \stexpprod{\proj{\asstexp}}{\stexpit{\astexp}}$, \mbox{}
  and $\proj{\astexp} \defdby \astexp$,
  for all $\asstexp\in\stackStExpover{\actions}$, and $\astexp\in\StExpover{\actions}$.
\end{defi}  
  
  %
  %
  
\renewcommand{\sone}{\alert{1}}  
  
\begin{defi}  
  By a \emph{labeled transition system with termination, actions in $\actions$ and empty steps (a \oneLTS)}
  we mean a 5\nb-tuple $\tuple{\states,\actions,\sone,\sred,\sterminates}$ 
    where $\states$ is a \nonempty\ set of \emph{states},
    $\actions$ is a set of $\emph{(proper) action labels}$,
    $\sone\notin\actions$ is the specified \emph{empty step label},
    $\sred \subseteq \states\times\oneactions\times\states$ is the \emph{labeled transition relation},
    where $\oneactions \defdby \actions \cup \setexp{\sone}$ is the set of action labels including $\sone$, 
    and $\sterminates \subseteq \verts$ is a set of \emph{states with immediate termination}. 
    Note that then $\tuple{\states,\oneactions,\sred,\sterminates}$ is an \LTS.
  In such a \oneLTS, 
  we call a transition in $\transs\cap(\states\times\actions\times\states)$ (labeled by a \emph{proper action} in $\actions$) 
          a \emph{proper transition},
  and a transition in $\transs\cap(\states\times\setexp{\sone}\times\states)$ (labeled by the \emph{empty-step symbol}~$\sone$)
          a \emph{\onetransition} .  
  Reserving non-underlined action labels like $\aact,\bact,\ldots$ for proper actions,
  we use underlined action label symbols like $\alert{\aoneact}$ for actions labels in the set $\oneactions$ 
  that includes the label $\sone$.    
\end{defi}
  
\begin{defi}\label{def:stackStExpTSS}\nf  
  The transition system specification~$\stackStExpTSSover{\actions}$ 
  has the following axioms and rules,
  where \mbox{$\alert{\stexpone}\notin\actions$} is an additional label (for representing empty steps),
  $\aact\in\actions$,
  $\alert{\aoneact}\in\oneactions\defdby\actions\cup\setexp{\alert{\stexpone}}$,
  stacked star expressions $\asstexpi{1},\asstexpi{2},\asstexpacci{1},\asstexpacci{2},\asstexpacc\in\stackStExpTSSover{\actions}$,
  and star expressions $\astexpi{1},\astexpi{2},\stexpit{\astexpi{2}},\stexpit{\astexp}\in\StExpover{\actions}$
  (here and below we highlight in red transitions that may involve $\alert{\stexpone}$\nb-transitions):
  \begin{center}
    $
   \begin{aligned}
     &
     \AxiomC{\phantom{$\terminates{\stexpone}$}}
     \UnaryInfC{$\terminates{\stexpone}$}
     \DisplayProof
     & & 
     & \hspace*{2ex} & & & 
     \AxiomC{$ \terminates{\astexpi{i}} $}
     \RightLabel{\scriptsize $(i\in\setexp{1,2})$}
     \UnaryInfC{$ \terminates{(\stexpsum{\astexpi{1}}{\astexpi{2}})} $}
     \DisplayProof
     & \hspace*{2ex} & & & 
     \AxiomC{$\terminates{\astexpi{1}}$}
     \AxiomC{$\terminates{\astexpi{2}}$}
     \BinaryInfC{$\terminates{(\stexpprod{\astexpi{1}}{\astexpi{2}})}$}
     \DisplayProof
     & \hspace*{2ex} & & & 
     \AxiomC{$\phantom{\terminates{\stexpit{\astexp}}}$}
     \UnaryInfC{$\terminates{(\stexpit{\astexp})}$}
     \DisplayProof
   \end{aligned} 
   $
   \\[0.35ex]
   $
   \begin{aligned}
     & 
     \AxiomC{$\phantom{a_i \:\lt{a_i}\: \stexpone}$}
     \UnaryInfC{$a \:\lt{a}\: \stexpone$}
     \DisplayProof
     & & & &
     \AxiomC{$ \astexpi{i} \:\lt{\aact}\: \asstexpacci{i} $}
     \RightLabel{\scriptsize $(i\in\setexp{1,2})$}
     \UnaryInfC{$ \stexpsum{\astexpi{1}}{\astexpi{2}} \:\lt{\aact}\: \asstexpacci{i} $}
     \DisplayProof 
     & & & & 
     \AxiomC{$ \asstexpi{1} \:\lt{\alert{\aoneact}}\: \asstexpacci{1} $}
     \UnaryInfC{$ \stexpprod{\asstexpi{1}}{\astexpi{2}} \:\lt{\alert{\aoneact}}\: \stexpprod{\asstexpacci{1}}{\astexpi{2}} $}
     \DisplayProof
     & &
     \AxiomC{$\terminates{\astexpi{1}}$}
     \AxiomC{$ \astexpi{2} \:\lt{a}\: \asstexpacci{2} $}
     \BinaryInfC{$ \stexpprod{\astexpi{1}}{\astexpi{2}} \:\lt{\aact}\: \asstexpacci{2} $}
     \DisplayProof
     & & & & 
     \AxiomC{$   \phantom{\astexpi{1}}
               \astexp \:\lt{a}\: \asstexpacc 
                 \phantom{\asstexpacci{1}} $}
     \UnaryInfC{$\stexpit{\astexp} \:\lt{a}\: \stexpstackprod{\asstexpacc}{\stexpit{\astexp}}$}
     \DisplayProof
     \\[-0.25ex]
     & 
     & & & &
     & & & & 
     \AxiomC{$ \asstexpi{1} \:\lt{\alert{\aoneact}}\: \asstexpacci{1} $}
     \UnaryInfC{$ \stexpstackprod{\asstexpi{1}}{\stexpit{\astexpi{2}}} \:\lt{\alert{\aoneact}}\: \stexpstackprod{\asstexpacci{1}}{\stexpit{\astexpi{2}}} $}
     \DisplayProof
     & &
     \AxiomC{$ \hspace*{5ex} \terminates{\astexpi{1}}\rule{0pt}{13pt} \hspace*{5ex} $}
     \UnaryInfC{$ \stexpstackprod{\astexpi{1}}{\stexpit{\astexpi{2}}} \:\lt{\alert{\stexpone}}\: \stexpit{\astexpi{2}} $}
     \DisplayProof
   \end{aligned}
    $
  \end{center}
  If $\asstexp \lt{\aoneact} \asstexpacc$ is derivable in $\stackStExpTSSover{\actions}$, for $\asstexp,\asstexpacc\in\StExpover{\actions}$,
  and $\aoneact\in\oneactions$, then we say that $\asstexpacc$ is a \emph{subderivative} of $\asstexp$.
  The \TSS\ $\stackStExpTSSover{\actions}$ defines the variant process semantics for stacked star expressions in $\stackStExpover{\actions}$
  as the \emph{stacked star expressions \oneLTS}~$\oneLTSof{\stackStExpover{\actions}} \defdby \oneLTSdefdby{\stackStExpTSSover{\actions}}$,
  where $\oneLTSdefdby{\stackStExpTSSover{\actions}} = \tuple{\stackStExpover{\actions},\actions,\transs,\termexts}$ 
  is the \emph{\oneLTS\ generated by} $\stackStExpTSSover{\actions}$, that is,
  its transitions $\transs \subseteq  \stackStExpover{\actions}\times\oneactions\times\stackStExpover{\actions}$,
  and its immediately terminating vertices $\termexts\subseteq\stackStExpover{\actions}$
  are defined via derivations in~$\StExpTSSover{\actions}$ in the natural way.
 
  For sets $S\subseteq\stackStExpover{\actions}$, $S$\nb-generated sub-1-LTSs are defined analogously as for~$\LTSof{\StExpover{\actions}}$.
\end{defi}%

\begin{defi}
  A \emph{\onechart} is a (rooted) \oneLTS\ $\tuple{\verts,\actions,\sone,\start,\transs,\termexts}$ 
  such that $\tuple{\verts,\actions,\sone,\transs,\termexts}$ is a \oneLTS\
  whose states we refer to as \emph{vertices},
  and where $\start\in\verts$ is called the \emph{start vertex} of the \onechart.  
\end{defi}

\begin{defi}\nf\label{def:onechart:interpretation}
  The \emph{\onechart\ interpretation $\onechartof{\astexp} = \tuple{\vertsof{\astexp},\actions,\sone,\astexp,\gentranssof{\astexp},\gentermextsof{\astexp}}$} 
  of a star expression $\astexp\in\StExpover{\actions}$ is 
  the $\astexp$\nb-rooted version 
    of the $\setexp{\astexp}$\nb-generated 
           sub-\oneLTS~$\oneLTSof{\setexp{\astexp}} = \tuple{\genvertsof{\setexp{\astexp}},\actions,\sone,\gentranssof{\setexp{\astexp}},\gentermextsof{\setexp{\astexp}}}$
  of $\oneLTSof{\StExpover{\actions}}$.
\end{defi}

\renewcommand{\iact}{\iactalert}
\renewcommand{\soneterminates}{\soneterminatesalert}%

%
%
In order to link the \oneLTS~$\oneLTSof{\stackStExpover{\actions}}$
to the \LTS~$\LTSof{\StExpover{\actions}}$, 
we need to take account\vspace*{-2pt} of the semantics of $\alert{\stexpone}$\nb-transitions as empty steps
(see Vrancken \cite{vran:1997}).
For this,\vspace*{-0.25ex}
  we define 
  `induced transitions'~$\silt{\cdot}$, and `induced termination'~$\soneterminates$ 
  as follows:                                                        
$\asstexp \ilt{\aact} \asstexpacc$ holds if there is a sequence of $\alert{\stexpone}$\nb-transitions from $\asstexp$ 
to\vspace*{-2pt} some $\asstexptilde$ from which there is an $\aact$\nb-transition to $\asstexpacc$
(the asymmetric notation $\silt{\cdot}$ is intended to reflect this asymmetry),
and $\oneterminates{\asstexp}$ holds\vspace*{-0.35ex}
if there is a sequence of $\alert{\stexpone}$\nb-transitions from $\asstexp$ to some $\asstexptilde$ with~$\terminates{\asstexptilde}$.

\begin{defi}\nf\label{def:indLTS}
  Let $\aoneLTS = \tuple{\states,\actions,\sone,\slt{},\termexts}$ be a \oneLTS.
  By the \emph{induced LTS of $\aoneLTS$}, and the \emph{LTS induced~by~$\aoneLTS$},
  we mean the \LTS~$\indLTSof{\aoneLTS} = \tuple{\states,\actions,\silt{\cdot},\soneterminatesalert}$
  where $\silt{\cdot} \subseteq \states\times\oneactions\times\states$ 
  is the \emph{induced transition relation},
  and~$\soneterminatesalert\subseteq\verts$ is the set of vertices with \emph{induced termination}
  that are defined as follows, for all $\astate,\astateacc\in\states$ and $\aact\in\actions\,$:
  \begin{enumerate}[label={(ind-\arabic{*})},align=right,leftmargin=*]
    \item{}\label{it:indtrans::def:indLTS}
      $\astate \ilt{\aact} \astateacc$ holds if 
        $\astate = \astatei{0} \lt{\sone} \astatei{1} \lt{\sone} \ldots \lt{\sone} \astatei{n} \lt{\aact} \astateacc$, 
        for some $\astatei{0},\ldots,\astatei{n}\in\verts$ and $n\in\nat$
      (we then say that there is an \emph{induced transition} between $\astate,\astateacc\in\states$ \emph{with respect to $\aoneLTS$}),
    \item{}\label{it:indterm::def:indLTS}  
      $\oneterminates{\astate}$ holds if
      $\astate = \astatei{0} \lt{\sone} \astatei{1} \lt{\sone} \ldots \lt{\sone} \astatei{n} \logand \oneterminates{\astatei{n}}$, 
        for some $\astatei{0},\ldots,\astatei{n}\in\verts$ and $n\in\nat$
      (then we say that $\astate$ \emph{has induced termination with respect to $\aoneLTS$}). 
  \end{enumerate}
\end{defi}

\begin{defi}\nf\label{def:indchart}
  The \emph{induced chart} $\indscchartof{\aonechart} = \tuple{\verts,\actions,\start,\silt{\cdot},\soneterminatesalert}$\vspace*{-2.5pt}
    of a \onechart~$\aonechart = \tuple{\verts,\actions,\sone,\start,\transs,\termexts}$ is defined
  analogously as the induced \LTS\ of a \oneLTS\ in Def.~\ref{def:indLTS}
  with the induced transition relation $\silt{\cdot}$ defined analogously to \ref{it:indtrans::def:indLTS},
  and the set of vertices $\soneterminatesalert$ with induced termination defined analogously to \mbox{\ref{it:indterm::def:indLTS}}.
\end{defi}

\begin{lem}\label{lem:StExpLTS:funbisim:StExpindLTS}
  The projection function $\sproj$ defines a bisimulation between the
  induced LTS $\indLTSof{\oneLTSof{\stackStExpover{\actions}}}$ of the stacked star expressions \oneLTS~$\stackStExpover{\actions}$
  and the star expressions LTS $\LTSof{\StExpover{\actions}}$.
\end{lem}

 
With this lemma, the proof of which we outline in Section~\ref{proofs}, we will be able to prove the following connection between the chart interpretation
and the \onechart\ interpretation of a star expression. 

\begin{thm}\label{thm:onechart-int:funbisim:chart-int}
  $\indscchartof{\onechartof{\astexp}} \funbisim \chartof{\astexp}$ holds for all $\astexp\in\StExpover{\actions}$,
  that is, there is a functional bisimulation from the induced chart of the \onechart\ interpretation of a star expression~$\astexp$
  to the chart interpretation of $\astexp$.
\end{thm}
  
For the construction of \LLEEwitnesses\ for the \onechart\ interpretation $\onechartof{\cdot}$
we need to distinguish `\txtnormedplus' stacked star expressions 
  that permit an induced-transition path of \underline{\smash{positive length}} to an expression with induced termination
from those that do not enable such a path. 
                        This property slightly strengthens normedness of expressions, 
  which means the existence of a path to an expression with immediate termination
  (or equally, an \underline{\smash{arbitrary-length}} induced-transition path to an expression~with~induced~termination).

\begin{defi}
  Let $\asstexp\in\stackStExpover{\actions}$.
  We say that $\asstexp$ is \emph{\txtnormedplus} (and respectively, $\asstexp$ is \emph{normed}) 
  if there\vspace*{-2pt} is $\asstexpacc\in\stackStExpover{\actions}$ 
  such that 
  $\asstexp \ilttc{\cdot} \asstexpacc$ and $\oneterminates{\asstexpacc}$ in $\indLTSof{\oneLTSof{\StExpover{\actions}}}$
    (resp., $\asstexp \redrtc \asstexpacc$ and $\terminates{\asstexpacc}$ in $\oneLTSof{\StExpover{\actions}}$).
\end{defi}

These properties permit inductive definitions, and therefore they are easily decidable.
Also, a stacked star expression is \txtnormedplus\ if and only if it enables a transition to a normed stacked~star~expression.


Now we define, similarly as we have done so for \onefree\ star expressions in \cite{grab:fokk:2020:LICS,grab:fokk:2020:arxiv},
a refinement of the TSS~$\stackStExpTSS$ into a TSS that will supply \entrybodylabelings\ for \LLEEwitnesses, 
by adding marking labels to the rules of $\stackStExpTSS$.
In particular, body labels are added to transitions that cannot return to their source expression. 
The rule for transitions from an iteration $\stexpit{\astexp}$ is split into the case in which $\astexp$ is \txtnormedplus\ or not.
Only if $\astexp$ is \txtnormedplus\ can $\stexpit{\astexp}$ return to itself after a positive number of steps, and then a \mbox{(loop-)} \entrytransition\ 
with the star height $\sth{\stexpit{\astexp}}$ of $\stexpit{\astexp}$ as its level is created;
otherwise a body label~is~introduced.

\begin{defi}\label{def:stackStExpTSShat}\nf
  The TSS~$\stackStExpTSShatover{\actions}$ 
  has the following rules, where $\darkcyan{\alab} \in \setexp{\bodylabcol} \cup \descsetexp{ \darkcyan{\loopnsteplab{\aLname}} }{ \aLname\in\natplus }$:\vspace{-1.5ex}
  \begin{center}
    $
    \begin{gathered}
      \begin{aligned}
        &
        \AxiomC{\phantom{$\terminates{\stexpone}$}}
        \UnaryInfC{$\terminates{\stexpone}$}
        \DisplayProof
        & & 
        & \hspace*{2ex} & & & 
        \AxiomC{$ \terminates{\astexpi{i}} $}
        \RightLabel{\scriptsize $(i\in\setexp{1,2})$}
        \UnaryInfC{$ \terminates{(\stexpsum{\astexpi{1}}{\astexpi{2}})} $}
        \DisplayProof
        & \hspace*{2ex} & & & 
        \AxiomC{$\terminates{\astexpi{1}}$}
        \AxiomC{$\terminates{\astexpi{2}}$}
        \BinaryInfC{$\terminates{(\stexpprod{\astexpi{1}}{\astexpi{2}})}$}
        \DisplayProof
        & \hspace*{2ex} & & & 
        \AxiomC{$\phantom{\terminates{\stexpit{\astexp}}}$}
        \UnaryInfC{$\terminates{(\stexpit{\astexp})}$}
        \DisplayProof
      \end{aligned} 
      \\
      \begin{aligned}
        & \hspace*{8ex}
        \AxiomC{$\phantom{a_i \:\lti{a_i}{\darkcyan{\bodylabcol}}\: \stexpone}$}
        \UnaryInfC{$a \:\lti{a}{\bodylabcol}\: \stexpone$}
        \DisplayProof
        & & 
        \AxiomC{$ \astexpi{i} \:\lti{\aact}{\darkcyan{\alab}}\: \asstexpacci{i} $}
        \RightLabel{\scriptsize $(i\in\setexp{1,2})$}
        \UnaryInfC{$ \stexpsum{\astexpi{1}}{\astexpi{2}} \:\lti{\aact}{\bodylabcol}\: \asstexpacci{i} $}
        \DisplayProof 
        & & 
        \AxiomC{$ \asstexpi{1} \:\lti{\alert{\aoneact}}{\darkcyan{\alab}}\: \asstexpacci{1} $}
        \UnaryInfC{$ \stexpprod{\asstexpi{1}}{\astexpi{2}} \:\lti{\alert{\aoneact}}{\darkcyan{\alab}}\: \stexpprod{\asstexpacci{1}}{\astexpi{2}} $}
        \DisplayProof
        & &
        \AxiomC{$\terminates{\astexpi{1}}$}
        \AxiomC{$ \astexpi{2} \:\lti{a}{\darkcyan{\alab}}\: \asstexpacci{2} $}
        \BinaryInfC{$ \stexpprod{\astexpi{1}}{\astexpi{2}} \:\lti{a}{\bodylabcol}\: \asstexpacci{2} $}
        \DisplayProof
        \\
        & 
        \AxiomC{$   \phantom{\astexpi{1}}
                 \astexp \:\lti{a}{\darkcyan{\alab}}\: \asstexpacc 
                   \phantom{\asstexpacci{1}} $}
        \AxiomC{\small ($\astexp$ \txtnormedplus)}    
        \insertBetweenHyps{\hspace*{-0.5ex}}        
        \BinaryInfC{$\stexpit{\astexp} \:\lti{\aact}{\darkcyan{\loopnsteplab{\sth{\stexpit{\astexp}}}}}\: \stexpstackprod{\asstexpacc}{\stexpit{\astexp}}$}
        \DisplayProof
        & &
        \AxiomC{$   \phantom{\astexpi{1}}
                 \astexp \:\lti{a}{\darkcyan{\alab}}\: \asstexpacc 
                   \phantom{\astexpacci{1}} $}
        \insertBetweenHyps{\hspace*{-0.5ex}}     
        \AxiomC{\small ($\astexp$ not \txtnormedplus)} 
        \BinaryInfC{$\stexpit{\astexp} \:\lti{a}{\bodylabcol}\: \stexpstackprod{\asstexpacc}{\stexpit{\astexp}}$}
        \DisplayProof 
        & & 
        \AxiomC{$ \asstexpi{1} \:\lti{\alert{\aoneact}}{\darkcyan{\alab}}\: \asstexpacci{1} $}
        \UnaryInfC{$ \stexpstackprod{\asstexpi{1}}{\stexpit{\astexpi{2}}} \:\lti{\alert{\aoneact}}{\darkcyan{\alab}}\: \stexpstackprod{\asstexpacci{1}}{\stexpit{\astexpi{2}}} $}
        \DisplayProof
        & &
        \AxiomC{$ \hspace*{5ex} \terminates{\astexpi{1}}\rule{0pt}{13pt} \hspace*{5ex} $}
        \UnaryInfC{$ \stexpstackprod{\astexpi{1}}{\stexpit{\astexpi{2}}} \:\lti{\alert{\stexpone}}{\bodylabcol}\: \stexpit{\astexpi{2}} $}
        \DisplayProof
      \end{aligned}
    \end{gathered}
    $
  \end{center}
  The \entrybodylabeling~$\oneLTShatof{\stackStExpover{\actions}} = \tuple{\StExpover{\actions},\oneactions,\slti{\cdot}{\cdot},\sterminates}$ 
  of the stacked star expressions \oneLTS\vspace*{-2pt} $\oneLTSof{\stackStExpover{\actions}}$
  is defined as the \LTS\ generated by $\stackStExpTSShatover{\actions}$,
  with $\termexts\subseteq\stackStExpover{\actions}$ as set of immediately terminating vertices,
  and with set $\slti{\cdot}{\cdot} \subseteq  \stackStExpover{\actions}\times(\oneactions\times\nat)\times\stackStExpover{\actions}$ of transitions,
  where $\oneactions = \actions\cup\setexp{\sone}$. 
\end{defi}      

\begin{defi}
  For every star expression~$\astexp\in\StExpover{\actions}$\vspace*{-0.4ex}
  we denote by $\onecharthatof{\astexp}$ the \entrybodylabeling\ 
  that is the chart formed as the $\astexp$\nb-rooted sub-\LTS\ generated by $\setexp{\astexp}$ of the \entrybodylabeling~$\oneLTShatof{\stackStExpover{\actions}}$.
\end{defi}

For this \entrybodylabeling\ we will show in Section~\ref{proofs} that it recovers the property \LEE\ 
for the stacked star expressions \oneLTS, and as a consequence, for the \onechart\ interpretation of~star~expressions.\enlargethispage{1ex}

\begin{lem}\label{lem:oneLTShat:stackStExps:is:LLEEw}
  $\oneLTShatof{\stackStExpover{\actions}}$ is an \entrybodylabeling, and indeed a \LLEEwitness, of $\oneLTSof{\stackStExpover{\actions}}$. 
\end{lem}

\begin{thm}\label{thm:onechart-int:LLEEw}
  For every $\astexp\in\StExpover{\actions}$,
  the \entrybodylabeling~$\onecharthatof{\astexp}$ of $\onechartof{\astexp}$ is a \LLEEwitness\ of $\onechartof{\astexp}$.
  Hence the \onechart\ interpretation $\onechartof{\astexp}$ of a star expression $\astexp\in\StExpover{\actions}$ satisfies the property \LEE. 
\end{thm}

\begin{exa}
  We consider the chart interpretations~$\chartof{\astexp}$ and $\chartof{\bstexp}$ 
  for the star expressions $\astexp$ and $\bstexp$ in Ex.~\ref{ex:chart:interpretation}
  for which we saw in Section~\ref{LLEE:fail} that \LEE\ fails. 
  We first illustrate the \onechart\ interpretations~$\onechartof{\astexp}$ of $\astexp$ 
  together with the \entrybodylabeling~$\onecharthatof{\astexp}$ of $\onechartof{\astexp}$. 
  The dotted transitions indicate \onetransitions.
  The expressions at \noninitial\ vertices of $\onechartof{\astexp}$ are  
  $\asstexpacci{1} = \stexpstackprod{(\stexpprod{(\stexpstackprod{\stexpone}{\stexpit{\aact}})}{\stexpit{\bact}})}{\astexp}$, \mbox{}
  $\asstexpi{1} = \stexpstackprod{(\stexpprod{\stexpit{\aact}}{\stexpit{\bact}})}{\astexp}$, \mbox{}
  $\asstexpi{2} = \stexpstackprod{\stexpit{\bact}}{\astexp}$, and 
  $\asstexpacci{2} = \stexpstackprod{(\stexpstackprod{\stexpone}{\stexpit{\bact}})}{\astexp}$,
  which are obtained as iterated derivatives via the TSS~$\stackStExpTSS$.
  Furthermore, we depict the induced chart $\indscchartof{\onechartof{\astexp}}$ of $\onechartof{\astexp}$,
  and its relationship to the chart interpretation $\chartof{\astexp}$ of $\astexp$
  via the projection function $\sproj$ that defines a functional bisimulation 
  (which we indicate via arrows 
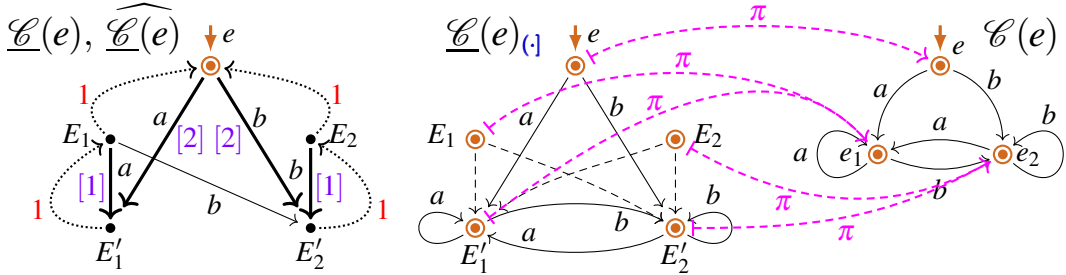
\begin{tikzpicture}\draw[|->,thick,magenta,densely dashed](0,0) -- ++ (0:{12pt});\end{tikzpicture}).
  \begin{center}
    \begin{tikzpicture}[scale=0.97]

\matrix[anchor=north,row sep=1cm,column sep=1.2cm,every node/.style={draw,very thick,circle,minimum width=2.5pt,fill,inner sep=0pt,outer sep=2pt}] at (0,0) {
                 &  \node[color=chocolate](e){};
  \\[-0.2cm]
  \node(e-1){};  &  \node[draw=none,fill=none](dummy){};  
                                   & \node(e-2){}; 
  \\
  \node(e'-1){}; &                 & \node(e'-2){};
  \\
};
\calcLength(e,dummy){mylen}
\draw[<-,very thick,>=latex,chocolate,shorten <=2pt](e) -- ++ (90:{0.55*\mylen pt});
\path(e) ++ ({0.25*\mylen pt},{0.4*\mylen pt}) node{$\astexp$};  
\path(e) ++ ({-0.25*\mylen pt},{0.5*\mylen pt}) node[left]{\Large $\onechartof{\astexp},\;\widehat{\onechartof{\astexp}}$};  

\draw[thick,chocolate] (e) circle (0.12cm);
\draw[->,very thick,shorten <=2pt,shorten >=6pt] (e) to node[left,pos=0.31,xshift={0.015*\mylen pt}]{$\aact$} 
                                                        node[right,pos=0.45,xshift={-0.015*\mylen pt}]{$\darkcyan{\loopnsteplab{2}}$} (e'-1);
\draw[->,very thick,shorten <=2pt,shorten >=6pt] (e) to node[right,pos=0.3,xshift={-0.015*\mylen pt}]{$\bact$} 
                                                        node[left,pos=0.45,xshift={-0.015*\mylen pt}]{$\darkcyan{\loopnsteplab{2}}$} (e'-2);

\path(e-1) ++ ({-0.45*\mylen pt},{0.05*\mylen pt}) node{$\asstexpi{1}$}; 
\draw[->,densely dotted,thick,shorten <=0pt,shorten >=2pt,out=135,in=190,distance={0.85*\mylen pt}] (e-1) to node[left,pos=0.4]{$\sone$} (e);
\draw[->,very thick] (e-1) to node[right,pos=0.3,xshift={-0.05*\mylen pt}]{$\aact$}
                              node[left,pos=0.55,xshift={0.1*\mylen pt}]{$\darkcyan{\loopnsteplab{1}}$} (e'-1);
\draw[->,shorten >=2pt] (e-1) to node[below,pos=0.5,xshift={0.05*\mylen pt}]{$\bact$} (e'-2);

\path(e-2) ++ ({0.45*\mylen pt},{0.05*\mylen pt}) node{$\asstexpi{2}$};  
\draw[->,densely dotted,thick,shorten <=0pt,shorten >=2pt,out=45,in=-10,distance={0.85*\mylen pt}] (e-2) to node[right,pos=0.4]{$\sone$} (e);
\draw[->,very thick] (e-2) to node[left,pos=0.3,xshift={0.05*\mylen pt}]{$\bact$} 
                              node[right,pos=0.55,xshift={-0.1*\mylen pt}]{$\darkcyan{\loopnsteplab{1}}$} (e'-2);

\path(e'-1) ++ ({0*\mylen pt},{-0.35*\mylen pt}) node{$\asstexpacci{1}$};
\draw[->,densely dotted,thick,shorten <=0pt,shorten >=0pt,out=180,in=210,distance={1*\mylen pt}] (e'-1) to node[left,pos=0.4,xshift={0.05*\mylen pt}]{$\sone$} (e-1);  
  
\path(e'-2) ++ ({0*\mylen pt},{-0.35*\mylen pt}) node{$\asstexpacci{2}$};
\draw[->,densely dotted,thick,shorten <=0pt,shorten >=0pt,out=0,in=-30,distance={1*\mylen pt}] (e'-2) to node[right,pos=0.4,xshift={-0.05*\mylen pt}]{$\sone$} (e-2);

\matrix[anchor=north,row sep=1cm,column sep=1.2cm,every node/.style={draw,very thick,circle,minimum width=2.5pt,fill,inner sep=0pt,outer sep=2pt}] at (5,0) {
                 &  \node[color=chocolate](e--ind){};
  \\[-0.2cm]
  \node[color=chocolate](e-1--ind){};  &  \node[draw=none,fill=none](dummy--ind){};  
                                   & \node[color=chocolate](e-2--ind){}; 
  \\
  \node[color=chocolate](e'-1--ind){}; &                 & \node[color=chocolate](e'-2--ind){};
  \\
};
\calcLength(e--ind,dummy--ind){mylen}
\draw[<-,very thick,>=latex,chocolate,shorten <=2pt](e--ind) -- ++ (90:{0.55*\mylen pt});
\path(e--ind) ++ ({0.25*\mylen pt},{0.4*\mylen pt}) node{$\astexp$};  
\path(e--ind) ++ ({-0.3*\mylen pt},{0.4*\mylen pt}) node[left]{\Large $\indscchartof{\onechartof{\astexp}}$};  

\draw[thick,chocolate] (e--ind) circle (0.12cm);
\draw[->,shorten <=2pt,shorten >=6pt] (e--ind) to node[left,pos=0.31,xshift={0.015*\mylen pt}]{$\aact$} (e'-1--ind);
\draw[->,shorten <=2pt,shorten >=6pt] (e--ind) to node[right,pos=0.225,xshift={-0.015*\mylen pt}]{$\bact$} (e'-2--ind);

\draw[thick,chocolate] (e-1--ind) circle (0.12cm);
\path(e-1--ind) ++ ({-0.45*\mylen pt},{0.05*\mylen pt}) node{$\asstexpi{1}$}; 
\draw[->,shorten <=2pt,shorten >= 2pt,densely dashed] (e-1--ind) to 
                                                              (e'-1--ind);
\draw[->,shorten <=2pt,shorten >=2pt,densely dashed,out=-20,in=150] (e-1--ind) to 
                                                              (e'-2--ind);

\draw[thick,chocolate] (e-2--ind) circle (0.12cm);
\path(e-2--ind) ++ ({0.45*\mylen pt},{0.05*\mylen pt}) node{$\asstexpi{2}$};  
\draw[->,densely dashed,shorten <= 2pt,shorten >= 2pt] (e-2--ind) to 
                                                        (e'-2--ind); 
\draw[->,shorten <=2pt,shorten >=8pt,densely dashed,out=200,in=37.5] (e-2--ind) to 
                                                                     (e'-1--ind);

\draw[thick,chocolate] (e'-1--ind) circle (0.12cm);
\path(e'-1--ind) ++ ({0*\mylen pt},{-0.4*\mylen pt}) node{$\asstexpacci{1}$};
\draw[->,out=30,in=150,distance={0.7*\mylen pt},shorten <=2pt,shorten >=2pt] (e'-1--ind) to node[below,pos=0.725,yshift={0.1*\mylen pt}]{$\bact$} (e'-2--ind); 
\draw[->,distance={1*\mylen pt},shorten <=2pt,shorten >=2pt,out=210,in=150] (e'-1--ind) to node[above,pos=0.8]{$\aact$} (e'-1--ind);

\draw[thick,chocolate] (e'-2--ind) circle (0.12cm);   
\path(e'-2--ind) ++ ({0*\mylen pt},{-0.4*\mylen pt}) node{$\asstexpacci{2}$};
\draw[->,bend left,distance={0.7*\mylen pt},shorten <=2pt,shorten >=2pt] (e'-2--ind) to node[above,pos=0.725,yshift={-0.05*\mylen pt}]{$\aact$} (e'-1--ind);
\draw[->,distance={1*\mylen pt},shorten <=2pt,shorten >=2pt,out=-30,in=30] (e'-2--ind) to node[above,pos=0.8]{$\bact$} (e'-2--ind);

\matrix[anchor=north,row sep=1cm,column sep=0.7cm,every node/.style={draw,very thick,circle,minimum width=2.5pt,fill,inner sep=0pt,outer sep=2pt}] at (10,0) {
                 &  \node[color=chocolate](e){};
  \\
  \node[color=chocolate](e-1){};  &  \node[draw=none,fill=none](dummy){};  
                                   & \node[color=chocolate](e-2){}; 
  \\
};
%
\draw[<-,very thick,>=latex,chocolate,shorten <=2pt](e) -- ++ (90:{0.55*\mylen pt});
\path(e) ++ ({0.25*\mylen pt},{0.25*\mylen pt}) node{$\astexp$};  
\path(e) ++ ({0.45*\mylen pt},{0.4*\mylen pt}) node[right] {\Large $\chartof{\astexp}$};

\draw[chocolate,thick] (e) circle (0.12cm);
\draw[->,shorten <=2pt,shorten >=2pt,out=200,in=90] (e) to node[left,pos=0.4]{$\aact$} (e-1);
\draw[->,shorten <=2pt,shorten >=2pt,out=-20,in=90] (e) to node[right,pos=0.225,xshift={0.1*\mylen pt}]{$\bact$} (e-2);

\draw[chocolate,thick] (e-1) circle (0.12cm);
\path(e-1) ++ ({-0.35*\mylen pt},{-0.035*\mylen pt}) node{$\astexpi{1}$}; 
\draw[->,shorten <=2pt,shorten >=2pt,out=220,in=140,distance={1.25*\mylen pt}] (e-1) to node[left]{$\aact$} (e-1);
\draw[->,shorten <=2pt,shorten >=2pt,out=-20,in=200,distance={0.6*\mylen pt}] (e-1) to node[below]{$\bact$} (e-2);

\draw[thick,chocolate] (e-2) circle (0.12cm);
\path(e-2) ++ ({0.35*\mylen pt},{-0.035*\mylen pt}) node{$\astexpi{2}$};  
\draw[->,shorten <=2pt,shorten >=2pt,out=-40,in=40,distance={1.25*\mylen pt}] (e-2) to node[above,yshift={0*\mylen pt},pos=0.75]{$\bact$} (e-2);
\draw[->,shorten <=2pt,shorten >=2pt,out=160,in=20,distance={0.6*\mylen pt}] (e-2) to node[above]{$\aact$} (e-1);

\draw[|->,thick,magenta,densely dashed,shorten <= 3pt,shorten >= 2pt,bend left,distance={1.25*\mylen pt}
                                                                    ] (e--ind) to node[above,pos=0.5,yshift={-0.05*\mylen pt}]{$\sproj$} (e);    
\draw[|->,thick,magenta,densely dashed,shorten <= 3pt,shorten >= 3pt,distance={1.7*\mylen pt},
          out=40,in=130] (e-1--ind) to node[above,pos=0.5,yshift={-0.05*\mylen pt}]{$\sproj$} (e-1);    
\draw[|->,thick,magenta,densely dashed,shorten <= 3pt,shorten >= 2pt,
          out=-30,in=210] (e-2--ind) to node[below,pos=0.32,yshift={0.025*\mylen pt}]{$\sproj$} (e-2);    
\draw[|-,thick,magenta,densely dashed,out=37.5,in=130,shorten <= 3pt,shorten >= 3pt] (e'-1--ind) to node[above,pos=0.425,yshift={-0.05*\mylen pt}]{$\sproj$} (e-1);    
\draw[|-,thick,magenta,densely dashed,out=0,in=210,shorten <= 3pt,shorten >= 2pt] (e'-2--ind) to node[below,pos=0.5,yshift={0.025*\mylen pt}]{$\sproj$} (e-2);

\end{tikzpicture}\vspace*{-1ex}
  \end{center}
  The transitions of the induced chart $\indscchartof{\onechartof{\astexp}}$ of the \onechart\ interpretation~$\onechartof{\astexp}$ of $\astexp$ correspond 
  to paths in $\onechartof{\astexp}$ that start with a (potentially empty) \onetransition\ path and have a final proper action transition, which also provides the label
  of the induced transition.
  For example the $\bact$\nb-tran\-si\-tion from $\asstexpacci{1}$ to $\asstexpacci{2}$ in $\indscchartof{\onechartof{\astexp}}$ arises
  as the induced transition in $\onechartof{\astexp}$ that is the path that consists of the \onetransitions\ from $\asstexpacci{1}$ to $\asstexpi{1}$,
  and from $\asstexpi{1}$ to $\astexp$,
  followed by the final $\bact$\nb-transition from $\astexp$ to $\asstexpacci{2}$.
  The vertices with immediate termination in $\indscchartof{\onechartof{\astexp}}$ 
  are all those that permit \onetransition\ paths in $\onechartof{\astexp}$ to vertices with immediate termination.
  Therefore in $\onechartof{\astexp}$ only $\astexp$ needs to permit immediate termination
  in order to get induced termination in $\indscchartof{\onechartof{\astexp}}$ also at all other vertices (like in $\chartof{\astexp}$).
  Now clearly the projection function $\sproj$ that maps $\astexp \mapsto \astexp$, and
  $\asstexpi{i},\asstexpacci{i} \mapsto \astexpi{i}$ for $i\in\setexp{1,2}$
  defines a bisimulation from $\indscchartof{\onechartof{\astexp}}$ to $\chartof{\astexp}$.
  Provided that the unreachable vertices $\asstexpi{1}$ and $\asstexpi{2}$ in $\indscchartof{\onechartof{\astexp}}$ 
  are removed by garbage collection, 
  $\sproj$ defines an isomorphism.
  What is more, the \entrybodylabeling~$\onecharthatof{\astexp}$ of $\onechartof{\astexp}$ 
  can be readily checked to be a \LLEEwitness\ of $\onechartof{\astexp}$.
  
  Below we illustrate the \onechart\ interpretation~$\onechartof{\bstexp}$ of $\bstexp$ 
  together with the \entrybodylabeling~$\onecharthatof{\bstexp}$ of $\onechartof{\bstexp}$,
  the induced chart $\indscchartof{\onechartof{\bstexp}}$ of $\onechartof{\bstexp}$,
  and its connection via $\sproj$ to the chart interpretation~$\chartof{\bstexp}$ of $\bstexp$.
  The stacked star expressions at \noninitial\ vertices in $\onechartof{\bstexp}$ are
  $\bsstexpi{i} = 
    \stexpprod{(\stexpstackprod{(\stexpprod{\stexpone}{(\stexpsum{\stexpone}{\stexpprod{\bacti{i}}{\stexpzero}})})}{\stexpit{\bstexpi{0}}})}{\stexpzero}$ \mbox{} for $i\in\setexp{1,2,3}$,
  and $\textit{Sink} = \stexpprod{(\stexpstackprod{(\stexpprod{\stexpone}{\stexpzero})}{\stexpit{\bstexpi{0}}})}{\stexpzero}$. 
  Here the projection function $\sproj$ maps $\bstexp \mapsto \bstexp$, $\textit{Sink} \mapsto \textit{sink}$,
  and $\bsstexpi{i} \mapsto \bstexpi{i}$ for $i\in\setexp{1,2,3}$,
  defining a functional bisimulation from $\indscchartof{\onechartof{\bstexp}}$ to $\chartof{\bstexp}$, which is a isomorphism.
  The \entrybodylabeling~$\onecharthatof{\bstexp}$ of $\onechartof{\bstexp}$ 
  can again readily be checked to be a \LLEEwitness\ of $\onechartof{\bstexp}$.
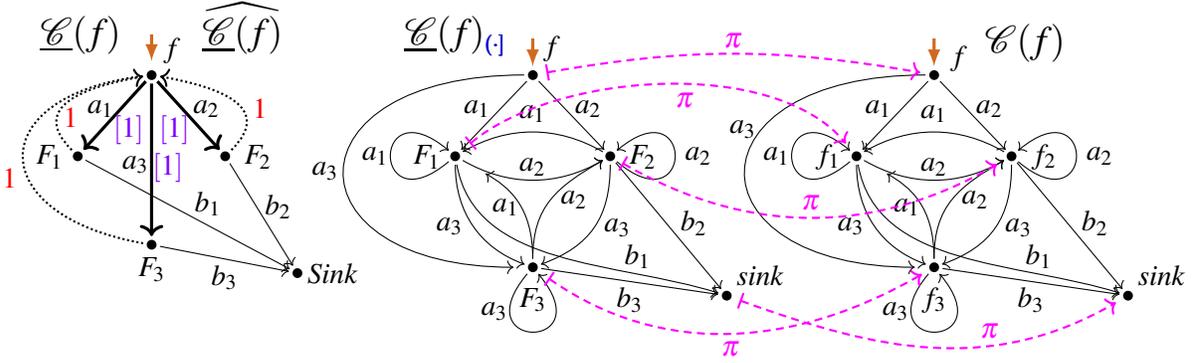
  \begin{flushleft}
    \hspace*{-1.75ex}\begin{tikzpicture}[scale=0.97]

\matrix[anchor=north,row sep=1cm,column sep=0.85cm,every node/.style={draw,very thick,circle,minimum width=2.5pt,fill,inner sep=0pt,outer sep=2pt}] at (0,0) {
                 &  \node(f){};
  \\[-0.1cm]
  \node(F-1){};  &  \node[draw=none,fill=none](dummy){};  
                                   & \node(F-2){}; 
  \\[0cm]
                 &  \node(F-3){};  &                      & \node[draw=none,fill=none](helper){};
  \\
};
\calcLength(f,dummy){mylen}
\path (helper) ++ ({0*\mylen pt},{-0.35*\mylen pt}) node[draw,very thick,circle,minimum width=2.5pt,fill,inner sep=0pt,outer sep=2pt](sink){}; 
\draw[<-,very thick,>=latex,chocolate,shorten <=2pt](f) -- ++ (90:{0.5*\mylen pt});
\path(f) ++ ({0.25*\mylen pt},{0.3*\mylen pt}) node{$\bstexp$};  
\path(f) ++ ({-0.25*\mylen pt},{0.5*\mylen pt}) node[left]{\Large $\onechartof{\bstexp}$};
\path(f) ++ ({0.5*\mylen pt},{0.58*\mylen pt}) node[right]{\Large $\widehat{\onechartof{\bstexp}}$};

\draw[->,very thick] (f) to node[left,pos=0.4,xshift={0.05*\mylen pt}]{$\aacti{1}$} node[right,pos=0.7,xshift={0*\mylen pt}]{$\darkcyan{\loopnsteplab{1}}$} (F-1);
\draw[->,very thick] (f) to node[right,pos=0.4]{$\aacti{2}$} node[left,pos=0.7,xshift={0*\mylen pt}]{$\darkcyan{\loopnsteplab{1}}$} (F-2);
\draw[->,very thick] (f) to node[left,pos=0.53,xshift={0.1*\mylen pt}]{$\aacti{3}$}
                            node[right,pos=0.55,xshift={-0.12*\mylen pt}]{$\darkcyan{\loopnsteplab{1}}$} (F-3);

\path(F-1) ++ ({-0.35*\mylen pt},{-0.01*\mylen pt}) node{$\bsstexpi{1}$}; 
\draw[->] (F-1) to node[above,pos=0.6,yshift={-0*\mylen pt},xshift={0*\mylen pt}]{$\bacti{1}$} (sink);
\draw[->,densely dotted,thick,shorten <=0pt,shorten >=0pt,out=135,in=190,distance={0.75*\mylen pt}] (F-1) to node[pos=0.3,right,xshift={-0.05*\mylen pt}]{$\sone$} (f); 

\path(F-2) ++ ({0.4*\mylen pt},{-0.01*\mylen pt}) node{$\bsstexpi{2}$}; 
\draw[->] (F-2) to node[above,pos=0.7,yshift={0.075*\mylen pt},xshift={0.05*\mylen pt}]{$\bacti{2}$} (sink);
\draw[->,densely dotted,thick,shorten <=0pt,shorten >=0pt,out=45,in=-10,distance={0.75*\mylen pt}] (F-2) to node[right,pos=0.4]{$\sone$} (f); 

\path(F-3) ++ ({0*\mylen pt},{-0.3*\mylen pt}) node{$\bsstexpi{3}$};  
\draw[->] (F-3) to node[below,pos=0.5,yshift={0.05*\mylen pt}]{$\bacti{3}$} (sink);
\draw[-,densely dotted,thick,shorten <=0pt,shorten >=0pt,out=170,in=187,distance={2*\mylen pt}] (F-3) to node[left,pos=0.4]{$\sone$} (f);
  
\path(sink) ++ ({0*\mylen pt},{0*\mylen pt}) node[right]{$\textit{Sink}$};

\matrix[anchor=north,row sep=1cm,column sep=0.9cm,every node/.style={draw,very thick,circle,minimum width=2.5pt,fill,inner sep=0pt,outer sep=2pt}] at (5.25,0) {
                 &  \node(f-ind){};
  \\[-0.1cm]
  \node(F-1-ind){};  &  \node[draw=none,fill=none](dummy-ind){};  
                                   & \node(F-2-ind){}; 
  \\[0.3cm]
                 &  \node(F-3-ind){};  &                      & \node[draw=none,fill=none](helper){};
  \\
};
\calcLength(f-ind,dummy-ind){mylen}
\path (helper) ++ ({0.5*\mylen pt},{-0.35*\mylen pt}) node[draw,very thick,circle,minimum width=2.5pt,fill,inner sep=0pt,outer sep=2pt](sink-ind){}; 
\draw[<-,very thick,>=latex,chocolate](f-ind) -- ++ (90:{0.45*\mylen pt});
\path(f-ind) ++ ({0.23*\mylen pt},{0.32*\mylen pt}) node{$\bstexp$};  
\path(f-ind) ++ 
                ({-0.2*\mylen pt},{0.48*\mylen pt}) node[left] {\Large $\indscchartof{\onechartof{\bstexp}}$};

\draw[->
         ] (f-ind) to node[left,pos=0.4]{$\aacti{1}$} (F-1-ind);
\draw[->
        ] (f-ind) to node[right,pos=0.4]{$\aacti{2}$} (F-2-ind);
\draw[->,out=180,in=180,distance={2.95*\mylen pt},shorten >=4pt] (f-ind) to node[left,pos=0.5,xshift={0.05*\mylen pt}]{$\aacti{3}$} (F-3-ind);

\path(F-1-ind) ++ ({-0.35*\mylen pt},{-0.01*\mylen pt}) node{$\bsstexpi{1}$};  
\draw[->,out=270,in=160] (F-1-ind) to node[left]{$\aacti{3}$} (F-3-ind);
\draw[->,out=220,in=140,distance={1.25*\mylen pt}]  (F-1-ind) to node[left,xshift={0.1*\mylen pt}]{$\aacti{1}$} (F-1-ind);
\draw[->,out=-30,in=210,distance={0.65*\mylen pt}]  (F-1-ind) to node[above,yshift={-0.065*\mylen pt}]{$\aacti{2}$} (F-2-ind);
\draw[->,out=-80,in=162.5,distance={1.3*\mylen pt}] (F-1-ind) to node[above,pos=0.75,yshift={-0.05*\mylen pt},xshift={0*\mylen pt}]{$\bacti{1}$} (sink-ind);

\path(F-2-ind) ++ ({0.4*\mylen pt},{-0.01*\mylen pt}) node{$\bsstexpi{2}$};  
\draw[->,out=270,in=20,distance={0.65*\mylen pt}] (F-2-ind) to node[right]{$\aacti{3}$} (F-3-ind);
\draw[->,out=-40,in=40,distance={1.25*\mylen pt}] (F-2-ind) to node[right]{$\aacti{2}$} (F-2-ind);
\draw[->,out=150,in=30,distance={0.65*\mylen pt}] (F-2-ind) to node[above]{$\aacti{1}$} (F-1-ind);
\draw[->] (F-2-ind) to node[above,pos=0.7,yshift={0.075*\mylen pt},xshift={0.05*\mylen pt}]{$\bacti{2}$} (sink-ind);

\path(F-3-ind) ++ ({0*\mylen pt},{-0.4*\mylen pt}) node{$\bsstexpi{3}$};  
\draw[->,out=230,in=310,distance={1.25*\mylen pt}] (F-3-ind) to node[left,pos=0.2,xshift={0.1*\mylen pt}]{$\aacti{3}$} (F-3-ind);
\draw[->,out=90,in=-30,distance={0.65*\mylen pt},shorten >=10pt] (F-3-ind) to node[pos=0.42,left,xshift={0.11*\mylen pt}]{$\aacti{1}$} (F-1-ind);
\draw[->,out=90,in=210,distance={0.65*\mylen pt}] (F-3-ind) to node[right,pos=0.5]{$\aacti{2}$} (F-2-ind);
\draw[->] (F-3-ind) to node[below,pos=0.5,yshift={0.05*\mylen pt}]{$\bacti{3}$} (sink-ind);

\path(sink-ind) ++ ({0*\mylen pt},{0.25*\mylen pt}) node[right]{$\textit{sink}$};

\matrix[anchor=north,row sep=1cm,column sep=0.9cm,every node/.style={draw,very thick,circle,minimum width=2.5pt,fill,inner sep=0pt,outer sep=2pt}] at (10.75,0) {
                 &  \node(f){};
  \\[-0.1cm]
  \node(f-1){};  &  \node[draw=none,fill=none](dummy){};  
                                   & \node(f-2){}; 
  \\[0.3cm]
                 &  \node(f-3){};  &                      & \node[draw=none,fill=none](helper){};
  \\
};
\calcLength(f,dummy){mylen}
\path (helper) ++ ({0.5*\mylen pt},{-0.35*\mylen pt}) node[draw,very thick,circle,minimum width=2.5pt,fill,inner sep=0pt,outer sep=2pt](sink){}; 
\draw[<-,very thick,>=latex,chocolate](f) -- ++ (90:{0.45*\mylen pt});
\path(f) ++ ({0.3*\mylen pt},{0.2*\mylen pt}) node{$\bstexp$};  
\path(f) ++ 
            ({0.5*\mylen pt},{0.4*\mylen pt}) node[right] {\Large $\chartof{\bstexp}$};

\draw[->
         ] (f) to node[left,pos=0.4]{$\aacti{1}$} (f-1);
\draw[->
        ] (f) to node[right,pos=0.4]{$\aacti{2}$} (f-2);
\draw[->,out=180,in=180,distance={2.95*\mylen pt},shorten >=4pt] (f) to node[left,pos=0.35,xshift={0.05*\mylen pt}]{$\aacti{3}$} (f-3);

\path(f-1) ++ ({-0.35*\mylen pt},{-0.01*\mylen pt}) node{$\bstexpi{1}$};  
\draw[->,out=270,in=160] (f-1) to node[left]{$\aacti{3}$} (f-3);
\draw[->,out=220,in=140,distance={1.25*\mylen pt}]  (f-1) to node[left,xshift={0.1*\mylen pt}]{$\aacti{1}$} (f-1);
\draw[->,out=-30,in=210,distance={0.65*\mylen pt}]  (f-1) to node[above,yshift={-0.065*\mylen pt}]{$\aacti{2}$} (f-2);
\draw[->,out=-80,in=162.5,distance={1.3*\mylen pt}] (f-1) to node[above,pos=0.75,yshift={-0.05*\mylen pt},xshift={0*\mylen pt}]{$\bacti{1}$} (sink);

\path(f-2) ++ ({0.4*\mylen pt},{-0.01*\mylen pt}) node{$\bstexpi{2}$};  
\draw[->,out=270,in=20,distance={0.65*\mylen pt}] (f-2) to node[right]{$\aacti{3}$} (f-3);
\draw[->,out=-40,in=40,distance={1.25*\mylen pt}] (f-2) to node[right]{$\aacti{2}$} (f-2);
\draw[->,out=150,in=30,distance={0.65*\mylen pt}] (f-2) to node[above]{$\aacti{1}$} (f-1);
\draw[->] (f-2) to node[above,pos=0.7,yshift={0.075*\mylen pt},xshift={0.05*\mylen pt}]{$\bacti{2}$} (sink);

\path(f-3) ++ ({0*\mylen pt},{-0.4*\mylen pt}) node{$\bstexpi{3}$};  
\draw[->,out=230,in=310,distance={1.25*\mylen pt}] (f-3) to node[left,pos=0.2,xshift={0.1*\mylen pt}]{$\aacti{3}$} (f-3);
\draw[->,out=90,in=-30,distance={0.65*\mylen pt},shorten >=10pt] (f-3) to node[pos=0.42,left,xshift={0.11*\mylen pt}]{$\aacti{1}$} (f-1);
\draw[->,out=90,in=210,distance={0.65*\mylen pt}] (f-3) to node[right,pos=0.5]{$\aacti{2}$} (f-2);
\draw[->] (f-3) to node[below,pos=0.5,yshift={0.05*\mylen pt}]{$\bacti{3}$} (sink);

\path(sink) ++ ({0*\mylen pt},{0.25*\mylen pt}) node[right]{$\textit{sink}$};

\draw[|->,thick,magenta,densely dashed,shorten <= 2pt,shorten >= 2pt,out=10,in=170] (f-ind) to node[above,pos=0.5,yshift={-0.05*\mylen pt}]{$\sproj$} (f);    
\draw[|->,thick,magenta,densely dashed,shorten <= 4pt,shorten >= 3pt,distance={1.5*\mylen pt},
          out=40,in=130] (F-1-ind) to node[below,pos=0.55,yshift={0.05*\mylen pt}]{$\sproj$} (f-1);    
\draw[|->,thick,magenta,densely dashed,shorten <= 2pt,shorten >= 3pt,
          out=-30,in=210] (F-2-ind) to node[above,pos=0.5,yshift={-0.05*\mylen pt}]{$\sproj$} (f-2);   
\draw[|->,thick,magenta,densely dashed,shorten <= 4pt,shorten >= 2pt,
          out=-35,in=210] (F-3-ind) to node[below,pos=0.5,yshift={0.025*\mylen pt}]{$\sproj$} (f-3);    
\draw[|->,thick,magenta,densely dashed,shorten <= 2pt,shorten >= 2pt,
          out=-20,in=210] (sink-ind) to node[above,pos=0.655,yshift={-0.05*\mylen pt}]{$\sproj$} (sink);    

\end{tikzpicture}\vspace*{-1ex}
  \end{flushleft}\vspace{-1ex}
\end{exa}
     
Thus we have verified the joint claims of Theorem~\ref{thm:onechart-int:funbisim:chart-int} and of Theorem~\ref{thm:onechart-int:LLEEw}  
for the two examples $\astexp$ and $\bstexp$ of star expressions from Ex.~\ref{ex:chart:interpretation}
for which, as we saw in Section~\ref{LLEE:fail}, \LEE\ fails for their chart interpretations $\chartof{\astexp}$ and $\chartof{\bstexp}$:
the property \LEE\ can be recovered for the \onechart\ interpretations $\onechartof{\astexp}$ of~$\astexp$, and $\onechartof{\bstexp}$ of~$\bstexp$
(by Theorem~\ref{thm:onechart-int:LLEEw}),
and also, the induced charts $\indscchartof{\onechartof{\astexp}}$ of $\onechartof{\astexp}$, and $\indscchartof{\onechartof{\bstexp}}$ of $\onechartof{\bstexp}$
map to the original process interpretations $\chartof{\astexp}$ of $\astexp$, and $\chartof{\bstexp}$ of $\bstexp$, respectively,
via a functional bisimulation (by Theorem~\ref{thm:onechart-int:funbisim:chart-int}). 

\enlargethispage{4ex}

\section{Proofs}%
  \label{proofs}

In this section we give the proofs of the properties~\ref{property:1} and \ref{property:2} (see in the Introduction)
of the variant process semantics $\onechartof{\cdot}$ in relation to the process semantics~$\chartof{\cdot}$. 
In Section~\ref{property:1::proofs} we prove Lemma~\ref{lem:StExpLTS:funbisim:StExpindLTS}, and Theorem~\ref{thm:onechart-int:funbisim:chart-int}.
In this way we demonstrate property~\ref{property:1}.
Then in Section~\ref{property:2::proofs} 
we provide the details of the proofs of Lemma~\ref{lem:oneLTShat:stackStExps:is:LLEEw}, 
and of Theorem~\ref{thm:onechart-int:LLEEw}, thereby demonstrating property~\ref{property:2}. 

\subsection{Proof of property~\protect\ref{property:1} of $\protect\onechartof{\protect\cdot}$ and $\protect{\protect\chartof{\cdot}}$}
  \label{property:1::proofs}
%


We first develop the proof of Lemma~\ref{lem:StExpLTS:funbisim:StExpindLTS},
and then prove Theorem~\ref{thm:onechart-int:funbisim:chart-int} from it. In doing so we will then have demonstrated property~\ref{property:1}.

Lemma~\ref{lem:StExpLTS:funbisim:StExpindLTS} states that the projection function $\sproj \funin \stackStExpover{\actions} \to \StExpover{\actions}$ 
defines a bisimulation between the induced \LTS\ $\indLTSof{\oneLTSof{\stackStExpover{\actions}}}$ 
  of the stacked star expressions \oneLTS~$\oneLTSof{\stackStExpover{\actions}}$ 
  and the star expressions \LTS~$\LTSof{\StExpover{\actions}}$.
We will show this by using statements (see Lemma~\ref{lem:stackStExpTSS:proj:StExpTSS}, Lemma~\ref{lem:1:lem:indTSS:stackStExpTSS:2:StExpTSS}, and crucially, Lemma~\ref{lem:indTSS:stackStExpTSS:2:StExpTSS})
  that relate derivability in the \TSS~$\StExpTSSover{\actions}$ that defines transitions in $\LTSof{\StExpover{\actions}}$
to derivability in a \TSS\ $\stackStExpindTSSover{\actions}$ that defines transitions in the induced \LTS\ $\indLTSof{\oneLTSof{\stackStExpover{\actions}}}$.

For this purpose we define the TSS $\stackStExpindTSSover{\actions}$ as an extension of $\stackStExpTSSover{\actions}$
by rules\vspace*{-1pt} that produce induced transitions and induced termination
(see Definition~\ref{def:indTSS:stackStExpTSS}).
Then we show that $\stackStExpindTSSover{\actions}$ does in fact generate 
  the induced $\indLTSof{\oneLTSof{\stackStExpover{\actions}}}$ of $\oneLTSof{\stackStExpover{\actions}}$
(see Lemma~\ref{lem:indLTS:oneLTS:stackStExps}, which will be shown via Lemma~\ref{lem:prf:lem:indLTS:oneLTS:stackStExps}). 


\begin{defi}\label{def:indTSS:stackStExpTSS}\nf
  The LTS $\LTSdefdby{\stackStExpindTSSover{\actions}} = \tuple{\stackStExpover{\actions},\actions,\soneterminates,\silt{\cdot}}$ 
  is generated by derivations in the TSS~$\stackStExpindTSSover{\actions}$\vspace*{-2pt}
  that in addition to the axioms and rules of $\stackStExpTSSover{\actions}$ also contains the following rules:\vspace*{-1ex}
  \begin{center}
  $
  \begin{aligned}
    &
    \AxiomC{$ \terminates{\asstexp}\rule[-4pt]{0pt}{15pt} $}
    \UnaryInfC{$ \oneterminates{\asstexp}\rule[-4pt]{0pt}{11.5pt} $}
    \DisplayProof
    & \hspace*{1ex} &
    \AxiomC{$ \asstexp \lt{\sone} \asstexptilde $}
    \AxiomC{$ \oneterminates{\asstexptilde} $}
    \BinaryInfC{$ \oneterminates{\asstexp} $}
    \DisplayProof
    & & \hspace*{2ex} & &
    \AxiomC{$ \asstexp \lt{\aact} \asstexpacc \rule[-3pt]{0pt}{19.5pt} $}
    \UnaryInfC{$ \asstexp \ilt{\aact} \asstexpacc $}
    \DisplayProof
    & \hspace*{1ex} &
    \AxiomC{$ \asstexp \lt{\sone} \asstexptilde \rule[-3pt]{0pt}{19.5pt}$}
    \AxiomC{$ \asstexptilde \ilt{\aact} \asstexpacc \rule[-3pt]{0pt}{19.5pt}$}
    \BinaryInfC{$ \asstexp \ilt{\aact} \asstexpacc $}
    \DisplayProof
  \end{aligned}
  $
\end{center}
\end{defi}

\begin{lem}\label{lem:prf:lem:indLTS:oneLTS:stackStExps}
  For  all $\asstexp,\asstexpacc\in\stackStExpover{\actions}$,
  and $\aact\in\actions$, 
  provability in $\stackStExpindTSS$ of induced termination statements\vspace*{-3.5pt} $\oneterminates{\asstexp}$,
  and of induced transition statements $\asstexp \ilt{\aact} \asstexpacc$
  can be characterized as follows:\vspace{-0.75ex}
  \begin{enumerate}[label={(\roman{*})},leftmargin=*,align=right,itemsep=0ex]
    \item{}\label{stmt:1:prf:lem:indLTS:oneLTS:stackStExps}
      $\derivablein{\stackStExpindTSS} \oneterminates{\asstexp}$
        if and only if
      $\asstexp = \asstexptildei{0}$,   
      $\derivablein{\stackStExpTSS} \asstexptildei{0} \lt{\sone} \asstexptildei{1}$,
      $\derivablein{\stackStExpTSS} \asstexptildei{1} \lt{\sone} \asstexptildei{2}$,
      \ldots,
      $\derivablein{\stackStExpTSS} \asstexptildei{n-1} \lt{\sone} \asstexptildei{n}$,
      and
      $\derivablein{\stackStExpTSS} \terminates{\asstexptildei{n}}$,
      for\vspace*{-2pt} some $n\in\nat$, and $\asstexptildei{0},\asstexptildei{1},\ldots,\asstexptildei{n}\in\stackStExpover{\actions}$. 
      
    \item{}\label{stmt:2:prf:lem:indLTS:oneLTS:stackStExps}
      $\derivablein{\stackStExpindTSS} \asstexp \ilt{\aact} \asstexpacc$
        if and only if
      $\asstexp = \asstexptildei{0}$, 
      $\derivablein{\stackStExpTSS} \asstexptildei{0} \lt{\sone} \asstexptildei{1}$,
      $\derivablein{\stackStExpTSS} \asstexptildei{1} \lt{\sone} \asstexptildei{2}$,
      \ldots,
      $\derivablein{\stackStExpTSS} \asstexptildei{n-1} \lt{\sone} \asstexptildei{n}$,
      and
      $\derivablein{\stackStExpTSS} \asstexptildei{n} \lt{\aact} \asstexpacc$,\vspace*{-2pt}
      for some $n\in\nat$, and $\asstexptildei{0},\asstexptildei{1},\ldots,\asstexptildei{n}\in\stackStExpover{\actions}$.  
  \end{enumerate}
\end{lem}
  
\begin{proof}  
  The directions ``$\Rightarrow$'' in \ref{stmt:1:prf:lem:indLTS:oneLTS:stackStExps} and \ref{stmt:2:prf:lem:indLTS:oneLTS:stackStExps}
  can be established by straightforward induction on derivations in $\stackStExpindTSSover{\actions}$ with bottommost applications of rules from Def.~\ref{def:indTSS:stackStExpTSS}.
  The directions ``$\Leftarrow$'' by can be shown by a straightforward proof by inductions~on~$n\in\nat$. 
\end{proof}

Now we can show that the \LTS\ that is generated by $\stackStExpTSSover{\actions}$ coincides with $\indLTSof{\oneLTSof{\stackStExpover{\actions}}}$.

\begin{lem}\label{lem:indLTS:oneLTS:stackStExps}
  $\indLTSof{\oneLTSof{\stackStExpover{\actions}}}
     =
   \LTSdefdby{\stackStExpTSSover{\actions}}\,$. 
\end{lem}

\begin{proof}
  The statements \ref{stmt:1:prf:lem:indLTS:oneLTS:stackStExps} and \ref{stmt:2:prf:lem:indLTS:oneLTS:stackStExps} in Lemma~\ref{lem:prf:lem:indLTS:oneLTS:stackStExps}
  together demonstrate that 
  1-termination $\soneterminates$ defined by\vspace*{-2pt} $\stackStExpindTSSover{\actions}$ 
  coincides with induced termination of the \oneLTS\ defined by $\stackStExpTSSover{\actions}$, 
  and transitions $\silt{\aact}$ defined by $\stackStExpindTSSover{\actions}$, for $\aact\in\actions$,
  coincide with induced transitions of the \oneLTS\ defined by $\stackStExpTSSover{\actions}$.
\end{proof}

In order to transform derivations in $\StExpTSSover{\actions}$ into derivations in $\stackStExpindTSSover{\actions}$ in Lemma~\ref{lem:indTSS:stackStExpTSS:2:StExpTSS} below,
it will be expedient to have a number of additional rules available 
that are \emph{admissible for $\stackStExpindTSS$},
that is, their instances can be eliminated from derivations
that contain them in addition to axioms and rules of $\stackStExpindTSS$.

\begin{lem}\label{lem:2:lem:indTSS:stackStExpTSS:2:StExpTSS}
  The following rules are admissible for $\stackStExpindTSS$
  :\vspace*{-2.5ex}
  \begin{center}
    $  
    \begin{aligned}
      &
      \AxiomC{$ \oneterminates{\asstexpi{1}} $}
      \AxiomC{$ \terminates{\astexpi{2}} $}
      \insertBetweenHyps{\hspace*{0.5ex}}
      \BinaryInfC{$ \oneterminates{ (\stexpprod{\asstexpi{1}}{\astexpi{2}}) } $}
      \DisplayProof
      & \hspace*{-0.8ex} &   
      \AxiomC{$ \oneterminates{\asstexpi{1}} $}
      \UnaryInfC{$ \oneterminates{ (\stexpstackprod{\asstexpi{1}}{\stexpit{\astexpi{2}}}) } $}
      \DisplayProof
      & \hspace*{-0.8ex} &
      \AxiomC{$ \asstexpi{1} \ilt{\aact} \asstexpacci{1} $}
      \UnaryInfC{$ \stexpprod{\asstexpi{1}}{\astexpi{2}} \ilt{\aact} \stexpprod{\asstexpacci{1}}{\astexpi{2}} $}
      \DisplayProof
      & \hspace*{-0.8ex} &
      \AxiomC{$ \asstexpi{1} \ilt{\aact} \asstexpacci{1} $}
      \UnaryInfC{$ \stexpstackprod{\asstexpi{1}}{\stexpit{\astexpi{2}}} \ilt{\aact} \stexpstackprod{\asstexpacci{1}}{\stexpit{\astexpi{2}}} $}
      \DisplayProof
      & \hspace*{-0.8ex} &   
      \AxiomC{$ \oneterminates{\asstexpi{1}} $}
      \AxiomC{$ \astexpi{2} \ilt{\aact} \asstexpacci{2} $}
      \insertBetweenHyps{\hspace*{0.5ex}}
      \BinaryInfC{$ \stexpprod{\asstexpi{1}}{\astexpi{2}} \ilt{\aact} \asstexpacci{2} $}
      \DisplayProof
      & \hspace*{-0.8ex} &  
      \AxiomC{$ \oneterminates{\asstexpi{1}} $}
      \AxiomC{$ \stexpit{\astexpi{2}} \ilt{\aact} \asstexpacci{2} $}
      \insertBetweenHyps{\hspace*{0.5ex}}
      \BinaryInfC{$ \stexpstackprod{\asstexpi{1}}{\stexpit{\astexpi{2}}} \ilt{\aact} \asstexpacci{2} $}
      \DisplayProof
    \end{aligned}
    $
  \end{center}
\end{lem}

\begin{proof}
  Admissibility of the six rules in the lemma 
  can be established by straightforward proof-theoretic arguments
  that describe how instances of any of these rules 
  can be eliminated effectively provided that they have immediate subderivations in $\stackStExpindTSS$. 
  More specifically, instances of these rules with immediate subderivations in $\stackStExpindTSS$
    can be permuted upwards to get closer to axioms,
    and they can eventually be removed when they are close enough to axioms.
  
  As an example we demonstrate that for the sixth rule in Lemma~\ref{lem:2:lem:indTSS:stackStExpTSS:2:StExpTSS}, which we denote here by $\arulei{6}$.
  This rule concerns a slightly more involved case. 
  We first note that no \onetransitions\ are possible from\vspace*{-2pt} star expressions like $\stexpit{\astexpi{2}}$ in $\stackStExpTSS$ and $\stackStExpindTSS$,
  and that
  therefore the right premise $\stexpit{\astexpi{2}} \ilt{\aact} \asstexpacci{2}$ in $\arulei{6}$\vspace*{-2pt}
  can only be derived from $\stexpit{\astexpi{2}} \lt{\aact} \asstexpacci{2}$.
  Consequently it suffices to show that the rule $\aruleacci{6}$ that results from $\arulei{6}$ 
  by using\vspace*{-2pt} the stronger form $\stexpit{\astexpi{2}} \lt{\aact} \asstexpacci{2}$ of its right premise
  is admissible in $\stackStExpindTSS$. 
  For this we need to show that every instance of $\aruleacci{6}$
  can be eliminated from the bottom of every derivation $\thplus{\stackStExpindTSS}{\arulei{6}}$ of the form
  as below left:
  \begin{center}
    $
    \aDeriv \, \left\{ \quad
    \begin{aligned}
      \AxiomC{$ \aDerivi{1} $}
      \noLine
      \UnaryInfC{$ \oneterminates{\asstexpi{1}} $}
      \AxiomC{$ \aDerivi{2} $}
      \noLine
      \UnaryInfC{$ \stexpit{\astexpi{2}} \lt{\aact} \asstexpacci{2} $}
      \RightLabel{$\aruleacci{6}$}
      \BinaryInfC{$ \stexpstackprod{\asstexpi{1}}{\stexpit{\astexpi{2}}} \ilt{\aact} \astexpacci{2} $}
      \DisplayProof
    \end{aligned}
    \quad\right.
    \quad\Longrightarrow\qquad
    \begin{aligned}
      \AxiomC{$\aDerivacc$}
      \noLine
      \UnaryInfC{$ \stexpstackprod{\asstexpi{1}}{\stexpit{\astexpi{2}}} \ilt{\aact} \astexpacci{2} $}
      \DisplayProof
    \end{aligned}
    $
  \end{center}
  where the immediate subderivations $\aDerivi{1}$ and $\aDerivi{2}$ of $\aDeriv$ are derivations in $\stackStExpindTSS$,
  with as result a derivation $\aDerivacc$ in $\stackStExpindTSS$ with the same conclusion as $\aDeriv$, as indicated above on the right. 
  We proceed by induction on the depth $\depth{\aDerivi{1}}$ of the left immediate subderivation $\aDerivi{1}$ of $\aDeriv$.
  
  The case $\depth{\aDerivi{1}} = 0$ is not possible, since there is no axiom of $\stackStExpindTSS$ with induced termination $\soneterminates$.
  
  If $\depth{\aDerivi{1}} = 1$, then $\oneterminates{\asstexpi{1}}$ must be derived from an axiom of $\stackStExpTSS$ and $\stackStExpindTSS$ of the form $\terminates{\asstexpi{1}}$. 
  Then $\aDeriv$ is of the form as on the left below, and it can be transformed into a derivation $\aDerivacc$ in $\stackStExpindTSS$
  as on the right:
  \begin{center}
    $
    \aDeriv \left\{\, \quad
    \begin{aligned}
      \AxiomC{}
      \UnaryInfC{$ \terminates{\asstexpi{1}} $}
      \UnaryInfC{$ \oneterminates{\asstexpi{1}} $}
      \AxiomC{$ \aDerivi{2} $}
      \noLine
      \UnaryInfC{$ \stexpit{\astexpi{2}} \lt{\aact} \asstexpacci{2} $}
      \RightLabel{$\aruleacci{6}$}
      \BinaryInfC{$ \stexpstackprod{\asstexpi{1}}{\stexpit{\astexpi{2}}} \ilt{\aact} \astexpacci{2} $}
      \DisplayProof
    \end{aligned}
    \quad \right.
    \quad\Longrightarrow\quad
    \quad \left.
    \begin{aligned}
      \AxiomC{}
      \UnaryInfC{$ \terminates{\asstexpi{1}} $}
      \UnaryInfC{$ \stexpstackprod{\asstexpi{1}}{\stexpit{\astexpi{2}}} \lt{\sone} \stexpit{\astexpi{2}} $}
      \AxiomC{$ \aDerivi{2} $}
      \noLine
      \UnaryInfC{$ \stexpit{\astexpi{2}} \lt{\aact} \asstexpacci{2} $}
      \UnaryInfC{$ \stexpit{\astexpi{2}} \ilt{\aact} \asstexpacci{2} $}
      \BinaryInfC{$ \stexpstackprod{\asstexpi{1}}{\stexpit{\astexpi{2}}} \ilt{\aact} \asstexpacci{2} $} 
      \DisplayProof
    \end{aligned}
    \quad \,\right\} \aDerivacc
    $
  \end{center}
  In this way we have transformed $\aDeriv$ into a derivation $\aDerivacc$ in $\stackStExpindTSS$ with the same conclusion.
 
  If $\depth{\aDerivi{1}} > 1$, then $\oneterminates{\asstexpi{1}}$ must be derived from 
  the premises $\asstexpi{1} \lt{\sone} \asstexptildei{1}$ and $\oneterminates{\asstexptildei{1}}$, neither of which can be an axiom of $\stackStExpTSS$ and $\stackStExpindTSS$. 
  In this case $\aDeriv$ is of the form as on the left below. Then $\aDeriv$ can be transformed,\vspace*{-2pt}
  by permuting the instance of $\aruleacci{6}$ upwards, into a derivation $\aDerivtildeacc$ in $\thplus{\stackStExpTSS}{\aruleacci{6}}\,$:
  \begin{equation*}
    \aDeriv \left\{\, \;
    \begin{aligned}
      \AxiomC{$ \aDerivi{11} $}
      \noLine
      \UnaryInfC{$ \asstexpi{1} \lt{\sone} \asstexptildei{1} $}
      \AxiomC{$ \aDerivi{12} $}
      \noLine
      \UnaryInfC{$ \oneterminates{\asstexptildei{1}} $}
      \BinaryInfC{$ \oneterminates{\asstexpi{1}} $}
      \AxiomC{$ \aDerivi{2} $}
      \noLine
      \UnaryInfC{$ \stexpit{\astexpi{2}} \lt{\aact} \asstexpacci{2} $}
      \RightLabel{$\aruleacci{6}$}
      \BinaryInfC{$ \stexpstackprod{\asstexpi{1}}{\stexpit{\astexpi{2}}} \ilt{\aact} \asstexpacci{2} $}
      \DisplayProof
    \end{aligned}
    \;\,\right. 
    \Longrightarrow\;\;
    \left.\,\;
    \begin{aligned}
      \AxiomC{$ \aDerivi{11} $}
      \noLine
      \UnaryInfC{$ \asstexpi{1} \lt{\sone} \asstexptildei{1} $}
      \UnaryInfC{$ \stexpprod{\asstexpi{1}}{\stexpit{\astexpi{2}}}  \lt{\sone}  \stexpstackprod{\asstexptildei{1}}{\stexpit{\astexpi{2}}} $}
      \AxiomC{$ \aDerivi{12} $}
      \noLine
      \UnaryInfC{$ \oneterminates{\asstexptildei{1}} $}
      \AxiomC{$ \aDerivi{2} $}
      \noLine
      \UnaryInfC{$ \stexpit{\astexpi{2}} \lt{\aact} \asstexpacci{2} $}
      %
      \RightLabel{$\aruleacci{6}$}
      \BinaryInfC{$ \stexpstackprod{\asstexptildei{1}}{\stexpit{\astexpi{2}}} \ilt{\aact} \asstexpacci{2} $}
      \BinaryInfC{$ \stexpstackprod{\asstexpi{1}}{\stexpit{\astexpi{2}}} \ilt{\aact} \asstexpacci{2} $}
      \DisplayProof
    \end{aligned}
    \;\,\right\} \aDerivtildeacc
  \end{equation*}
  Note that $\aDerivtildeacc$ still contains an occurrence of $\aruleacci{6}$ that we need to eliminate in order to obtain a derivation in $\stackStExpindTSS$.
  Now since $\depth{\aDerivi{12}} < \depth{\aDerivi{1}}$ holds because $\aDerivi{12}$ is a subderivation of $\aDerivi{1}$, 
  we can apply the induction hypothesis to the subderivation of the instance of $\aruleacci{6}$ in $\aDerivtildeacc$.
  We denote\vspace*{-2pt} that subderivation of $\aDerivtildeacc$ by $\aDerivtildeacci{1}$. 
  We get a derivation $\aDerivacci{1}$ in $\stackStExpindTSS$
    with the same conclusion $\stexpstackprod{\asstexptildei{1}}{\stexpit{\astexpi{2}}} \ilt{\aact} \asstexpacci{2}$ as $\aDerivtildeacci{1}$. 
  By replacing\vspace*{-2pt} $\aDerivtildeacci{1}$ in $\aDerivtildeacc$
    with the derivation $\aDerivacci{1}$
  we obtain a derivation $\aDerivacc$ in $\stackStExpindTSS$ 
  with conclusion $\stexpstackprod{\asstexpi{1}}{\stexpit{\astexpi{2}}} \ilt{\aact} \asstexpacci{2}$.
  
  In this way we have completed the proof by induction that the rule $\aruleacci{6}$ as on the bottom of a derivation $\aDeriv$ 
  with immediate subderivations in $\stackStExpindTSS$ can always be eliminated effectively.   
  This shows admissibility of $\aruleacci{6}$ in $\stackStExpindTSS$. As we have argued above that admissibility of $\aruleacci{6}$ in $\stackStExpindTSS$ 
  implies admissibility in $\stackStExpindTSS$ of $\arulei{6}$,
  we have now also established admissibility of the sixth (and rightmost) rule in Lemma~\ref{lem:2:lem:indTSS:stackStExpTSS:2:StExpTSS}.
  
  \smallskip
  Admissibility of the fifth (rightmost but one) rule in Lemma~\ref{lem:2:lem:indTSS:stackStExpTSS:2:StExpTSS}
  can be demonstrated analogously as above.
  Admissibility of the second, the third, and the forth rule in Lemma~\ref{lem:2:lem:indTSS:stackStExpTSS:2:StExpTSS}
  can be shown by similar upward-permutation transformations by induction on the depth of the immediate subderivations of instances of these rules.
  Finally, admissibility of the first (leftmost) rule in Lemma~\ref{lem:2:lem:indTSS:stackStExpTSS:2:StExpTSS}
  can also be shown by upward-permutation of this rule, using induction on the depth of the immediate subderivation of its left premise $\oneterminates{\asstexpi{1}}$. 
\end{proof}

\medskip

Now we turn our attention to the relationship between the \LTSs\
  $\oneLTSof{\stackStExpover{\actions}}$ and $\LTSof{\StExpover{\actions}}$,
which are generated by the \TSSs~$\stackStExpTSSover{\actions}$, and $\StExpTSSover{\actions}$, respectively.
We formulate two lemmas 
that concern the projection under the projection function $\sproj$ of derivations in the TSS~$\stackStExpTSSover{\actions}$
to derivations in the TSS~$\StExpTSSover{\actions}$. 
The first lemma, Lemma~\ref{lem:stackStExpTSS:proj:StExpTSS},
concerns statements with immediate termination or a proper transition in the conclusion,
and the second lemma, Lemma~\ref{lem:1:lem:indTSS:stackStExpTSS:2:StExpTSS},
concerns also statements with \onetransitions\ in the conclusion.

\begin{lem}\label{lem:stackStExpTSS:proj:StExpTSS}
  Derivations in the TSS~$\stackStExpTSS$ project to derivations in the TSS~$\StExpTSS$, as follows:\vspace*{0ex} 
  \begin{enumerate}[label={(\roman{*})},itemsep=0ex]
    \item{}\label{it:1:lem:stackStExpTSS:proj:StExpTSS}   
      Every derivation $\aDeriv$ in the TSS~$\stackStExpTSS$ with conclusion $\terminates{\asstexp}$
      is also a derivation in the TSS~$\StExpTSS$,
      and consequently $\asstexp$ and all other occurring stacked star expressions are star expressions in $\StExpover{\actions}$. 
      
    \item{}\label{it:2:lem:stackStExpTSS:proj:StExpTSS}   
      Every derivation $\aDeriv$ in the TSS~$\stackStExpTSS$ with conclusion $\asstexp \lt{a} \asstexpacc$
      projects, via the projection function $\sproj$ (applied to all stacked star expressions in the derivation),
      to a derivation $\proj{\aDeriv}$ in the TSS~$\StExpTSS$ with conclusion $\proj{\asstexp} \lt{a} \proj{\asstexpacc}$. 
    
%
%
  \end{enumerate}      
\end{lem}

\begin{proof}
  For statement~\ref{it:1:lem:stackStExpTSS:proj:StExpTSS} 
  we notice, by inspecting the rules of $\stackStExpTSS$,
  that every instance of an axiom of the TSS~$\stackStExpTSS$ with immediate termination,
  and every instance of a rule of $\stackStExpTSS$ with immediate termination in its conclusion
  contains only star expressions, and is also an axiom or a rule instance, respectively,
  of the TSS~$\StExpTSS$. This entails that every derivation in $\stackStExpTSS$ with immediate termination in the conclusion
  consists only of star expressions, and is also a derivation with the same conclusion in $\StExpTSS$.
  
  For statement~\ref{it:2:lem:stackStExpTSS:proj:StExpTSS}
  we notice analogously that  
  every instance of an axiom or of a rule of the TSS~$\stackStExpTSS$ that does not have an occurrence of a \onetransition\
  projects, via the projection function $\sproj$ (applied to all stacked star expressions in the instance),
  to an instance of an axiom or a rule, respectively, of the TSS~$\StExpTSS$.
  As an example we consider an instance of the rule of $\stackStExpTSS$ that permits to preserve a transition
  by putting it into a context $\acxt = \stexpstackprod{\Box}{\stexpit{\astexpi{2}}}$ involving the stacked product symbol $\sstexpstackprod\,$:
  \begin{center}
     $
     \AxiomC{$ \asstexpi{1} \:\lt{\aact}\: \asstexpacci{1} $}
     \UnaryInfC{$ \stexpstackprod{\asstexpi{1}}{\stexpit{\astexpi{2}}} \:\lt{\aact}\: \stexpstackprod{\asstexpacci{1}}{\stexpit{\astexpi{2}}} $}
     \DisplayProof
     $
  \end{center}
  By applying $\sproj$ to the stacked star expressions in this instance, we obtain the instance:
  \begin{center}
    $
    \begin{aligned}[c]
       \AxiomC{$ \proj{\asstexpi{1}} \:\lt{\aact}\: \proj{\asstexpacci{1}} $}
       \UnaryInfC{$ \proj{\stexpstackprod{\asstexpi{1}}{\stexpit{\astexpi{2}}}} \:\lt{\aact}\: \proj{\stexpstackprod{\asstexpacci{1}}{\stexpit{\astexpi{2}}}} $}
       \DisplayProof
    \end{aligned} 
    \quad\text{ which by definition of $\sproj$ is equal to }\quad 
    \begin{aligned}[c]
       \AxiomC{$ \proj{\asstexpi{1}} \:\lt{\aact}\: \proj{\asstexpacci{1}} $}
       \UnaryInfC{$ \stexpprod{\proj{\asstexpi{1}}}{\stexpit{\astexpi{2}}} 
                      \:\lt{\aact}\: 
                    \stexpprod{\proj{\asstexpacci{1}}}{\stexpit{\astexpi{2}}} $}
       \DisplayProof
    \end{aligned} 
    $
  \end{center}
  and as such is an instance of the rule of the TSS~$\StExpTSS$ that permits putting transitions into the context~$\acxt$.
  It is straightforward to check this also for all other instances of rules of $\stackStExpTSS$ without \onetransitions,
  where the rule that produces a proper transition and has $\terminates{\astexpi{1}}$ as its left premise 
  we use statement~\ref{it:1:lem:stackStExpTSS:proj:StExpTSS} of the lemma. 
  Then statement~\ref{it:2:lem:stackStExpTSS:proj:StExpTSS} of the lemma 
  can be established by induction on the depth of derivations in $\stackStExpTSS$, 
  thereby using the projection property for axioms and rules of $\stackStExpTSS$,
  in the induction step. 
\end{proof}

%
%

\begin{lem}\label{lem:1:lem:indTSS:stackStExpTSS:2:StExpTSS}
  For all $\asstexp,\asstexpacc\in\stackStExpover{\actions}$, $\astexpacc\in\StExpover{\actions}$, and $\aact\in\actions$
  the following implications hold
  from derivability in the TSS $\stackStExpTSSover{\actions}$ 
  to  derivability in the TSS $\StExpTSSover{\actions}$
  (again we drop $\actions$ from their designations):
  \begin{align}
    \derivablein{\stackStExpTSS}
      \terminates{\asstexp}
          \;\;\; & \Longrightarrow \;\;\;
    \asstexp\in\StExpover{\actions}
      \;\logand\;\;
    \derivablein{\StExpTSS}
      \terminates{\asstexp} \punc{,}
    \label{eq:1:lem:1:lem:indTSS:stackStExpTSS:2:StExpTSS}
    \displaybreak[0]\\
    \derivablein{\stackStExpTSS}
      \asstexp \lt{\aact} \asstexpacc
          \;\;\; & \Longrightarrow \;\;\;
    \derivablein{\StExpTSS}
      \proj{\asstexp} \lt{\aact} \proj{\asstexpacc} \punc{,}
    \label{eq:2:lem:1:lem:indTSS:stackStExpTSS:2:StExpTSS}
    \displaybreak[0]\\
    \derivablein{\stackStExpTSS}
      \asstexp \lt{\sone} \asstexptilde
          \;\;\; & \Longrightarrow \;\;\;
    \bigl(\,
      \derivablein{\StExpTSS}
        \terminates{\proj{\asstexptilde}}
          \;\;\Longrightarrow\;\;
      \derivablein{\StExpTSS}    
        \terminates{\proj{\asstexp}}
    \,\bigr) \punc{,} 
    \label{eq:3:lem:1:lem:indTSS:stackStExpTSS:2:StExpTSS}
    \displaybreak[0]\\
    \derivablein{\stackStExpTSS}
      \asstexp \lt{\sone} \asstexptilde 
          \;\;\; & \Longrightarrow \;\;\;
    \bigl(\,
      \derivablein{\StExpTSS}
        \proj{\asstexptilde} \lt{\aact} \astexpacc
          \;\;\Longrightarrow\;\;
      \derivablein{\StExpTSS}    
        \proj{\asstexp} \lt{\aact} \astexpacc
    \,\bigr) \punc{.}
    \label{eq:4:lem:1:lem:indTSS:stackStExpTSS:2:StExpTSS}
  \end{align}
\end{lem}

\begin{proof}
  Statement~\eqref{eq:1:lem:1:lem:indTSS:stackStExpTSS:2:StExpTSS}
  is guaranteed by Lemma~\ref{lem:stackStExpTSS:proj:StExpTSS}, \ref{it:1:lem:stackStExpTSS:proj:StExpTSS}.
  Statement~\eqref{eq:2:lem:1:lem:indTSS:stackStExpTSS:2:StExpTSS}, universally quantified,
  is guaranteed by Lemma~\ref{lem:stackStExpTSS:proj:StExpTSS}, \ref{it:2:lem:stackStExpTSS:proj:StExpTSS}.
  
  It remains to show the two implications~\eqref{eq:3:lem:1:lem:indTSS:stackStExpTSS:2:StExpTSS} and \eqref{eq:4:lem:1:lem:indTSS:stackStExpTSS:2:StExpTSS},
  universally quantified. 
  We will use statement~\eqref{eq:3:lem:1:lem:indTSS:stackStExpTSS:2:StExpTSS} in the proof of statement~\eqref{eq:4:lem:1:lem:indTSS:stackStExpTSS:2:StExpTSS}.
  This notwithstanding,
  we will demonstrate them in parallel, because they have the same assumption.
  We proceed by induction on the structure of $\asstexp\in\stackStExpover{\actions}$. 
  
  For this, we assume $\asstexp,\asstexptilde\in\stackStExpover{\actions}$
  such that $\derivablein{\stackStExpTSS} \asstexp \lt{\sone} \asstexptilde$.
  We will use a case distinction to show the two implications on the right-hand side of the outer implications 
  in \eqref{eq:3:lem:1:lem:indTSS:stackStExpTSS:2:StExpTSS} and \eqref{eq:4:lem:1:lem:indTSS:stackStExpTSS:2:StExpTSS}. 
  There are three rules of $\stackStExpTSS$ that introduce \onetransitions. 
  We distinguish the following three cases according to which\vspace*{-1.5pt} of these three possible rules
  is applied at the bottom of a derivation of $\asstexp \lt{\sone} \asstexptilde$ in $\stackStExpTSS$.
  
  \begin{description}[itemsep=1.25ex]
    \item{\emph{Case~1:}} \mbox{}
      $\asstexp = \stexpstackprod{\asstexpi{1}}{\stexpit{\astexpi{2}}}$,
      $\asstexptilde = \stexpstackprod{\asstexptildei{1}}{\stexpit{\astexpi{2}}}$,
      the bottommost rule application in a derivation of $\asstexp \lt{\sone} \asstexptilde$ in $\stackStExpTSS$ is:\vspace*{-0.5ex} 
      \begin{center}
        $
        \AxiomC{$ \asstexpi{1} \lt{\sone} \asstexptildei{1} $}
        \UnaryInfC{$ \stexpstackprod{\asstexpi{1}}{\stexpit{\astexpi{2}}} \lt{\sone} \stexpstackprod{\asstexptildei{1}}{\stexpit{\astexpi{2}}} $}
        \DisplayProof 
        $
      \end{center}
      In this case we have $\proj{\asstexp} = \proj{\stexpstackprod{\asstexpi{1}}{\stexpit{\astexpi{2}}}}
                                            = \stexpprod{\proj{\asstexpi{1}}}{\proj{\stexpit{\astexpi{2}}}}
                                            = \stexpprod{\proj{\asstexpi{1}}}{\stexpit{\astexpi{2}}}$,
      and                  $\proj{\asstexptilde} 
                                            = \stexpprod{\proj{\asstexptildei{1}}}{\stexpit{\astexpi{2}}}$,
      and\vspace*{-2pt} also that $\derivablein{\stackStExpTSS} \asstexpi{1} \lt{\sone} \asstexptildei{1}$ holds.  
      \vspace*{0.5ex}  
                   
      For showing \eqref{eq:1:lem:1:lem:indTSS:stackStExpTSS:2:StExpTSS} in this case,                                  
      suppose that $\derivablein{\StExpTSS} \terminates{\proj{\asstexptilde}}$ holds.
      This means $\derivablein{\StExpTSS} \terminates{(\stexpprod{\proj{\asstexptildei{1}}}{\stexpit{\astexpi{2}}})}$.
      Since this must be derived by the rule in $\StExpTSS$ for immediate termination of product expressions,
      it follows that also $\derivablein{\StExpTSS} \terminates{\proj{\asstexptildei{1}}}$ holds.  
      We have to show $\derivablein{\StExpTSS} \terminates{\proj{\asstexp}}$.
      Since $\asstexpi{1}$ is a proper subexpression\vspace*{-2pt} of $\stexpstackprod{\asstexpi{1}}{\stexpit{\astexpi{2}}}$,
      we can apply 
      the induction hypothesis for showing \eqref{eq:3:lem:1:lem:indTSS:stackStExpTSS:2:StExpTSS}
      to $\derivablein{\stackStExpTSS} \asstexpi{1} \lt{\sone} \asstexptildei{1}$ and $\derivablein{\StExpTSS} \terminates{\proj{\asstexptildei{1}}}$. 
      We obtain $\derivablein{\StExpTSS} \terminates{\proj{\asstexpi{1}}}$.
      Due to $\derivablein{\StExpTSS} \terminates{(\stexpit{\astexpi{2}})}$, and the rule for $\sterminates$ for product expressions in $\StExpTSS$ 
      we obtain $\derivablein{\StExpTSS} \terminates{(\stexpprod{\proj{\asstexpi{1}}}{\stexpit{\astexpi{2}}})}$.
      In this way we have shown $\derivablein{\StExpTSS} \terminates{\proj{\asstexp}}$,
      due to $\proj{\asstexp} = \stexpprod{\proj{\asstexpi{1}}}{\stexpit{\astexpi{2}}}$.
      
      \smallskip\enlargethispage{2ex}
      For showing \eqref{eq:2:lem:1:lem:indTSS:stackStExpTSS:2:StExpTSS} in this case,
      suppose that $\derivablein{\StExpTSS} \proj{\asstexptilde} \lt{\aact} \astexpacc$ holds, for some $\aact\in\actions$ and $\astexpacc\in\StExpover{\actions}$.  
      Hence $\derivablein{\StExpTSS} \stexpprod{\proj{\asstexptildei{1}}}{\stexpit{\astexpi{2}}} \lt{\aact} \astexpacc$           
      holds. We have to show that $\derivablein{\StExpTSS} \proj{\asstexp} \lt{\aact} \astexpacc$ holds as well.
      We distinguish the two possible cases in which the \transitionact{\aact} from $\stexpprod{\proj{\asstexptildei{1}}}{\stexpit{\astexpi{2}}}$
      arises via a step from $\proj{\asstexptildei{1}}$ or via a step from $\stexpit{\astexpi{2}}$.  
      
      If the step arises via a step from $\proj{\asstexptildei{1}}$,
      then $\derivablein{\StExpTSS} \proj{\asstexptildei{1}} \lt{\aact} \astexpacci{0}$ holds
      for some $\astexpacci{0}\in\StExpover{\actions}$ with \mbox{$\astexpacc = \stexpprod{\astexpacci{0}}{\stexpit{\astexpi{2}}}$}.
      As $\asstexpi{1}$ is a proper subexpression of $\stexpstackprod{\asstexpi{1}}{\stexpit{\astexpi{2}}}$,
      we can apply the induction hypothesis\vspace*{-1pt} to $\derivablein{\stackStExpTSS} \asstexpi{1} \lt{\sone} \asstexptildei{1}$
      and $\derivablein{\StExpTSS} \proj{\asstexptildei{1}} \lt{\aact} \astexpacci{0}$.
      We obtain that $\derivablein{\StExpTSS} \proj{\asstexpi{1}} \lt{\aact} \astexpacci{0}$ holds as well.
      Then by using\vspace*{-1pt} the rule of $\StExpTSS$ for putting this step into the context $\stexpprod{\Box}{\stexpit{\astexpi{2}}}$
      we obtain $\derivablein{\StExpTSS} \stexpprod{\proj{\asstexpi{1}}}{\stexpit{\astexpi{2}}} \lt{\aact} \stexpprod{\astexpacci{0}}{\stexpit{\astexpi{2}}}$,
      and hence $\derivablein{\StExpTSS} \proj{\asstexp} \lt{\aact} \astexpacc$.
      
      If the step arises via a step from $\stexpit{\astexpi{2}}$,
      then $\derivablein{\StExpTSS} \terminates{\proj{\asstexptildei{1}}}$,
      and  $\derivablein{\StExpTSS} \stexpit{\astexpi{2}} \lt{\aact} \astexpacc$.
      Now by applying statement~\eqref{eq:1:lem:1:lem:indTSS:stackStExpTSS:2:StExpTSS}
      to $\derivablein{\stackStExpTSS} \asstexpi{1} \lt{\sone} \asstexptildei{1}$ 
      and $\derivablein{\StExpTSS} \terminates{\proj{\asstexptildei{1}}}$
      we obtain $\derivablein{\StExpTSS} \terminates{\proj{\asstexpi{1}}}$.
      From this 
           and $\derivablein{\StExpTSS} \stexpit{\astexpi{2}} \lt{\aact} \astexpacc$
      we obtain
      $\derivablein{\StExpTSS} \stexpprod{\proj{\asstexpi{1}}}{\stexpit{\astexpi{2}}} \lt{\aact} \astexpacc$
      by a rule application in $\StExpTSS$. 
      By using the consequence
      $\proj{\asstexp} = \stexpprod{\proj{\asstexpi{1}}}{\stexpit{\astexpi{2}}}$
      of the assumption
      in this case we obtain again $\derivablein{\StExpTSS} \proj{\asstexp} \lt{\aact} \astexpacc$.
      
      As this case distinction was exhaustive, and we have shown the proof obligation here for \eqref{eq:2:lem:1:lem:indTSS:stackStExpTSS:2:StExpTSS}.
      
    \item{\emph{Case~2:}} \mbox{}   \mbox{}
      $\asstexp = \stexpprod{\asstexpi{1}}{\astexpi{2}}$,
      $\asstexptilde = \stexpprod{\asstexptildei{1}}{\astexpi{2}}$,
      the bottommost rule application in a derivation of $\asstexp \lt{\sone} \asstexptilde$ in $\stackStExpTSS$ is:\vspace*{-0.5ex} 
      \begin{center}
        $
        \AxiomC{$ \asstexpi{1} \lt{\sone} \asstexptildei{1} $}
        \UnaryInfC{$ \stexpprod{\asstexpi{1}}{\astexpi{2}} \lt{\sone} \stexpprod{\asstexptildei{1}}{\astexpi{2}} $}
        \DisplayProof 
        $
      \end{center}
      In this case we have $\proj{\asstexp} = \proj{\stexpprod{\asstexpi{1}}{\astexpi{2}}}
                                            = \stexpprod{\proj{\asstexpi{1}}}{\proj{\astexpi{2}}}
                                            = \stexpprod{\proj{\asstexpi{1}}}{\astexpi{2}}$,
      and                  $\proj{\asstexptilde} 
                                            = \stexpprod{\proj{\asstexptildei{1}}}{\astexpi{2}}$,
      and\vspace*{-2pt} we also find $\derivablein{\stackStExpTSS} \asstexpi{1} \lt{\sone} \asstexptildei{1}$.  
      \vspace*{0.5ex}
      
      This case can be settled in close analogy with our argumentation for Case~1.
      
    \item{\emph{Case~3:}} \mbox{}
      $\asstexp = \stexpprod{\astexpi{1}}{\stexpit{\astexpi{2}}}$,
      $\asstexptilde = \stexpit{\astexpi{2}}$,
      the bottommost rule application in a derivation of $\asstexp \lt{\sone} \asstexptilde$ in $\stackStExpTSS$ is:\vspace*{-0.5ex} 
      \begin{center}
        $
        \AxiomC{$ \terminates{\astexpi{1}} $}
        \UnaryInfC{$ \stexpprod{\astexpi{1}}{\stexpit{\astexpi{2}}} \lt{\sone} \stexpit{\astexpi{2}} $}
        \DisplayProof 
        $
      \end{center}
      In this case we have $\proj{\asstexp} = \stexpprod{\astexpi{1}}{\stexpit{\astexpi{2}}}$ and
                           $\proj{\asstexptilde} = \stexpit{\astexpi{2}}$,
      and we find $\derivablein{\stackStExpTSS} \terminates{\astexpi{1}}$. 
      From the latter we obtain $\derivablein{\StExpTSS} \terminates{\astexpi{1}}$ by Lemma~\ref{lem:stackStExpTSS:proj:StExpTSS}. 
      \vspace*{0.5ex}                   
      
      For showing \eqref{eq:1:lem:1:lem:indTSS:stackStExpTSS:2:StExpTSS} in this case, it suffices 
      to prove $\derivablein{\StExpTSS} \terminates{\proj{\asstexp}}$.
      We do not have to assume $\derivablein{\StExpTSS} \terminates{\proj{\asstexptilde}}$, because that holds anyway in this case
      due to $\derivablein{\StExpTSS} \terminates{(\stexpit{\astexpi{2}})}$ and the termination rule for iteration expressions in $\StExpTSS$.
      Now from $\derivablein{\StExpTSS} \terminates{\astexpi{1}}$ and $\derivablein{\StExpTSS} \terminates{(\stexpit{\astexpi{2}})}$
      we obtain $\derivablein{\StExpTSS} \terminates{\stexpprod{\astexpi{1}}{(\stexpit{\astexpi{2}})}}$,
      and hence $\derivablein{\StExpTSS} \terminates{\proj{\asstexp}}$.
      
      \smallskip
      For showing \eqref{eq:2:lem:1:lem:indTSS:stackStExpTSS:2:StExpTSS} in this case, 
      we suppose $\derivablein{\StExpTSS} \proj{\asstexptilde} \lt{\aact} \astexpacc$,
      and hence $\derivablein{\StExpTSS} \stexpit{\astexpi{2}} \lt{\aact} \astexpacc$,
      for some $\aact\in\actions$ and $\astexpacc\in\StExpover{\actions}$. 
      We have to show $\derivablein{\StExpTSS} \proj{\asstexp} \lt{\aact} \astexpacc$.
      Due to $\derivablein{\StExpTSS} \terminates{\astexpi{1}}$,
      and $\derivablein{\StExpTSS} \stexpit{\astexpi{2}} \lt{\aact} \astexpacc$
      we obtain $\derivablein{\StExpTSS} \stexpprod{\astexpi{1}}{\stexpit{\astexpi{2}}} \lt{\aact} \astexpacc$
      by a rule application in $\StExpTSS$.
      By the consequence $\proj{\asstexp} = \stexpprod{\astexpi{1}}{\stexpit{\astexpi{2}}}$ 
      of the assumption in this case we have established
      the proof obligation $\derivablein{\StExpTSS} \proj{\asstexp} \lt{\aact} \astexpacc$.
  \end{description}\enlargethispage{4ex}
  By having verified, in all of these three possible cases,
  the induction steps for the proofs by induction of \eqref{eq:1:lem:1:lem:indTSS:stackStExpTSS:2:StExpTSS} and \eqref{eq:2:lem:1:lem:indTSS:stackStExpTSS:2:StExpTSS},
  we have shown the universally quantified statements \eqref{eq:1:lem:1:lem:indTSS:stackStExpTSS:2:StExpTSS} and \eqref{eq:2:lem:1:lem:indTSS:stackStExpTSS:2:StExpTSS}.
\end{proof}

Now we formulate and prove a crucial lemma that relates derivability statements in $\stackStExpindTSSover{\actions}$ 
with derivability statements in $\StExpTSSover{\actions}$
in a way that will enable us to show that the projection function $\sproj$ 
defines a bisimulation 
  between the \LTSs\ $\indLTSof{\oneLTSof{\stackStExpTSSover{\actions}}} 
                        =
                      \indLTSof{\oneLTSof{\stackStExpover{\actions}}} 
                        = 
                      \LTSof{\stackStExpindTSSover{\actions}}$ 
                      (see Definition~\ref{def:stackStExpTSS}, and Lemma~\ref{lem:indLTS:oneLTS:stackStExps})
  and $\LTSof{\StExpTSSover{\actions}}$                     
  that are generated by these two \TSSs, respectively.

\begin{lem}\label{lem:indTSS:stackStExpTSS:2:StExpTSS}
  For all $\asstexp,\asstexpacc\in\stackStExpover{\actions}$, $\astexpacc\in\StExpover{\actions}$, and $\aact\in\actions$
  the following statements hold
  concerning derivability in the TSS $\stackStExpindTSSover{\actions}$ 
  and derivability in the TSS $\StExpTSSover{\actions}$
  (we drop $\actions$ from their designations):\vspace*{-0.5ex}
  \begin{align}
    \derivablein{\stackStExpindTSS}
      \oneterminates{\asstexp}
        \;\;\; & \:\Longrightarrow \;\;\;\; 
    \derivablein{\StExpTSS}
      \terminates{\proj{\asstexp}}\punc{,}
        \label{eq:1:lem:indTSS:stackStExpTSS:2:StExpTSS}
    \displaybreak[0]\\[-0.75ex] 
    \derivablein{\stackStExpindTSS}    
      \asstexp
        \ilt{\aact}
      \asstexpacc
        \;\;\; & \:\Longrightarrow \;\;\;\;
    \derivablein{\StExpTSS}
      \proj{\asstexp}
        \lt{\aact}
      \proj{\asstexpacc} \punc{,}
        \label{eq:2:lem:indTSS:stackStExpTSS:2:StExpTSS}
    \displaybreak[0]\\[-0.6ex] 
      \derivablein{\stackStExpindTSS}
        \oneterminates{\asstexp}
          \;\;\; & \Longleftarrow\: \;\;\;\; 
      \derivablein{\StExpTSS}
        \terminates{\proj{\asstexp}} \punc{,}
        \label{eq:3:lem:indTSS:stackStExpTSS:2:StExpTSS}
    \displaybreak[0]\\[-0.6ex]
    \existsstzero{\asstexpacc\in\stackStExp}
      \bigl[\,
        \proj{\asstexpacc} = \astexpacc
          \;\logand\;\; 
        {\derivablein{\stackStExpindTSS}
           \asstexp \ilt{\aact} \asstexpacc} 
      \,\bigr] 
        \;\;\; & \Longleftarrow\: \;\;\;\;
    \derivablein{\StExpTSS}
      \proj{\asstexp}
        \lt{\aact}
      \astexpacc\punc{.}
        \label{eq:4:lem:indTSS:stackStExpTSS:2:StExpTSS}
  \end{align}
\end{lem}

\begin{proof}
  For $\asstexp,\asstexpacc\in\stackStExpover{\actions}$, $\astexpacc\in\StExpover{\actions}$, and $\aact\in\actions$,
  statement~\eqref{eq:1:lem:indTSS:stackStExpTSS:2:StExpTSS} of the lemma
  can be shown from \eqref{eq:1:lem:1:lem:indTSS:stackStExpTSS:2:StExpTSS} and \eqref{eq:3:lem:1:lem:indTSS:stackStExpTSS:2:StExpTSS} 
               in Lemma~\ref{lem:1:lem:indTSS:stackStExpTSS:2:StExpTSS},
  and statement
  \eqref{eq:2:lem:indTSS:stackStExpTSS:2:StExpTSS} from \eqref{eq:2:lem:1:lem:indTSS:stackStExpTSS:2:StExpTSS} and \eqref{eq:4:lem:1:lem:indTSS:stackStExpTSS:2:StExpTSS} 
               in Lemma~\ref{lem:1:lem:indTSS:stackStExpTSS:2:StExpTSS}.
  
  \smallskip
  Statements~\eqref{eq:3:lem:indTSS:stackStExpTSS:2:StExpTSS} and \eqref{eq:4:lem:indTSS:stackStExpTSS:2:StExpTSS} 
  of the lemma can be shown by means of a proof by induction on the structure of $\asstexp$
  that makes crucial use of the admissible rules of $\stackStExpindTSS$ in Lemma~\ref{lem:2:lem:indTSS:stackStExpTSS:2:StExpTSS}. 
  For proving that 
  we proceed by structural induction on the stacked star expression $\asstexp\in\stackStExpover{\actions}$.
  
  For performing the induction step, we let $\asstexp\in\stackStExpover{\actions}$ be arbitrary. 
  We distinguish the seven possible cases that arise from the grammar of stacked star expressions in Def.~\ref{def:stackStExp},
  which recurs on the grammar for star expressions in Def.~\ref{def:StExp},
  provided that we subsume 
  the case $\asstexp = \stexpprod{\astexpi{1}}{\astexpi{2}}$ in Def.~\ref{def:StExp} with $\astexpi{1},\astexpi{2}\in\StExpover{\actions}$
  under the case $\asstexp = \stexpprod{\asstexpi{1}}{\astexpi{2}}$ in Def.~\ref{def:stackStExp} 
    with $\asstexpi{1}\in\stackStExpover{\actions}$, and $\astexpi{2}\in\StExpover{\actions}$.
  In each of these seven cases we establish the induction steps, restricted to the case under consideration, 
  for proofs of \eqref{eq:3:lem:indTSS:stackStExpTSS:2:StExpTSS} and \eqref{eq:4:lem:indTSS:stackStExpTSS:2:StExpTSS} 
  by structural induction on $\asstexp\,$:
  \begin{description}[itemsep=0.5ex]
    \item{\emph{Case~1:}} \mbox{}
      $\asstexp = \stexpzero$.   
      Then $\proj{\asstexp} = \stexpzero = \asstexp$.
      
      Since $\proj{\asstexp} = \stexpzero$ neither permits immediate termination nor a transition according to the TSS~$\StExpTSS$,
      the assumptions of both implications \eqref{eq:3:lem:indTSS:stackStExpTSS:2:StExpTSS} and \eqref{eq:4:lem:indTSS:stackStExpTSS:2:StExpTSS}
      are wrong, and so \eqref{eq:3:lem:indTSS:stackStExpTSS:2:StExpTSS} and \eqref{eq:4:lem:indTSS:stackStExpTSS:2:StExpTSS} hold in this case.
  
    \item{\emph{Case~2:}} \mbox{}
      $\asstexp = \stexpone$.  
      Then $\proj{\asstexp} = \stexpone = \asstexp$.
      
      Then $\derivablein{\stackStExpTSS} \terminates{\asstexp}$ holds, and hence also $\derivablein{\stackStExpindTSS} \oneterminates{\asstexp}$.
      This shows \eqref{eq:3:lem:indTSS:stackStExpTSS:2:StExpTSS}. 
      Since $\proj{\asstexp} = \stexpone$\vspace*{-1.5pt} does not permit a step in $\StExpTSS$, 
      \eqref{eq:4:lem:indTSS:stackStExpTSS:2:StExpTSS} holds trivially as well. 
      
    \item{\emph{Case~3:}} \mbox{}\enlargethispage{3ex}
      $\asstexp = \aact$ for some $\aact\in\actions$.  
      Then $\proj{\asstexp} = \aact = \asstexp$. 
      
      Since $\proj{\asstexp} = \aact$ does not permit immediate termination according to $\StExpTSS$, \eqref{eq:3:lem:indTSS:stackStExpTSS:2:StExpTSS} holds trivially.
      In order to show \eqref{eq:4:lem:indTSS:stackStExpTSS:2:StExpTSS}, 
      we consider a step $\derivablein{\StExpTSS} \proj{\asstexp} \lt{\aact} \astexpacc$ for some $\aact\in\actions$, and $\astexpacc\in\StExpover{\actions}$.
      We have to find $\asstexpacc\in\stackStExpover{\actions}$\vspace*{-4pt} 
        with $\derivablein{\stackStExpindTSS} \asstexp \ilt{\aact} \asstexpacc$ and $\proj{\asstexpacc} = \astexpacc$.   
      Now $\derivablein{\StExpTSS} \proj{\asstexp} \lt{\aact} \astexpacc$\vspace*{-2pt}
      means $\derivablein{\StExpTSS} \aact \lt{\aact} \stexpone$ by the rule for actions in $\StExpTSS$, and hence $\astexpacc = \stexpone$.
      Then with $\asstexpacc \defdby \stexpone$ we get $\derivablein{\stackStExpTSS} \asstexp \lt{\aact} \asstexpacc$
      due\vspace*{-4pt} to the rule for actions in $\stackStExpTSS$, and hence we obtain $\derivablein{\stackStExpindTSS} \asstexp \ilt{\aact} \asstexpacc$
      with $\proj{\asstexpacc} = \stexpone = \astexpacc$. 
      
    \item{\emph{Case~4:}} \mbox{}
      $\asstexp = \stexpsum{\astexpi{1}}{\astexpi{2}}$ for some $\astexpi{1},\astexpi{2}\in\StExpover{\actions}$.
      Then $\proj{\asstexp} = \stexpsum{\proj{\astexpi{1}}}{\proj{\astexpi{2}}} = \stexpsum{\astexpi{1}}{\astexpi{2}} = \asstexp$. 
      
      For showing \eqref{eq:3:lem:indTSS:stackStExpTSS:2:StExpTSS}, we suppose that $\derivablein{\StExpTSS} \terminates{\proj{\asstexp}}$ holds. 
      We have to show $\derivablein{\stackStExpTSS} \oneterminates{\asstexp}$.
      Since $\asstexp$ is a star expression in this case,
      and the rules for immediate termination of star expressions coincide in $\StExpTSS$ and $\stackStExpTSS$,
      in view of $\proj{\asstexp} = \asstexp$ from
      $\derivablein{\StExpTSS} \terminates{\proj{\asstexp}}$ 
      we conclude
      $\derivablein{\stackStExpTSS} \terminates{\asstexp}$.
      This entails $\derivablein{\stackStExpTSS} \oneterminates{\asstexp}$.
      
      For showing \eqref{eq:4:lem:indTSS:stackStExpTSS:2:StExpTSS},
      we suppose that $\derivablein{\StExpTSS} \proj{\asstexp} \lt{\aact} \astexpacc$, for some $\astexpacc\in\StExpover{\actions}$ and $\aact\in\actions$.
      Then due to the rule for $\sstexpsum$ in $\StExpTSS$
      there is $i\in\setexp{1,2}$ such that $\derivablein{\StExpTSS} \astexpi{i} \lt{\aact} \astexpacc$, which we pick accordingly.
      We have to find $\asstexpacc\in\stackStExpover{\actions}$
        such that $\derivablein{\stackStExpindTSS} \asstexp \ilt{\aact} \asstexpacci{i}$ and $\proj{\asstexpacc} = \astexpacc$.
      Since $\astexpi{i}$ is a subexpression~of~$\asstexp$,\vspace*{-4pt} we can apply the induction hypothesis for $\asstexpi{i}$.
      We obtain $\asstexpacci{i}\in\stackStExpover{\actions}$ 
        with\vspace*{-2.5pt} $\derivablein{\stackStExpindTSS} \astexpi{i} \ilt{\aact} \asstexpacci{i}$ and $\proj{\asstexpacci{i}} = \astexpi{i}$.
      Since no \onetransitions\ can depart from a star expression like $\astexpi{i}$ according to $\stackStExpTSS$,
        we also have $\derivablein{\stackStExpTSS} \astexpi{i} \lt{\aact} \asstexpacci{i}$.
      From this we get $\derivablein{\stackStExpTSS} \stexpsum{\astexpi{1}}{\astexpi{2}} \lt{\aact} \asstexpacci{i}$.
      Hence $\derivablein{\stackStExpindTSS} \asstexp \lt{\aact} \asstexpacci{i}$
      with $\proj{\asstexpacc} = \astexpacc$. 
      
    \item{\emph{Case~5:}} \mbox{}
      $\asstexp = \stexpprod{\asstexpi{1}}{\astexpi{2}}$ for some $\asstexpi{1}\in\stackStExpover{\actions}$, and $\astexpi{2}\in\StExpover{\actions}$.
      Then $\proj{\asstexp} 
                            = \stexpprod{\proj{\asstexpi{1}}}{\astexpi{2}}$.
      
      This case can be settled analogously to the slightly more complicated Case~7. Therefore we defer argumentation to that case.
      
    \item{\emph{Case~6:}} \mbox{}
      $\asstexp = \stexpit{\astexpi{0}}$ for some $\astexpi{0}\in\StExpover{\actions}$.
      Then $\proj{\asstexp} = \stexpit{\astexpi{0}} = \asstexp$.
      \vspace*{0.25ex}  
      
      For showing \eqref{eq:3:lem:indTSS:stackStExpTSS:2:StExpTSS}, it suffices to show $\derivablein{\stackStExpindTSS} \oneterminates{\asstexp}$.
      We get $\derivablein{\stackStExpTSS} \terminates{\asstexp}$ due to $\asstexp = \stexpit{\astexpi{0}}$ by the rule in $\stackStExpTSS$\vspace*{-2pt}
      for\vspace*{-2pt} immediate termination of iteration expressions. 
      From $\derivablein{\stackStExpTSS} \terminates{\asstexp}$ we obtain $\derivablein{\stackStExpindTSS} \oneterminates{\asstexp}$
      in $\stackStExpindTSS$.
      
      For showing \eqref{eq:4:lem:indTSS:stackStExpTSS:2:StExpTSS}, we suppose $\derivablein{\StExpTSS} \proj{\asstexp} \lt{\aact} \astexpacc$
      for $\aact\in\actions$, and $\astexpacc\in\StExpover{\actions}$. 
      This entails, due to $\proj{\asstexp} = \stexpit{\astexpi{0}}$, and the rule in $\StExpTSS$ for steps for iteration expressions,
      that $\derivablein{\StExpTSS} \astexpi{0} \lt{\aact} \astexpacci{0}$
      for some $\astexpacci{0}\in\StExpover{\actions}$ such that $\astexpacc = \stexpprod{\astexpacci{0}}{\stexpit{\astexpi{0}}}$.
      We have to find $\asstexpacc\in\stackStExpover{\actions}$
        such that $\derivablein{\stackStExpindTSS} \asstexp \ilt{\aact} \asstexpacc$ and\vspace*{-2pt} $\proj{\asstexpacc} = \astexpacc$.
      As $\astexpi{0}$ is a subexpression of $\stexpit{\astexpi{0}}$, and $\proj{\astexpi{0}} = \astexpi{0}$, 
      we can apply the induction hypothesis\vspace*{-2pt} to $\derivablein{\StExpTSS} \proj{\astexpi{0}} \lt{\aact} \astexpacci{0}$.
      We obtain that there is $\asstexpacci{0}\in\stackStExpover{\actions}$ 
      with $\derivablein{\stackStExpindTSS} \astexpi{0} \ilt{\aact} \asstexpacci{0}$ and $\proj{\asstexpacci{0}} = \astexpacci{0}$.\vspace*{-2pt}
      Now since there are\vspace*{-3pt} no \onetransitions\ possible in $\stackStExpTSS$ from star expressions like $\astexp$,
      we also get \mbox{$\derivablein{\stackStExpTSS} \astexpi{0} \lt{\aact} \asstexpacci{0}$}.
      From this we obtain, by applying the rule for transitions from iteration expressions\vspace*{-4pt} in $\stackStExpTSS$,
      that $\derivablein{\stackStExpTSS} \stexpit{\astexpi{0}} \lt{\aact} \stexpstackprod{\asstexpacci{0}}{\stexpit{\astexpi{0}}}$.
      Consequently we also get 
      $\derivablein{\stackStExpindTSS} \stexpit{\astexpi{0}} \ilt{\aact} \stexpstackprod{\asstexpacci{0}}{\stexpit{\astexpi{0}}}$.
      Then for\vspace*{-4pt} \mbox{$\asstexpacc \defdby \stexpstackprod{\asstexpacci{0}}{\astexpi{0}}$}
      we thus obtain 
      $\derivablein{\stackStExpTSS} \asstexp \ilt{\aact} \asstexpacc$
      with $\proj{\asstexpacc} = \stexpprod{\proj{\asstexpacci{0}}}{\proj{\stexpit{\astexpi{0}}}}
                              = \stexpprod{\astexpacci{0}}{\stexpit{\astexpi{0}}}
                              = \astexpacc$. 
      
    \item{\emph{Case~7:}} \mbox{}
      $\asstexp = \stexpstackprod{\asstexpi{1}}{\stexpit{\astexpi{2}}}$ for some $\asstexpi{1}\in\stackStExpover{\actions}$, and $\astexpi{2}\in\StExpover{\actions}$.
      Then $\proj{\asstexp} = \stexpprod{\proj{\asstexpi{1}}}{\stexpit{\astexpi{2}}}$.
      \vspace*{0.25ex}  
      
      For showing \eqref{eq:3:lem:indTSS:stackStExpTSS:2:StExpTSS}, we suppose that
      $\derivablein{\StExpTSS} \terminates{\proj{\asstexp}}$ holds, 
      and hence $\derivablein{\StExpTSS} \terminates{(\stexpprod{\proj{\asstexpi{1}}}{\stexpit{\astexpi{2}}})}\,$.
      By the rule of $\StExpTSS$ for immediate termination of product expressions this must be a consequence of
      $\derivablein{\StExpTSS} \terminates{\proj{\asstexpi{1}}}$ and $\derivablein{\StExpTSS} \terminates{\stexpit{\astexpi{2}}}\,$.
      We have to show that $\derivablein{\stackStExpindTSS} \asstexp$ holds.
      and $\derivablein{\stackStExpindTSS} \terminates{\stexpit{\astexpi{2}}}$. 
      Since $\asstexpi{1}$ is a proper subexpression\vspace*{-3pt} of $\asstexp$,
      we can apply the induction hypothesis for \eqref{eq:3:lem:indTSS:stackStExpTSS:2:StExpTSS} 
      to $\derivablein{\StExpTSS} \terminates{\proj{\asstexpi{1}}}$.
      We obtain that $\derivablein{\stackStExpindTSS} \oneterminates{\asstexpi{1}}$ holds.\vspace*{-3pt}
      By applying 
      the second admissible rule in Lemma~\ref{lem:2:lem:indTSS:stackStExpTSS:2:StExpTSS} 
      to $\derivablein{\stackStExpindTSS} \oneterminates{\asstexpi{1}}$ 
      we obtain
      $\derivablein{\stackStExpindTSS} \oneterminates{(\stexpstackprod{\asstexpi{1}}{\stexpit{\astexpi{2}}})}$,\vspace*{-3pt}
      and thus
      $\derivablein{\stackStExpindTSS} \oneterminates{\asstexp}$.
      
      For showing \eqref{eq:4:lem:indTSS:stackStExpTSS:2:StExpTSS}, we suppose $\derivablein{\StExpTSS} \proj{\asstexp} \lt{\aact} \astexpacc$
      for $\aact\in\actions$, and $\astexpacc\in\StExpover{\actions}$.
      We have to find\vspace*{-2pt} a stacked star expression $\asstexpacc\in\stackStExpover{\actions}$
        such that $\derivablein{\stackStExpindTSS} \asstexp \ilt{\aact} \asstexpacc$ and\vspace*{-3pt} $\proj{\asstexpacc} = \astexpacc$.
      As $\proj{\asstexp} = \stexpprod{\proj{\asstexpi{1}}}{\stexpit{\astexpi{2}}}$,
      the \transitionact{\aact} from $\proj{\asstexp}$ according to $\StExpTSS$
      can result by one of the two rules of $\StExpTSS$ with product on the left-hand side of the transition in the conclusion.
      We distinguish two subcases accordingly. 
      \begin{description}[itemsep=0.5ex]
        \item{\emph{Subcase~a:}} \mbox{}
          $\derivablein{\StExpTSS} \proj{\asstexp} \lt{\aact} \astexpacc$
          is derived in $\StExpTSS$ by a derivation whose immediate subderivations guarantee
          $\derivablein{\StExpTSS} \terminates{\proj{\asstexpi{1}}}$ and $\derivablein{\StExpTSS} \stexpit{\astexpi{2}} \lt{\aact} \astexpacc$.
          \vspace*{0.25ex}
          
          Since $\asstexpi{1}$ is a proper subexpression of $\asstexp$, we can apply \eqref{eq:3:lem:indTSS:stackStExpTSS:2:StExpTSS} 
          to $\derivablein{\StExpTSS} \terminates{\proj{\asstexpi{1}}}$. We get $\derivablein{\stackStExpindTSS} \oneterminates{\asstexpi{1}}$.
          As also $\stexpit{\astexpi{2}}$ is a proper subexpression of $\asstexp$,
          we can apply the induction hypothesis for \eqref{eq:4:lem:indTSS:stackStExpTSS:2:StExpTSS} 
          to\vspace*{-2pt} $\derivablein{\StExpTSS} \proj{\stexpit{\astexpi{2}}} \lt{\aact} \astexpacc$.
          We obtain that $\derivablein{\stackStExpindTSS} \stexpit{\astexpi{2}} \ilt{\aact} \asstexpacc$ 
          for some $\asstexpacc\in\stackStExpover{\actions}$ with $\proj{\asstexpacc} = \astexpacc$.\vspace*{-5pt}
          By applying the last admissible rule in Lemma~\ref{lem:2:lem:indTSS:stackStExpTSS:2:StExpTSS} 
          to $\derivablein{\stackStExpindTSS} \oneterminates{\asstexpi{1}}$ and $\derivablein{\stackStExpindTSS} \stexpit{\astexpi{2}} \ilt{\aact} \asstexpacc$
          we obtain\vspace*{-4pt}
          $\derivablein{\stackStExpindTSS} \stexpstackprod{\asstexpi{1}}{\stexpit{\astexpi{2}}} \ilt{\aact} \asstexpacc$.
          Thus we have found $\asstexpacc$ with 
          $\derivablein{\stackStExpindTSS} \asstexp \ilt{\aact} \asstexpacc$ with $\proj{\asstexpacc} = \astexpacc$.
        
        \item{\emph{Subcase~b:}} \mbox{}
          $\derivablein{\StExpTSS} \proj{\asstexp} \lt{\aact} \astexpacc$
          is derived in $\StExpTSS$ by a derivation whose immediate subderivation guarantees
          $\derivablein{\StExpTSS} \proj{\asstexpi{1}} \lt{\aact} \astexpacci{1}$
          for some $\astexpacci{1}\in\StExpover{\actions}$ with $\astexpacc = \stexpprod{\astexpacci{1}}{\stexpit{\astexpi{2}}}$.
          \vspace*{0.5ex}
          
          Since $\asstexpi{1}$ is a proper subexpression of $\asstexp$, 
          we can apply the induction hypothesis for \eqref{eq:4:lem:indTSS:stackStExpTSS:2:StExpTSS}
          to $\derivablein{\StExpTSS} \proj{\asstexpi{1}} \lt{\aact} \astexpacci{1}$.
          We obtain $\asstexpacci{1}\in\stackStExpover{\actions}$
          such that $\derivablein{\stackStExpindTSS} \asstexpi{1} \ilt{\aact} \asstexpacci{1}$
          with $\proj{\asstexpi{1}} = \astexpacci{1}$.
          Now\vspace*{-4pt} we can apply the forth admissible rule in Lemma~\ref{lem:2:lem:indTSS:stackStExpTSS:2:StExpTSS} 
          to $\derivablein{\stackStExpindTSS} \asstexpi{1} \ilt{\aact} \asstexpacci{1}$
          in order to get that also\vspace*{-5pt}
          $\derivablein{\stackStExpindTSS} \stexpstackprod{\asstexpi{1}}{\stexpit{\astexpi{2}}} \ilt{\aact} \stexpstackprod{\asstexpacci{1}}{\stexpit{\astexpi{2}}} $
          holds. 
          This shows that by $\asstexpacc \defdby \stexpstackprod{\asstexpi{1}}{\stexpit{\astexpi{2}}}$
          we have found an expression\vspace*{-4pt}
          with $\derivablein{\stackStExpindTSS} \asstexp \ilt{\aact} \asstexpacc$
          and $\proj{\asstexpacc} = \stexpprod{\proj{\asstexpacci{1}}}{\stexpit{\astexpi{2}}} 
                                  = \stexpprod{\astexpacci{1}}{\stexpit{\astexpi{2}}}
                                  = \astexpacc$.
      \end{description} 
      As we have successfully performed the induction step for \eqref{eq:4:lem:indTSS:stackStExpTSS:2:StExpTSS} in the two possible subcases,
      we have also done so for the case as a whole. 
  \end{description}      
  By having performed the induction steps for \eqref{eq:3:lem:indTSS:stackStExpTSS:2:StExpTSS} and for \eqref{eq:4:lem:indTSS:stackStExpTSS:2:StExpTSS}
  in each of the 7 cases that are possible for the structure of $\asstexp$,
  we have successfully performed it for the proofs of 
  \eqref{eq:3:lem:indTSS:stackStExpTSS:2:StExpTSS} and for \eqref{eq:4:lem:indTSS:stackStExpTSS:2:StExpTSS}.
  In this way we have established 
  the universally quantified statements \eqref{eq:3:lem:indTSS:stackStExpTSS:2:StExpTSS} and for \eqref{eq:4:lem:indTSS:stackStExpTSS:2:StExpTSS}
  in Lemma~\ref{lem:indTSS:stackStExpTSS:2:StExpTSS}.
\end{proof}

At long last we have gathered all auxiliary results that we need for proving Lemma~\ref{lem:StExpLTS:funbisim:StExpindLTS}.
Before finishing its proof, we repeat the formulation of this lemma.

\begin{repeatedlem}[= Lemma~\ref{lem:StExpLTS:funbisim:StExpindLTS}]
  The projection function $\sproj$ defines, as its graph, a (functional) bisimulation between the
  induced LTS $\indLTSof{\oneLTSof{\stackStExpover{\actions}}}$ of the stacked star expressions \oneLTS~$\stackStExpover{\actions}$
  and the star expressions LTS $\LTSof{\StExpover{\actions}}$.
\end{repeatedlem}

\begin{proof}[Proof (of Lemma~\ref{lem:StExpLTS:funbisim:StExpindLTS})]
  By transferring the four statements 
  of Lemma~\ref{lem:indTSS:stackStExpTSS:2:StExpTSS} from the \TSSs~$\stackStExpindTSSover{\actions}$ and $\StExpTSSover{\actions}$ 
  to their generated \LTSs~$\LTSdefdby{\stackStExpindTSSover{\actions}}$ and $\LTSdefdby{\StExpTSSover{\actions}}$,
  and by using 
  $\LTSdefdby{\stackStExpindTSSover{\actions}} = \LTSof{\stackStExpover{\actions}}$
  from Lemma~\ref{lem:indLTS:oneLTS:stackStExps},
  and 
  $\LTSdefdby{\StExpTSSover{\actions}} = \LTSof{\StExpover{\actions}}$ from Def.~\ref{def:stackStExpTSS}
  we obtain the following statements,
  for all stacked star expressions $\asstexp,\asstexpacc\in\stackStExpover{\actions}$, star expressions $\astexpacc\in\StExpover{\actions}$, and actions $\aact\in\actions\,$:
  \begin{align}
      \asstexp
        \ilt{\aact}
      \asstexpacc 
          \;\;\;\text{in $\indLTSof{\LTSof{\stackStExpTSSover{\actions}}}$}
        \;\;\; & \:\:\Longrightarrow \;\;\;\;
      \proj{\asstexp}
        \lt{\aact}
      \proj{\asstexpacc} 
        \;\;\;\text{in $\LTSof{\StExpover{\actions}}$} \punc{,}
          \label{eq:1:prf:lem:StExpLTS:funbisim:StExpindLTS}
    \\[-0.5ex]
    \begin{aligned}[b]
        \existsstzero{\asstexpacc\in\stackStExp}
          \bigl[\,
              & \proj{\asstexpacc} = \astexpacc
                  \;\logand\;\; 
              \\[-0.75ex]
              & \asstexp \ilt{\aact} \asstexpacc  
                  \;\;\;\text{in $\indLTSof{\LTSof{\stackStExpover{\actions}}}$}
            \,\bigr] 
    \end{aligned} 
        \;\;\; & 
      \begin{aligned}[b]
        \phantom{
         \Longleftarrow\: \;\;\;\;
        \proj{\asstexp}
          \lt{\aact}
        \astexpacc 
        \;\;\;\text{in $\LTSof{\stackStExpover{\actions}}$}
                 } 
        \\
         \Longleftarrow\: \;\;\;\;
        \proj{\asstexp}
          \lt{\aact}
        \astexpacc 
        \;\;\;\text{in $\LTSof{\StExpover{\actions}}$} \punc{,}
      \end{aligned}
          \label{eq:2:prf:lem:StExpLTS:funbisim:StExpindLTS}
    \\[0.25ex]
      \oneterminates{\asstexp} \;\;\;\text{in $\indLTSof{\LTSof{\stackStExpover{\actions}}}$}%
        \;\;\; & \Longleftrightarrow\: \;\;\;\;
      \terminates{\proj{\asstexp}} \;\;\;\text{in $\LTSof{\StExpover{\actions}}$}  \punc{.}
          \label{eq:3:prf:lem:StExpLTS:funbisim:StExpindLTS}
  \end{align}
  These statements follow from respective parts of the statement of Lemma~\ref{lem:indTSS:stackStExpTSS:2:StExpTSS},
  in particular:
  \eqref{eq:1:prf:lem:StExpLTS:funbisim:StExpindLTS} follows from \eqref{eq:2:lem:indTSS:stackStExpTSS:2:StExpTSS},
  \eqref{eq:2:prf:lem:StExpLTS:funbisim:StExpindLTS} from \eqref{eq:4:lem:indTSS:stackStExpTSS:2:StExpTSS},
  and    \eqref{eq:3:prf:lem:StExpLTS:funbisim:StExpindLTS} from \eqref{eq:1:lem:indTSS:stackStExpTSS:2:StExpTSS} and \eqref{eq:3:lem:indTSS:stackStExpTSS:2:StExpTSS}.
  Now \eqref{eq:1:prf:lem:StExpLTS:funbisim:StExpindLTS}, \eqref{eq:2:prf:lem:StExpLTS:funbisim:StExpindLTS}, and \eqref{eq:3:prf:lem:StExpLTS:funbisim:StExpindLTS} 
  witness the forth, the back, and the termination condition in Def.~\ref{def:bisims:LTSs}
  for the graph of the projection function $\sproj$, which also is \nonempty.
  Therefore $\sproj$ defines a bisimulation
  between $\indLTSof{\LTSof{\stackStExpover{\actions}}}$ and $\LTSof{\StExpover{\actions}}$.
\end{proof}

At last we are able to give the proof of Theorem~\ref{thm:onechart-int:funbisim:chart-int}.
We first repeat its formulation.

\begin{repeatedthm}[= Theorem~\ref{thm:onechart-int:funbisim:chart-int}]
  $\indscchartof{\onechartof{\astexp}} \funbisim \chartof{\astexp}$ holds for all $\astexp\in\StExpover{\actions}$:
  there is a functional bisimulation from the induced chart of the \onechart\ interpretation of a star expression~$\astexp$
  to the chart interpretation of $\astexp$.
\end{repeatedthm}

\begin{proof}[Proof of Theorem~\ref{thm:onechart-int:funbisim:chart-int}]
  Let $\astexp\in\StExpover{\actions}$ be a star expression.
  By Def.~\ref{def:onechart:interpretation}, the \onechart\ interpretation $\onechartof{\astexp}$ of $\astexp$ with start vertex $\astexp$
  is the \generatedby{\setexp{\astexp}} \subonechart\ of the stacked star expressions \oneLTS~$\oneLTSof{\stackStExpover{\actions}}$. 
  It follows by the definition of induced transitions in Def.~\ref{def:indLTS} that
  the induced chart $\indLTSof{\onechartof{\astexp}}$ of $\onechartof{\astexp}$ is the \generatedby{\setexp{\astexp}} subchart
  of the induced \LTS~$\indLTSof{\oneLTSof{\stackStExpover{\actions}}}$ of $\oneLTSof{\stackStExpover{\actions}}$.
  On the other hand, the chart interpretation $\chartof{\astexp}$ of $\astexp$
  is the \generatedby{\setexp{\astexp}} subchart of the star expressions \LTS~$\LTSof{\StExpover{\actions}}$ by Def.~\ref{def:chart:interpretation}.
  Now since, due to Lemma~\ref{lem:StExpLTS:funbisim:StExpindLTS}, 
  the graph of the projection function $\sproj$ defines a bisimulation from 
  $\indLTSof{\oneLTSof{\stackStExpover{\actions}}}$ to the star expressions \LTS~$\LTSof{\StExpover{\actions}}$
  which due to $\proj{\astexp} = \astexp$ relates $\astexp$ with itself,
  it follows that the restriction of $\sproj$ to the \generatedby{\setexp{\astexp}} part
  $\indscchartof{\onechartof{\astexp}}$ of $\indLTSof{\oneLTSof{\StExpover{\actions}}}$
  defines a bisimulation to the \generatedby{\setexp{\astexp}} part $\chartof{\astexp}$ of $\LTSof{\StExpover{\actions}}$. 
  From this we conclude that $\indscchartof{\onechartof{\astexp}} \funbisim \chartof{\astexp}$ holds.
\end{proof}

\subsection{Proof of property~\ref{property:2} of $\onechartof{\cdot}$}
  \label{property:2::proofs}

To demonstrate property~\ref{property:2}
we first develop the proof of Lemma~\ref{lem:oneLTShat:stackStExps:is:LLEEw}
via three auxiliary statements 
  (Lemma~\ref{lem:steps:appcxt}, Lemma~\ref{lem:bodypaths:stexpstackprod:stexpprod}, and Lemma~\ref{lem:lem:oneLTShat:stackStExps:is:LLEEw} below), 
and then use Lemma~\ref{lem:oneLTShat:stackStExps:is:LLEEw} to show Theorem~\ref{thm:onechart-int:LLEEw}.

Lemma~\ref{lem:oneLTShat:stackStExps:is:LLEEw}
states that the \entrybodylabeling\ $\oneLTShatof{\stackStExpover{\actions}}$ from Def.~\ref{def:stackStExpTSShat} 
is a \LLEEwitness\ for the stacked star expression \oneLTS\ $\oneLTSof{\stackStExpover{\actions}}$.
As a preparation for verifying that the properties of a \LLEEwitness\ are fulfilled, 
we prove a number of technical statements in Lemma~\ref{lem:bodypaths:stexpstackprod:stexpprod}, and Lemma~\ref{lem:lem:oneLTShat:stackStExps:is:LLEEw} below. 
For their formulation, we first introduce the set $\AppCxtover{\actions}$ of \emph{applicative contexts of stacked star expressions} over $\actions$ 
by which we mean the set of contexts that are defined by the grammar:
\begin{center}
  $
  \acxtwh
    \;\;\;\BNFdefdby\;\;\;
      \Box
        \;\BNFor\;
      \stexpprod{\acxtwh}{\astexp}
        \;\BNFor\;
      \stexpstackprod{\acxtwh}{\stexpit{\astexp}}
        \qquad\text{(where $\astexp\in\StExpover{\actions}$)} \punc{.}   
  $
\end{center}
In view of Def.~\ref{def:StExp} every stacked star expression $\asstexp\in\stackStExpover{\actions}$ 
  can be parsed uniquely as of the form $\asstexp = \acxtap{\astexp}$ for some star expression~$\astexp\in\StExpover{\actions}$, 
  and an applicative context $\acxtwh\in\AppCxtover{\actions}$.
  
The following lemma states that transitions of the \entrybodylabeling~$\oneLTShatof{\stackStExpover{\actions}}$
are preserved under the operation of filling stacked star expressions into applicative contexts.

\begin{lem}\label{lem:steps:appcxt}
  If $\asstexp \redi{\alab} \asstexpacc$ in $\oneLTShatof{\stackStExpover{\actions}}$, 
  then also $\acxtap{\asstexp} \redi{\alab} \acxtap{\asstexpacc}$,
  for each $\darkcyan{\alab} \in \setexp{\bodylabcol} \cup \descsetexp{ \darkcyan{\loopnsteplab{\aLname}} }{ \aLname\in\natplus }$.
\end{lem}

\begin{proof}
  By induction on the structure of applicative contexts, using the rules for $\sstexpprod$ and $\sstexpstackprod$ of $\stackStExpTSShatover{\actions}$.
\end{proof}

\begin{lem}\label{lem:bodypaths:stexpstackprod:stexpprod}
  \begin{enumerate}[label={(\alph{*})},leftmargin=*,align=left,labelsep=0.25ex,itemsep=0ex]
    \item{}\label{it:1:lem:bodypaths:stexpstackprod:stexpprod}
      Every maximal $\sredi{\bodylab}$ path from $\acxtap{\stexpstackprod{\asstexp}{\stexpit{\astexp}}}$ in $\oneLTShatof{\stackStExpover{\actions}}$,
      where $\asstexp\in\stackStExpover{\actions}$, $\astexp\in\StExpover{\actions}$, and $\acxtwh\in\AppCxtover{\actions}$,
      is of either of the following two forms: 
      \begin{enumerate}[label={(\roman{*})},leftmargin=*,align=right,labelsep=1ex]
        \item{}\label{path:1:it:1:lem:bodypaths:stexpstackprod:stexpprod}
          $ \acxtap{\stexpstackprod{\asstexp}{\stexpit{\astexp}}}
              =
            \acxtap{\stexpstackprod{\asstexpi{0}}{\stexpit{\astexp}}}  
              \redi{\bodylab}
            \acxtap{\stexpstackprod{\asstexpi{1}}{\stexpit{\astexp}}}  
              \redi{\bodylab}
            \ldots
              \redi{\bodylab}
            \acxtap{\stexpstackprod{\asstexpi{n}}{\stexpit{\astexp}}}
              \redi{\bodylab}
            \ldots  
          $  
          is finite or infinite with  $\asstexpi{0},\asstexpi{1},\ldots,\asstexpi{n},\ldots\in\stackStExpover{\actions}$
          such that
          $\asstexpi{0}
             \redi{\bodylab}
           \asstexpi{1}
             \redi{\bodylab}
           \ldots
             \redi{\bodylab}
           \asstexpi{n}
             \redi{\bodylab}
           \ldots  
          \,$, 
          
        \item{}\label{path:2:it:1:lem:bodypaths:stexpstackprod:stexpprod}
          $ \acxtap{\stexpstackprod{\asstexp}{\stexpit{\astexp}}}
              =
            \acxtap{\stexpstackprod{\asstexpi{0}}{\stexpit{\astexp}}}  
              \redi{\bodylab}
            \acxtap{\stexpstackprod{\asstexpi{1}}{\stexpit{\astexp}}}  
              \redi{\bodylab}
            \ldots
              \redi{\bodylab}
            \acxtap{\stexpstackprod{\asstexpi{n}}{\stexpit{\astexp}}}
              \redi{\bodylab}
            \acxtap{\stexpit{\astexp}}
              \redi{\bodylab}  
            \ldots  
          $
          for $n\in\nat$, $\asstexpi{0},\asstexpi{1},\ldots,\asstexpi{n}\in\stackStExpover{\actions}$     
          with 
          $\asstexpi{0}
             \redi{\bodylab}
           \asstexpi{1}
             \redi{\bodylab}
           \ldots
             \redi{\bodylab}
           \asstexpi{n}
          $,
          $\terminates{\asstexpi{n}}$, and $\stexpstackprod{\asstexpi{n}}{\stexpit{\astexp}} \redi{\bodylab} \stexpit{\astexp}$. 
      \end{enumerate}
  
    \item{}\label{it:2:lem:bodypaths:stexpstackprod:stexpprod}
      Every maximal $\sredi{\bodylab}$ path from $\stexpprod{\asstexp}{\astexp}$ in $\oneLTShatof{\stackStExpover{\actions}}$, 
      for $\asstexp\in\stackStExpover{\actions}$, and $\astexp\in\StExpover{\actions}$, 
      where now no filling in an applicative context is permitted, 
      is of either of the following two forms: 
      \begin{enumerate}[label={(\roman{*})},leftmargin=*,align=right,labelsep=1ex]
        \item{}\label{path:1:it:2:lem:bodypaths:stexpstackprod:stexpprod}
          $ {\stexpprod{\asstexp}{\astexp}}
              =
            {\stexpprod{\asstexpi{0}}{\astexp}}  
              \redi{\bodylab}
            {\stexpprod{\asstexpi{1}}{\astexp}}  
              \redi{\bodylab}
            \ldots
              \redi{\bodylab}
            {\stexpprod{\asstexpi{n}}{\astexp}}
              \redi{\bodylab}
            \ldots  
          $  
          is finite or infinite
          with stacked star expressions 
              $\asstexpi{0},\asstexpi{1},\ldots,\asstexpi{n},\ldots\in\stackStExpover{\actions}$
          such that
          $\asstexpi{0}
             \redi{\bodylab}
           \asstexpi{1}
             \redi{\bodylab}
           \ldots
             \redi{\bodylab}
           \asstexpi{n}
             \redi{\bodylab}
           \ldots  
          \;$, 
        \item{}\label{path:2:it:2:lem:bodypaths:stexpstackprod:stexpprod}
          $ {\stexpprod{\asstexp}{\astexp}}
              =
            {\stexpprod{\asstexpi{0}}{\astexp}}  
              \redi{\bodylab}
            {\stexpprod{\asstexpi{1}}{\astexp}}  
              \redi{\bodylab}
            \ldots
              \redi{\bodylab}
            {\stexpprod{\asstexpi{n}}{\astexp}}
              \redi{\bodylab}
            {\asstexpacc}
              \redi{\bodylab}  
            \ldots  
          $
          with $n\in\nat$, and stacked star expressions $\asstexpi{0},\asstexpi{1},\ldots,\asstexpi{n},\bsstexp\in\stackStExpover{\actions}$     
          such that
          $\asstexpi{0}
             \redi{\bodylab}
           \asstexpi{1}
             \redi{\bodylab}
           \ldots
             \redi{\bodylab}
           \asstexpi{n}
          $,
          $\terminates{\asstexpi{n}}$ 
          and $\astexp \redi{\alab} \asstexpacc$.
      \end{enumerate}
      
  \end{enumerate}      
\end{lem}

\begin{proof}
  We first argue for statement \ref{it:2:lem:bodypaths:stexpstackprod:stexpprod}. 
  Due to the rules for $\sstexpprod$ in $\stackStExpTSShat$
  it holds for every step $\stexpprod{\asstexp}{\astexp} \redi{\bodylab} \csstexp$   
  either (1) $\asstexp \redi{\bodylab} \asstexpi{1}$, and $\csstexp = \stexpprod{\asstexpi{1}}{\astexp}$, for some $\asstexpi{1}$,
  or     (2) $\terminates{\asstexp}$, and $\astexp \redi{\alab} \bsstexp$, and $\csstexp = \bsstexp$,  for some $\bsstexp$ and a marking label~$\alab$.
  In case (2) we have recognized any maximal path $\stexpprod{\asstexp}{\astexp} \redi{\bodylab} \csstexp \redi{\bodylab} \ldots$ 
  as of the form~\ref{path:2:it:2:lem:bodypaths:stexpstackprod:stexpprod} of item \ref{it:2:lem:bodypaths:stexpstackprod:stexpprod} for $n=0$.
  In case (1) such a maximal path is parsed as 
    $\stexpprod{\asstexp}{\astexp} \redi{\bodylab} \stexpprod{\asstexpi{1}}{\astexp} = \csstexpi{1} \redi{\bodylab} \ldots$.
  Then we can use the same argument again for the second step $\stexpprod{\asstexpi{1}}{\astexp} \redi{\bodylab} \csstexpi{2}$ of that path.
  So by parsing a (finite or infinite) maximal path
  $\stexpprod{\asstexp}{\astexp} \redi{\bodylab} \csstexpi{1} \redi{\bodylab} \ldots \redi{\bodylab} \csstexpi{n} \redi{\bodylab} \ldots $ 
  with this argument over its steps from left to right,
  we either discover at some finite stage that one of its steps is of case~(2), 
  and then can conclude that the path is of form~\ref{path:2:it:2:lem:bodypaths:stexpstackprod:stexpprod},
  or we only encounter steps of case~(1) and thereby parse the path as of~form~\ref{path:1:it:2:lem:bodypaths:stexpstackprod:stexpprod}.
  
  We turn to statement \ref{it:1:lem:bodypaths:stexpstackprod:stexpprod}. 
  We observe that any first step 
  $ \acxtap{\stexpstackprod{\asstexp}{\stexpit{\astexp}}} \redi{\bodylab} \csstexpi{1}$ 
  can, in view the TSS rules 
                             for $\sstexpprod$ and $\sstexpstackprod$
       and since $\stexpstackprod{\asstexp}{\stexpit{\astexp}}$ does not terminate immediately,
  only arise from a step $\stexpstackprod{\asstexp}{\stexpit{\astexp}} \redi{\bodylab} \bsstexp$
  via context filling as $\acxtap{\stexpstackprod{\asstexp}{\stexpit{\astexp}}} \redi{\bodylab} \acxtap{\bsstexp} = \csstexpi{1}$.
  Therefore it suffices to show the statement for \ref{it:2:lem:bodypaths:stexpstackprod:stexpprod} for the empty context $\acxt = \Box$.
  For this restricted statement we can argue analogously as for item~\ref{it:1:lem:bodypaths:stexpstackprod:stexpprod}. 
  Here we use that due to the TSS rules for $\sstexpstackprod$ 
  it holds for every step $\stexpstackprod{\asstexp}{\stexpit{\astexp}} \redi{\bodylab} \csstexp$ that
  either (1) $\asstexp \redi{\bodylab} \asstexpi{1}$, and $\csstexp = \stexpstackprod{\asstexpi{1}}{\astexp}$, for some $\asstexpi{1}$,
  or     (2) $\terminates{\asstexp}$, and $\stexpstackprod{\asstexp}{\stexpit{\astexp}} \redi{\bodylab} \stexpit{\astexp}$, and $\csstexp = \stexpit{\astexp}$.
  By parsing the steps of a maximal $\sredi{\bodylab}$ path 
  $\stexpstackprod{\asstexp}{\stexpit{\astexp}} \redi{\bodylab} \csstexpi{1} \redi{\bodylab} \ldots \redi{\bodylab} \csstexpi{n} \redi{\bodylab} \ldots $ 
  from left to right we either encounter a step of case (2), and then find that the path is of form~\ref{path:2:it:1:lem:bodypaths:stexpstackprod:stexpprod},
  or otherwise we find that the path is of form~\ref{path:1:it:1:lem:bodypaths:stexpstackprod:stexpprod} in statement~\ref{it:1:lem:bodypaths:stexpstackprod:stexpprod}.
\end{proof}

\begin{lem}\label{lem:lem:oneLTShat:stackStExps:is:LLEEw}
  The following statements hold for paths of transitions in $\oneLTShatof{\StExpover{\actions}}$:
  \begin{enumerate}[label={(\roman{*})},align=right,leftmargin=*,itemsep=0ex]
    \item{}\label{it:1:lem:lem:oneLTShat:stackStExps:is:LLEEw}
      There are no infinite $\sredi{\bodylab}$ paths in $\oneLTShatof{\StExpover{\actions}}$.
    
    \item{}\label{it:2:lem:lem:oneLTShat:stackStExps:is:LLEEw}
      If $\asstexp\in\stackStExpover{\actions}$ is normed,
      then $\asstexp \redrtci{\bodylab} \bstexp$ for some $\bstexp\in\StExpover{\actions}$ with $\terminates{\bstexp}$. 
      
    \item{}\label{it:3:lem:lem:oneLTShat:stackStExps:is:LLEEw}
      If $\asstexp \redi{\looplab{n}} \asstexpacc$ with $n>0$ and $\asstexp,\asstexpacc\in\stackStExpover{\actions}$,
      then
      $\asstexp = \acxtap{\stexpit{\astexp}}$, 
      $\astexp$ \txtnormedplus, 
      $\astexp \redi{\alab} \asstexpacci{0}$,
      $\asstexpacc = \acxtap{ \stexpstackprod{\asstexpacci{0}}{\stexpit{\astexp}} }$,
      and $n = \sth{\astexp} + 1$,
      for some $\astexp\in\StExpover{\actions}$, $\acxtwh\in\AppCxtover{\actions}$, and $\astexpacci{0}\in\stackStExpover{\actions}$. 
      
      
    \item{}\label{it:4:lem:lem:oneLTShat:stackStExps:is:LLEEw}
      Neither $\sredi{\bodylab}$ and $\sredi{\looplab{n}}$ steps in $\oneLTShatof{\StExpover{\actions}}$,
      where $n\ge 1$, increase the star height of expressions.

\end{enumerate}
\end{lem}

\begin{proof}
  For statement~\ref{it:1:lem:lem:oneLTShat:stackStExps:is:LLEEw} 
  we prove that there is no infinite $\sredi{\bodylab}$ path from any $\bsstexp\in\stackStExpover{\actions}$,
  by structural induction on $\bsstexp\in\stackStExpover{\actions}$.
  For $\bsstexp = \stexpzero$ and $\bsstexp = \stexpone$ there are no steps possible at all, 
  because there\vspace*{-2pt} is no rule with conclusion $\bsstexp \lti{\aoneact}{\alab} \bsstexpacc$ in the TSS~$\stackStExpTSShatover{\actions}$. 
  For $\bsstexp = \aact$, where $\aact\in\actions$, the single step possible is $\bsstexp \lti{\aact}{\bodylab} \stexpone$,
  which cannot be extended by any further steps. 
  For $\bsstexp = \stexpsum{\astexpi{1}}{\astexpi{2}}$ it holds, due to the TSS rule for $\sstexpsum$, 
  that every $\sredi{\bodylab}$ path $\bsstexp \redi{\bodylab} \bsstexpacc \redi{\bodylab} \ldots$ from $\bsstexp$
  gives rise to also a path  $\astexpi{i} \redi{\bodylab} \bsstexpacc \redi{\bodylab} \ldots$ from $\astexpi{i}$, for $i\in\setexp{1,2}$. 
  Since by induction hypothesis the $\sredi{\bodylab}$ path from $\astexpi{i}$ cannot be infinite, this follows also for the $\sredi{\bodylab}$ path from $\bsstexp$.
  Now we let $\bsstexp = \stexpprod{\asstexp}{\astexp}$ for $\asstexp\in\stackStExpover{\actions}$ and $\astexp\in\StExpover{\actions}$. 
  Then every maximal $\sredi{\bodylab}$ path from $\bsstexp$ 
  is of the form \ref{path:1:it:2:lem:bodypaths:stexpstackprod:stexpprod} or of the form \ref{path:2:it:2:lem:bodypaths:stexpstackprod:stexpprod}
  in by Lemma~\ref{lem:bodypaths:stexpstackprod:stexpprod}, \ref{it:2:lem:bodypaths:stexpstackprod:stexpprod}.
  So it either has the length of a $\sredi{\bodylab}$ path from $\asstexp$,
  or $1$ plus the finite length of a $\sredi{\bodylab}$ path from $\asstexp$, plus the length of a $\sredi{\bodylab}$ path from $\astexp$.
  Since by induction hypothesis every $\sredi{\bodylab}$ path from $\asstexp$ or from $\astexp$ is finite, 
  it follows that every $\sredi{\bodylab}$ path from $\bsstexp$ must have finite length as well.
  For the remaining case $\bsstexp = \stexpprod{\asstexp}{\stexpit{\astexp}}$, for $\asstexp\in\stackStExpover{\actions}$ and $\astexp\in\StExpover{\actions}$,
  we can argue analogously by using Lemma~\ref{lem:bodypaths:stexpstackprod:stexpprod}, \ref{it:1:lem:bodypaths:stexpstackprod:stexpprod}. 
  This concludes the proof 
                           of item~\ref{it:1:lem:lem:oneLTShat:stackStExps:is:LLEEw}.
  
  Item~\ref{it:2:lem:lem:oneLTShat:stackStExps:is:LLEEw} can be shown by induction on the structure of $\asstexp\in\stackStExpover{\actions}$. 
  We only consider the \nontrivial\ case in which $\asstexp = \stexpstackprod{\asstexpi{0}}{\stexpit{\astexp}}$, for $\asstexp$ normed.
  We have to show that there is $\bstexp\in\StExpover{\actions}$ with $\asstexp \redrtci{\bodylab} \bstexp$ and $\terminates{\bstexp}$.
  We first observe that unless $\terminates{\asstexpi{0}}$ holds, 
  any step from $\asstexp = \stexpstackprod{\asstexpi{0}}{\stexpit{\astexp}}$ must, due to the TSS rules, be of the form
  $\stexpstackprod{\asstexpi{0}}{\stexpit{\astexp}} \redi{\alab} \stexpstackprod{\asstexpacci{0}}{\stexpit{\astexp}}$ 
  for some $\asstexpacci{0}\in\stackStExpover{\actions}$ such that also $\asstexpi{0} \redi{\alab} \asstexpacci{0}$.
  Therefore the assumption that $\asstexp$ is normed implies that also $\asstexpi{0}$ is normed.
  Then we can apply the induction hypothesis for $\asstexpi{0}$
  in order to obtain a path $\asstexpi{0} \redrtci{\bodylab} \astexpacci{0}$ and $\terminates{\astexpacci{0}}$ for some $\astexpacci{0}\in\StExpover{\actions}$.
  From this we get
  $ \stexpstackprod{\asstexpi{0}}{\stexpit{\astexp}}
      \redrtci{\bodylab}
    \stexpstackprod{\astexpacci{0}}{\stexpit{\astexp}} $
  by Lemma~\ref{lem:steps:appcxt},
  and 
  $\stexpstackprod{\astexpacci{0}}{\stexpit{\astexp}} \redi{\bodylab} \stexpit{\astexp}$
  due to $\terminates{\astexpacci{0}}$. 
  Consequently we find 
  $ \asstexp
      =
    \stexpstackprod{\asstexpi{0}}{\stexpit{\astexp}}
      \redrtci{\bodylab}
    \stexpit{\astexp}$ and $\terminates{\stexpit{\astexp}}$,
  and hence by letting $\bstexp \defdby \stexpit{\astexp}\in\StExpover{\actions}$ 
  we have obtained the proof obligation 
  $ \asstexp
      \redrtci{\bodylab}
    \bstexp$ and $\terminates{\bstexp}$
  in this case.  
 
  Statement~\ref{it:3:lem:lem:oneLTShat:stackStExps:is:LLEEw} can be proved by induction
  on the depth of derivations in the TSS~$\stackStExpTSShat$.
  This proof employs two crucial facts. 
  First, every \entrystep\ is created by a step $\stexpit{\astexp} \redi{\looplab{n}} \stexpprod{\asstexpacc}{\stexpit{\astexp}}$
  with $\astexp \redi{\alab} \asstexpacc$ and $n = \sth{\stexpit{\astexp}} = \sth{\astexp} + 1$, for some $\astexp\in\StExpover{\actions}$, $\asstexpacc\in\stackStExpover{\actions}$,
  and marking label $\alab$.
  And second, \entrysteps\ are preserved under 
  contexts $\stexpprod{\Box}{\astexp}$ and $\stexpstackprod{\Box}{\stexpit{\astexp}}$,
  and hence are preserved under applicative contexts.
  
  Statement~\ref{it:4:lem:lem:oneLTShat:stackStExps:is:LLEEw} can be proved by a straightforward induction
  on the depth of derivations in the TSS~$\stackStExpTSShat$\hspace*{-3pt}.
\end{proof}

On the basis of the statements of Lemma~\ref{lem:steps:appcxt}, Lemma~\ref{lem:bodypaths:stexpstackprod:stexpprod},  and Lemma~\ref{lem:lem:oneLTShat:stackStExps:is:LLEEw}
we can now demonstrate Lemma~\ref{lem:oneLTShat:stackStExps:is:LLEEw}. We first repeat its formulation.

\begin{repeatedlem}[= Lemma~\ref{lem:oneLTShat:stackStExps:is:LLEEw}]
  $\oneLTShatof{\stackStExpover{\actions}}$ is an \entrybodylabeling\ of $\oneLTSof{\stackStExpover{\actions}}$, 
  which is also a \LLEEwitness\ of $\oneLTSof{\stackStExpover{\actions}}$. 
\end{repeatedlem}

\begin{proof}[Proof (of Lemma~\ref{lem:oneLTShat:stackStExps:is:LLEEw})]
  That $\oneLTShatof{\stackStExpover{\actions}}$ is an \entrybodylabeling\ of $\oneLTSof{\stackStExpover{\actions}}$\vspace*{-2.5pt}
  is a consequence of the fact that the rules of the TSS~$\stackStExpTSShatover{\actions}$ are marking labeled versions
  of the rules of the TSS~$\stackStExpTSSover{\actions}$.\vspace*{-0.5pt}
  
  It remains to show that $\aoneLTShat \defdby \oneLTShatof{\stackStExpover{\actions}}$ is a \LLEEwitness\ of the stacked star expressions 
  \oneLTS~$\aLTS = \LTSof{\stackStExpover{\actions}}$.
  Instead of verifying the \LLEEwitness\ conditions~\ref{LLEEw:1}, \ref{LLEEw:2}, \ref{LLEEw:3},
  we establish the following four equivalent conditions
  that are understood to be universally quantified 
  over all $\asstexp,\asstexpi{1},\bsstexp,\bsstexpi{1}\in\stackStExpover{\actions}$, and $n,m\in\nat$ with $n,m>0$:
  \begin{enumerate}[label={\mbox{\rm (LLEE\hspace{0.7pt}-\hspace{0.7pt}\arabic*)}},leftmargin=*,align=left,itemsep=-0.4ex]\vspace{-0.1ex}
    \item{}\label{LLEEw:alt:1}%
      $ \asstexp 
          \redi{\looplab{n}} 
        \asstexpi{1} 
          \;\Longrightarrow\;
        \asstexp \comprewrels{\sredi{\looplab{n}}}{\sredrtci{\bodylab}} \asstexp$,
        
    \item{}\label{LLEEw:alt:2}%
      $\sredi{\bodylab}$ is terminating from $\asstexp$,
      
    \item{}\label{LLEEw:alt:3}%
      $\asstexp
         \comprewrels{\sredtavoidsvi{\asstexp}{\looplab{n}}}{\sredtavoidsvrtci{\asstexp}{\bodylab}} 
       \bsstexp
           \;\Longrightarrow\;
       \notterminates{\bsstexp} $ \mbox{}
      (the premise means
       that $\bsstexp$ is in $\indsubchartinat{\aoneLTShat}{\asstexp,n}$ such that $\bsstexp\neq\asstexp$),
      
    \item{}\label{LLEEw:alt:4}
      $\asstexp
         \comprewrels{\sredtavoidsvi{\asstexp}{\looplab{n}}}{\sredtavoidsvrtci{\asstexp}{\bodylab}} 
       \bsstexp
         \redi{\looplab{m}}
       \bsstexpi{1}  
           \;\Longrightarrow\;
       n > m $,  
  \end{enumerate}\vspace{-0.1ex}    
  where $\sredtavoidsvi{\asstexp}{\looplab{n}}$ and $\sredtavoidsvi{\asstexp}{\bodylab}$ 
        mean $\sredi{\looplab{n}}$ and $\sredi{\bodylab}$ steps,
        respectively, that avoid $\asstexp$ as their \underline{t}argets.
  Hereby \ref{LLEEw:alt:2} obviously implies \ref{LLEEw:1}.
  For each entry identifier $\pair{\asstexp}{n}\in\entriesof{\aoneLTShat}$ it is not difficult to check that 
  \ref{LLEEw:alt:1}, \ref{LLEEw:alt:2}, and \ref{LLEEw:alt:3}
  imply that $\indsubchartinat{\aoneLTShat}{\asstexp,n}$ satisfies the loop properties \ref{loop:1}, \ref{loop:2}, and \ref{loop:3}, respectively,
  to obtain \ref{LLEEw:2}. 
  Finally, \ref{LLEEw:alt:4} is an easy reformulation of the condition \ref{LLEEw:3}.  
        
  Now \ref{LLEEw:alt:2} is guaranteed by Lemma~\ref{lem:lem:oneLTShat:stackStExps:is:LLEEw}, \ref{it:1:lem:lem:oneLTShat:stackStExps:is:LLEEw}.
  It remains to verify the three remaining conditions from above.\vspace*{-3pt}
  For reasoning about transitions in $\aoneLTShat$ we employ easy properties that follow from the rules of 
                                                                                                           $\stackStExpTSShatover{\actions}$.
  
  For showing \ref{LLEEw:alt:1}, we suppose that $\asstexp \redi{\looplab{n}} \asstexpi{1}$ for some $\asstexp,\asstexpi{1}\in\stackStExpover{\actions}$.
  We have to show $\asstexp \comprewrels{\sredi{\looplab{n}}}{\sredrtci{\bodylab}} \asstexp$.
  By Lemma~\ref{lem:lem:oneLTShat:stackStExps:is:LLEEw}, \ref{it:3:lem:lem:oneLTShat:stackStExps:is:LLEEw}, 
  $\asstexp \redi{\looplab{n}} \asstexpi{1}$
  implies that $\asstexp = \acxtap{\stexpit{\astexp}}$, $\astexp$ \txtnormedplus, $n = \sth{\astexp} + 1$,
    for some $\astexp\in\StExpover{\actions}$ and $\acxtwh\in\AppCxtover{\actions}$.
  Since $\astexp$ is \txtnormedplus, and because no \onetransitions\ can emanate from $\astexp$,
  there must be some normed $\csstexp\in\stackStExpover{\actions}$, and $\aact\in\actions$
  such that $\astexp \lti{\aact}{\alab} \csstexp$ holds. We pick $\csstexp$ accordingly.
  Together with  $n = \sth{\astexp} + 1$ this implies $\stexpit{\astexp} \redi{\looplab{n}} \stexpstackprod{\csstexp}{\stexpit{\astexp}}$.
  Since $\csstexp$ is normed, 
  by Lemma~\ref{lem:lem:oneLTShat:stackStExps:is:LLEEw}, \ref{it:2:lem:lem:oneLTShat:stackStExps:is:LLEEw}
  there are $k\in\nat$, $\csstexpi{0},\ldots,\csstexpi{k}\in\stackStExpover{\actions}$ such that
  $\csstexp = \csstexpi{0} \redi{\bodylab} \ldots \redi{\bodylab} \csstexpi{k}$ and $\terminates{\csstexpi{k}}$.
  Then 
  $ \stexpstackprod{\csstexpi{k}}{\stexpit{\astexp}}
      \redi{\bodylab}
    \stexpit{\astexp}$,
  and, by Lemma~\ref{lem:steps:appcxt},   
  $\asstexp = \acxtap{\stexpit{\astexp}}
     \redi{\looplab{n}}
   \acxtap{\stexpstackprod{\csstexpi{0}}{\stexpit{\astexp}}}
     \redi{\bodylab}
   \ldots
     \redi{\bodylab}
   \acxtap{\stexpstackprod{\csstexpi{k}}{\stexpit{\astexp}}}
     \redi{\bodylab}    
   \acxtap{\stexpit{\astexp}}
     =
   \asstexp$.  
  Hence $\asstexp \comprewrels{\sredi{\looplab{n}}}{\sredrtci{\bodylab}} \asstexp$.  
  
  For showing \ref{LLEEw:alt:3},
  we suppose that 
  $ \asstexp 
      \redi{\looplab{n}}
    \asstexpi{1}
      \redi{\bodylab}
    \ldots
      \redi{\bodylab}
    \asstexpi{k}
      =
    \bsstexp $
  for $\asstexp,\asstexpi{1},\ldots,\asstexpi{k},\bsstexp\in\stackStExpover{\actions}$ 
  with $\asstexpi{1},\ldots,\asstexpi{k},\bsstexp \neq \asstexp$, and $k\in\nat$, $k \ge 1$.
  We have to show that $\notterminates{\bsstexp}$. 
  By applying Lemma~\ref{lem:lem:oneLTShat:stackStExps:is:LLEEw}, \ref{it:3:lem:lem:oneLTShat:stackStExps:is:LLEEw}, 
  $\asstexp \redi{\looplab{n}} \asstexpi{1}$
  implies that 
  $\asstexp = \acxtap{\stexpit{\astexp}}$, 
  $\asstexpi{1} = \acxtap{\stexpstackprod{\csstexp}{\stexpit{\astexp}}}$,
  $\stexpit{\astexp} \redi{\looplab{n}} \stexpstackprod{\csstexp}{\stexpit{\astexp}}$,
  for some $\astexp\in\StExpover{\actions}$, $\csstexp\in\stackStExpover{\actions}$, and $\acxtwh\in\AppCxtover{\actions}$.    
  As now the considered path is of the form
  $ \asstexp = \acxtap{\stexpit{\astexp}} 
      \redi{\looplab{n}}
    \asstexpi{1} = \acxtap{\stexpstackprod{\csstexp}{\stexpit{\astexp}}}
      \redi{\bodylab}
    \asstexpi{2}
      \redi{\bodylab}  
    \ldots
      \redi{\bodylab}
    \asstexpi{k}
      =
    \bsstexp $, 
  it follows from Lemma~\ref{lem:bodypaths:stexpstackprod:stexpprod}, \ref{it:1:lem:bodypaths:stexpstackprod:stexpprod},
  due to $\asstexpi{1},\ldots,\asstexpi{k} \neq \asstexp = \acxtap{\stexpit{\astexp}}$
  that this path must be of form \ref{path:1:it:1:lem:bodypaths:stexpstackprod:stexpprod} there:
  $ \asstexp = \acxtap{\stexpit{\astexp}} 
      \redi{\looplab{n}}
    \asstexpi{1} = \acxtap{\stexpstackprod{\csstexpi{0}}{\stexpit{\astexp}}}
      \redi{\bodylab} 
    \asstexpi{2} = \acxtap{\stexpstackprod{\csstexpi{1}}{\stexpit{\astexp}}}
      \redi{\bodylab}
    \ldots
      \redi{\bodylab}
    \asstexpi{k} = \acxtap{\stexpstackprod{\csstexpi{k}}{\stexpit{\astexp}}}
      =
    \bsstexp $
  for some $\csstexpi{0},\csstexpi{1},\ldots,\csstexpi{k}\in\stackStExpover{\actions}$ with $\csstexpi{0} = \csstexp$.
  From the form $\acxtap{\stexpstackprod{\csstexpi{k}}{\stexpit{\astexp}}}$ of $\bsstexp$
  we can now conclude that $\notterminates{\bsstexp}$, 
  since a stacked star expression that is not also a star expression does not permit immediate termination.

  For showing \ref{LLEEw:alt:4},
  we suppose that 
  $ \asstexp 
      \redi{\looplab{n}}
    \asstexpi{1}
      \redi{\bodylab}
    \ldots
      \redi{\bodylab}
    \asstexpi{k}
      =
    \bsstexp 
      \redi{\looplab{m}}
    \bsstexpi{1}
    $
  for some $k,m,n\in\nat$, $k,m,n \ge 1$,   
  and $\asstexp,\asstexpi{1},\ldots,\asstexpi{k},\bsstexp,\bsstexpi{1}\in\stackStExpover{\actions}$ 
  with $\asstexpi{1},\ldots,\asstexpi{k} \neq \asstexp$.
  We have to show that $n > m$.
  Arguing analogously as for \ref{LLEEw:alt:3}
  with Lemma~\ref{lem:lem:oneLTShat:stackStExps:is:LLEEw}, \ref{it:3:lem:lem:oneLTShat:stackStExps:is:LLEEw},
  and Lemma~\ref{lem:bodypaths:stexpstackprod:stexpprod}, \ref{it:1:lem:bodypaths:stexpstackprod:stexpprod}, 
  we find that the path is of the form 
  $ \asstexp = \acxtap{\stexpit{\astexp}} 
      \redi{\looplab{n}}
    \asstexpi{1} = \acxtap{\stexpstackprod{\csstexpi{0}}{\stexpit{\astexp}}}
      \redi{\bodylab} 
    \asstexpi{2} = \acxtap{\stexpstackprod{\csstexpi{1}}{\stexpit{\astexp}}}
      \redi{\bodylab}
    \ldots
      \redi{\bodylab}
    \asstexpi{k} = \acxtap{\stexpstackprod{\csstexpi{k}}{\stexpit{\astexp}}}
      =
    \bsstexp
      \redi{\looplab{m}}
    \bsstexpi{1} $
  for some $\csstexpi{0},\csstexpi{1},\ldots,\csstexpi{k}\in\stackStExpover{\actions}$
  such that 
  $n = \sth{\astexp} + 1$,
  $\astexp \redi{\alab} \csstexpi{0}$ 
  and 
  $ \csstexpi{0}
      \redi{\bodylab} 
    \csstexpi{1}
      \redi{\bodylab}
    \ldots
      \redi{\bodylab}
    \csstexpi{k-1}
      =
    \bsstexp
      \redi{\looplab{m}}
    \bsstexpi{1} $.
   Due to Lemma~\ref{lem:lem:oneLTShat:stackStExps:is:LLEEw}, \ref{it:4:lem:lem:oneLTShat:stackStExps:is:LLEEw},
   it follows that 
   $n = \sth{\astexp} + 1
      > \sth{\astexp}
      \ge \sth{\csstexpi{0}}
      \ge \sth{\csstexpi{1}}
      \ge \ldots
      \ge \sth{\csstexpi{k}} $.
   Since due to the loop step
   $\acxtap{\stexpstackprod{\csstexpi{k}}{\stexpit{\astexp}}}
      \redi{\looplab{m}} 
    \bsstexpi{1}$ 
   can only originate from $\csstexpi{k}$,
   we conclude by Lemma~\ref{lem:lem:oneLTShat:stackStExps:is:LLEEw}, \ref{it:3:lem:lem:oneLTShat:stackStExps:is:LLEEw}, 
   that $\sth{\csstexpi{k}} \ge m$. 
   From what we have obtained before, we conclude~$n > m$.

   By having verified \ref{LLEEw:alt:1}--\ref{LLEEw:alt:4} we have shown that $\oneLTShatof{\StExpover{\actions}}$ is a \LLEEwitness.
\end{proof}

Now we can finally prove Theorem~\ref{thm:onechart-int:LLEEw}. We repeat its statement first.

\begin{repeatedthm}[= Theorem~\ref{thm:onechart-int:LLEEw}]
  For every $\astexp\in\StExpover{\actions}$,
  the \entrybodylabeling~$\onecharthatof{\astexp}$ of $\onechartof{\astexp}$ is a \LLEEwitness\ of $\onechartof{\astexp}$.
  Therefore the \onechart\ interpretation $\onechartof{\astexp}$ of a star expression $\astexp\in\StExpover{\actions}$ satisfies the property~\LEE. 
\end{repeatedthm}

\begin{proof}[Proof (of Theorem~\ref{thm:onechart-int:LLEEw})]
  Since $\oneLTShatof{\stackStExpover{\actions}}$ is an \entrybodylabeling\ of $\oneLTSof{\stackStExpover{\actions}}$ by Lemma~\ref{lem:oneLTShat:stackStExps:is:LLEEw},
  it follows that $\onecharthatof{\astexp}$, the $\astexp$-rooted \subonechart\ of $\oneLTShatof{\stackStExpover{\actions}}$,
  is an \entrybodylabeling\ of $\onechartof{\astexp}$, the $\astexp$-rooted subchart of the \LTS~$\oneLTSof{\stackStExpover{\actions}}$.
  Now \LLEEwitnesses\ are preserved under taking generated subcharts, which can be concluded directly 
  from the alternative characterization of \LLEEwitnesses\ \ref{LLEEw:alt:1}--\ref{LLEEw:alt:4} in the proof of Lemma~\ref{lem:oneLTShat:stackStExps:is:LLEEw}.
  It follows that $\onecharthatof{\astexp}$ is a \LLEEwitness\ of $\onechartof{\astexp}$. 
\end{proof}


\paragraph{Acknowledgement.} \hspace*{-0.4em}%
  I want to thank the reviewers of the workshop and the post-proceedings submission for their close reading, their questions, 
  and for their suggestions of improvements.
  I am thankful to Luca Aceto for his comments on a draft.
  For the idea to define LLEE-witnesses directly from transition system specifications
  I want to thank Wan Fokkink, with whom I worked it out for \onefree\ star expressions in \cite{grab:fokk:2020:LICS,grab:fokk:2020:arxiv}. 
  That inspired me for this generalization to the full class of star expressions.

\bibliographystyle{eptcs}
\bibliography{scpgs-arxiv.bib}

\end{document}